\documentclass[12pt]{article}
\usepackage{amsfonts}
\usepackage{eurosym}
\usepackage{float}
\usepackage{makeidx}
\usepackage{graphicx}
\usepackage{amsmath}
\usepackage{ifpdf}
\usepackage{epstopdf}
\usepackage{authblk}
\usepackage[dcucite]{harvard}

\newtheorem{Theo}{Theorem}

\newtheorem{Corol}{Corollary}
\newtheorem{Def}{Definition}
\newtheorem{Le}{Lemma}

\newtheorem{Rem}{Remark}

\input{tcilatex}
\begin{document}

\date{}
\title{\textbf{An alternative paradigm of fault diagnosis in dynamic
systems: orthogonal projection-based methods }}
\author[1]{\small Steven X. Ding}
\author[2]{\small Linlin Li}
\author[1]{\small Tianyu Liu}
\affil[1]{Institute for Automatic Control and Complex Systems, University of Duisburg-Essen, 47057, Duisburg, Germany}
\affil[2]{School of Automation and Electrical Engineering, University of Science and Technology Beijing, Beijing 100083, P. R. China (Corresponding author)}
\maketitle
%
%
%


\bigskip

\textbf{Abstract}: In this paper, we propose a new paradigm of fault
diagnosis in dynamic systems as an alternative to the well-established
observer-based framework. The basic idea behind this work is to (i)
formulate fault detection and isolation as projection of measurement signals
onto (system) subspaces in Hilbert space, and (ii) solve the resulting
problems by means of projection methods with orthogonal projection operators
and gap metric as major tools. In the new framework, fault diagnosis issues
are uniformly addressed both in the model-based and data-driven fashions.
Moreover, the design and implementation of the projection-based fault
diagnosis systems, from residual generation to threshold setting, can be
unifiedly handled. Thanks to the well-defined distance metric for
projections in Hilbert subspaces, the projection-based fault diagnosis
systems deliver optimal fault detectability. In particular, a new type of
residual-driven thresholds is proposed, which significantly increases the
fault detectability. In this work, various design schemes are proposed,
including a basic projection-based fault detection scheme, fault detection
schemes for feedback control systems, fault classification as well as two
modified fault detection schemes. As a part of our study, relations to the
existing observer-based fault detection systems are investigated, which
showcases that, with comparable online computations, the proposed
projection-based detection methods offer improved detection performance.

\bigskip


\section{Introduction}

Observer-based fault diagnosis is the state of the art technique in dealing
with fault detection, isolation and identification in dynamic systems \cite%
{PFC89,Gertler98,CP99,Blanke06,Ding2008}. Beginning 50 years ago with the
pioneer work \cite{Beard,Jones}, observer-based fault diagnosis technique
has undergone a rapid development over a couple of decades and become today
well established as an active research area in control theory and
engineering \cite%
{Frank90,DingJPC97,VRK03-I,MDES2009,HKKS2010survey,GCD-2015survey,ZXD_2018}.
Even in the era that data-driven and machine learning methods become the
most dominant research domain in technical fault diagnosis, observer-based
fault diagnosis technique is still the major and efficient tool applied to
addressing fault diagnosis issues for dynamic and particularly automatic
control systems \cite{Ding2020}. The recent development of data-driven
design of observer-based fault detection systems \cite%
{Ding2014,Ding_IJP_2014} has extended the application of this technique to
meet practical demands.

\bigskip

An observer-based fault diagnosis system consists of two major functional
blocks: residual (feature) generation and decision making. While decisions
for fault detection or isolation are made on the basis of residual
processing algorithms, the major task of residual generation is the
construction of an observer serving as residual generator. Thanks to this
intimate relation to control theory, most of investigations on
observer-based fault diagnosis systems are dedicated to observer design
issues. Keeping in mind that any fault detection (or isolation) problem is
in its core to distinguish two different system operations, i.e. the fault
(to be detected) vs. uncertainties as fault detection, or two different
classes of faults as fault isolation, research on observer-based residual
generation, beginning in its early stage, was shaped into the trade-off
framework of sensitivity (e.g. to fault) vs. robustness (against
uncertainty) \cite{Frank90,DingCDC93,Patton-chen-93,DingAUTO94,HouUKACC96}.
This development is a natural result, as robust control dominated control
theory and engineering between 80's and 90's, and has stamped the progress
of observer-based fault diagnosis technique since then. Reviewing the
literature in the relevant research and application domains gives the
impression that there is no more efficient and systematic alternative to
observer-based fault diagnosis technique in dealing with fault diagnosis in
dynamic systems, although there exist a number of open and challenging
issues even after extensive studies over past decades. These include,

\begin{itemize}
\item efficient and optimal detection of parametric (multiplicative) faults,
in particular in the presence of model uncertainties as well as in feedback
control systems,

\item definition and introduction of convincing and mathematically
well-established metrics to measure the distance between nominal and faulty
operations for fault detection purpose or the distance between two different
classes of faults towards fault isolation. Although the widely accepted
optimisation criteria like $\mathcal{H}_{\infty }/\mathcal{H}_{\infty },$ $%
\mathcal{H}_{-}/\mathcal{H}_{\infty }$ and $\mathcal{H}_{2}$ \cite%
{DingCDC93,DingAUTO94,HouUKACC96,DingIAS00,ZhongAUTO03,WYL2007} are applied
for optimal design of observer-based residual generators, they only reflect
the influences of faults and uncertainties on system dynamics under
consideration, and are not metric for measuring the distance of different
system operations,

\item residual generation and threshold setting in an integrated manner.
This is indeed the consequence of the above-mentioned problem with missing
metric for distance measurement, and

\item a uniform design of fault diagnosis systems, both in the model-based
and data-driven fashions.
\end{itemize}

These facts have considerably motivated us to explore potential alternative
strategies in recent years.

\bigskip

Hilbert space $\mathcal{H}$ is a vector space endowed with an inner product $%
\left\langle \cdot ,\cdot \right\rangle $. Given a vector $\alpha \in 
\mathcal{H},$ the norm of $\alpha $ is induced by the inner product, $%
\left\Vert \alpha \right\Vert ^{2}=\left\langle \alpha ,\alpha \right\rangle
.$ In this regard, Hilbert space is a complete metric space \cite%
{Kato_book,Feintuch_book}. Accordingly, a classification problem can be well
defined in Hilbert space as: given a subspace $\mathcal{K}\in \mathcal{H},$
check if $\alpha $ $\in \mathcal{H}$ belongs to $\mathcal{K}.$ The solution
is straightforward and consists of two steps: (i) projecting $\alpha $ onto $%
\mathcal{K},$ and then (ii) calculating the distance between $\alpha $ and
its projection using the induced norm. Figure 1 schematically illustrates
this projection-based classification problem and its solution. 
\begin{figure}[h]
\centering\includegraphics[width=10cm,height=5cm]{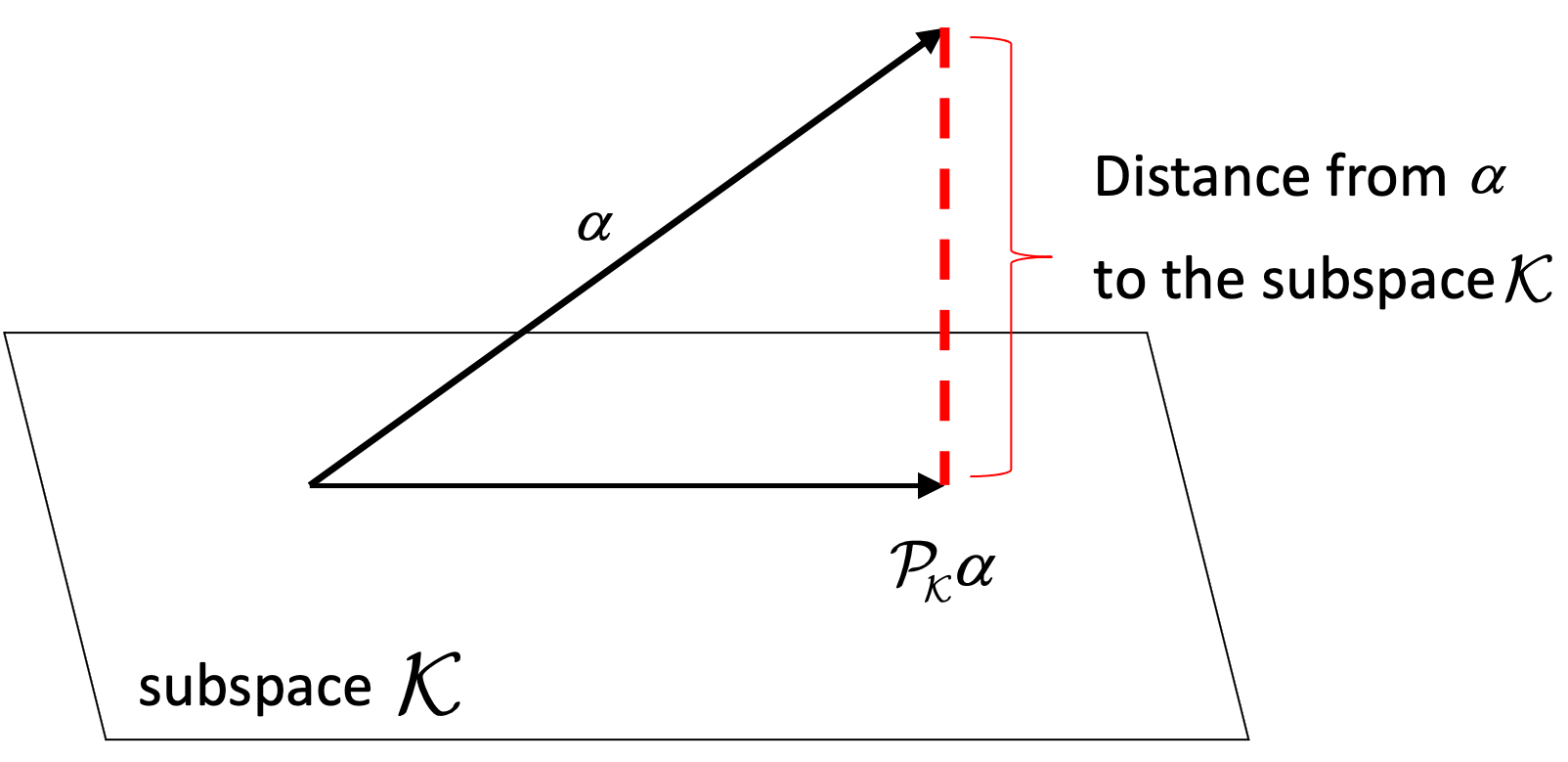}
\caption{Schemetic description of projection-based classification ($\mathcal{%
P}_{\mathcal{K}}\protect\alpha $ denotes the projection of $\protect\alpha $
onto $\mathcal{K}$)}
\end{figure}

This conceptional formulation of classification problem in terms of
projection and distance in Hilbert subspace, and above all, the facts that

\begin{itemize}
\item all signals in dynamic system can be represented by vectors in a
Hilbert subspace \cite{Francis87}, and

\item system dynamics can be modelled as subspaces in Hilbert space as well 
\cite{Francis87,Vinnicombe-book,Feintuch_book}
\end{itemize}

are convincing arguments for us to study fault diagnosis issues using
projection-based methods in the above described context. To this end, we
will first formulate fault detection and isolation issues, in particular
those open ones like detecting parametric faults in uncertainty systems or
fault detection in feedback control systems, as classification problems in
Hilbert subspaces and solve them using projection-based methods. This is the 
\textit{first} (intended) \textit{contribution} of our work towards
establishing an alternative paradigm for fault diagnosis in dynamic systems.
We will further demonstrate that, in this framework, fault detection and
isolation issues can be systematically handled in a uniform manner. To be
specific, residual generation and threshold setting will be addressed by
determining (i) the distance of a (signal) vector to a (system) subspace and
(ii) the gap between two (system) subspaces. In particular, the threshold
becomes residual-driven, which significantly increases the fault
detectability. These results would be the \textit{second contribution} of
our work. In our study, considerable attention will be paid to relations
between the proposed projection-based methods and the well-established
observer-based framework as well, and comparison will be made when possible.
As the \textit{third contribution}, it will be demonstrated that the
proposed projection-based methods result in higher fault detectability.
Moreover, a projection-based fault detection system will be proposed, which
is comparable with the existing observer-based systems with regard to the
needed online computation, but delivers better fault detection performance.

\bigskip

Projection methods are a standard technique in machine and statistic
learning and extensively applied in data-driven fault diagnosis. The most
popular ones are principal component analysis (PCA) method \cite%
{Jolliffe86,CRB2001,QinJC03,Ding2014} and partial least squares (PLS)
algorithm \cite{CRB2001,LQZautomatica2010,Ding2014}. Although their
application is limited to detecting faults in statistic processes (including
steady dynamic systems), PCA and PLS are, thanks to their uniform
formulation and computations, and above all, their clear geometric
interpretation of projection onto subspaces, not only widely accepted in
industry, but also recognised in research as a conceptional basis for the
development of advanced fault diagnosis methods. This example also inspirits
our efforts towards a projection-based fault diagnosis framework. Noticing
that projections adopted in PCA and PLS deal with linear algebraic
computations, and thus are limited to solving fault diagnosis problems in
statistic processes, different mathematical and control theoretic tools are
needed for our study on fault diagnosis in dynamic systems. We will consider
system signals (data) defined in Hilbert spaces like Lebesgue space or Hardy
space, and model the dynamic systems under consideration as subspaces in
Hilbert space, for instance, image and kernel subspaces \cite%
{Francis87,Vinnicombe-book}. In order to deal with orthogonal projections
and distance between (signal) vector and (system) subspace or gap between
two subspaces, basics of operator theory, gap metric as well as their
computations and realisations in form of dynamic systems will serve as the
major tools for our study.

\bigskip

The paper is organised as follows. In Section 2, necessary preliminaries of
system representations, orthogonal projection in Hilbert space as well as
gap metric as a similarity measurement between two Hilbert subspaces are
introduced. Section 3 is dedicated to the establishment of the basic form of
the projection-based fault detection framework. It includes (i) formulation
and solution of orthogonal projection-based residual generation, (ii) the
realisation and online computation algorithms, and (iii) gap metric-based
threshold setting. Section 4 demonstrates solutions of fault detection in
feedback control systems achieved in the framework of projection-based fault
detection. In Section 5, fault detection and isolation issues are formulated
as binary and multi-class classification problems and solved by means of
projection-based methods. Two modified projection-based fault detection
methods are proposed in Section 6. The first one is comparable with existing
observer-based fault detection schemes with regard to the associated online
computation. The second one provides us with a practical solution with
finite evaluation time interval. This method can be realised in the
data-driven fashion as well. In our experimental study in Section 7,
applications of some of the proposed methods and algorithms are illustrated
on a laboratory three-tank system. Finally, in Section 8, the major results
are first summarised, and two remarks on (i) the projection-based
interpretation of the so-called unified solution for detecting additive
faults, and (ii) computation of a type of gap metric are included.

\bigskip

Throughout this paper, standard notations known in advanced control theory
and linear algebra are adopted. In addition, $\mathcal{L}_{2}\left( -\infty
,\infty \right) =\mathcal{L}_{2}\left( -\infty ,0\right] \oplus \mathcal{L}%
_{2}\left[ 0,\infty \right) $ is time domain space of all square summable
Lebesgue signals (signals with bounded energy). $\mathcal{H}_{2}$ is the
space of Fourier transforms of signals in $\mathcal{L}_{2}\left[ 0,\infty
\right) .$ $\mathcal{H}_{\infty }$ $\left( \mathcal{RH}_{\infty }\right) $
is used to denote the set of all stable systems (with a real rational
transfer function). For transfer matrix $G(z),G^{\sim }(z)=G^{T}(z^{-1}).$
By the application of projection-based methods, following notations are
adopted. $\mathcal{P}_{\mathcal{I}}$ denotes an orthogonal projection onto
subspace $\mathcal{I}$, which is an operator whose norm is denoted by $%
\left\Vert \mathcal{P}_{\mathcal{I}}\right\Vert .$ $\mathcal{P}_{\mathcal{I}%
}^{\ast }$ is the adjoint of $\mathcal{P}_{\mathcal{I}}.$ $\mathcal{K}^{\bot
}$ represents the orthogonal complement of $\mathcal{K}$. $\mathcal{L}_{K}$
stands for Laurent (multiplication) operator with symbol $K$ \cite%
{Francis87,Vinnicombe-book,Feintuch_book}.

\section{Preliminaries}

As the methodological basis of our work, we first introduce the concepts of
kernel and image representations as alternative system model forms and some
associated computation issues. It is followed by the introduction of
orthogonal projection in Hilbert space and gap metric as a similarity
measurement between two closed Hilbert subspaces.

\subsection{Kernel and image representations and subspaces}

Consider discrete-time linear time invariant (LTI) systems modelled by%
\begin{equation}
y(z)=G(z)u(z),y(z)\in \mathbb{C}^{m},u(z)\in \mathbb{C}^{p}  \label{eq2-1}
\end{equation}%
with $u$ and $y$ as the plant input and output vectors. It is assumed that $%
G(z)$ is a proper real-rational matrix and its minimal state space
realisation is given by 
\begin{align}
x(k+1)& =Ax(k)+Bu(k),x(0)=x_{0},  \label{eq2-2a} \\
y(k)& =Cx(k)+Du(k),  \label{eq2-2b}
\end{align}%
where $x\in \mathbb{R}^{n}$ is the state vector and $x_{0}$ is the initial
condition of the system. Matrices $A,B,C,D$ are appropriately dimensioned
real constant matrices. A factorisation of a transfer function matrix over $%
\mathcal{RH}_{\infty }$ gives a further system representation form and
factorises the transfer matrix into two stable transfer matrices. The left
and right factorisations (LF and RF) of $G(z)$ are given by 
\begin{equation}
G(z)=\hat{M}^{-1}(z)\hat{N}(z)=N(z)M^{-1}(z),  \label{eq2-3}
\end{equation}%
whose state space representations are 
\begin{align}
\hat{M}(z)& =\left( A-LC,-L,WC,W\right) ,\hat{N}(z)=\left(
A-LC,B-LD,WC,WD\right) ,  \label{eq2-4a} \\
M(z)& =\left( A+BF,BV,F,V\right) ,N(z)=\left( A+BF,BV,C+DF,DV\right) ,
\label{eq2-4b}
\end{align}%
where (real) matrices $F$ and $L$ are selected such that $A+BF$ and $A-LC$
are Schur matrices, and $W$ and $V$ are invertible \cite{Hoffmann1996}. $%
\left( \hat{M},\hat{N}\right) $ and $\left( M,N\right) $ build left and
right coprime pairs (LCP and RCP), if there exist $\left( \hat{X},\hat{Y}%
\right) $ and $\left( X,Y\right) $ over $\mathcal{RH}_{\infty }$ so that the
Bezout identity holds%
\begin{equation}
\left[ 
\begin{array}{cc}
X(z) & \text{ }Y(z) \\ 
-\hat{N}(z) & \text{ }\hat{M}(z)%
\end{array}%
\right] \left[ 
\begin{array}{cc}
M(z) & \text{ }-\hat{Y}(z) \\ 
N(z) & \text{ }\hat{X}(z)%
\end{array}%
\right] =\left[ 
\begin{array}{cc}
I\text{ } & 0\text{ } \\ 
0\text{ } & I\text{ }%
\end{array}%
\right] .  \label{eq2-5}
\end{equation}%
The state space computation formulas of $\left( \hat{X},\hat{Y}\right) $ and 
$\left( X,Y\right) $ are \cite{Hoffmann1996}%
\begin{align}
\hat{X}(z)& =\left( A+BF,L,C+DF,W^{-1}\right) ,\hat{Y}(z)=\left(
A+BF,-LW^{-1},F,0\right) ,  \label{eq2-6a} \\
X(z)& =\left( A-LC,-(B-LD),F,V^{-1}\right) ,Y(z)=\left(
A-LC,-L,V^{-1}F,0\right) .  \label{eq2-6b}
\end{align}%
It follows from (\ref{eq2-3})-(\ref{eq2-4b}) that

\begin{itemize}
\item the LCP of $G(z)$ can be equivalently written as 
\begin{align}
r_{0}(z)& =\hat{M}(z)y(z)-\hat{N}(z)u(z)\Longrightarrow  \label{eq2-7a} \\
\hat{x}(k+1)& =\left( A-LC\right) \hat{x}(k)+\left( B-LD\right) u(k)+Ly(k),
\\
r_{0}(k)& =W\left( y(k)-\hat{y}(k)\right) ,\hat{y}(k)=C\hat{x}(k)+Du(k),
\label{eq2-7b}
\end{align}%
and, if $x(0)=\hat{x}(0),$ it holds $r_{0}(k)=0,$

\item the RCP of $G(z)$ can be represented by, for some $v\in \mathcal{H}%
_{2},$%
\begin{align}
u(z)& =M(z)v(z)\Longleftrightarrow
M^{-1}(z)u(z)=v(z),y(z)=N(z)v(z)\Longrightarrow  \label{eq2-8a} \\
x(k+1)& =\left( A+BF\right) x(k)+BVv(k) \\
\left[ 
\begin{array}{c}
u(k) \\ 
y(k)%
\end{array}%
\right] & =\left[ 
\begin{array}{c}
F \\ 
C+DF%
\end{array}%
\right] x(k)+\left[ 
\begin{array}{c}
V \\ 
DV%
\end{array}%
\right] v(k).  \label{eq2-8b}
\end{align}
\end{itemize}

System (\ref{eq2-7a})-(\ref{eq2-7b}) is the well-known observer-based
residual generator with $r_{0}(k)$ being a residual vector and $W$ as a
post-filter, while system (\ref{eq2-8a})-(\ref{eq2-8b}) describes the
closed-loop dynamics of a state feedback control with $v(k)$ as reference
vector and $V$ as a pre-filter. Systems (\ref{eq2-7a})-(\ref{eq2-7b}) and (%
\ref{eq2-8a})-(\ref{eq2-8b}) are called stable kernel and image
representations (SKR and SIR) of $G(z).$ For the sake of simplicity, we
introduce the following notation for SKR and SIR,%
\begin{align}
\text{SKR}& :r_{0}(z)=\left[ 
\begin{array}{cc}
-\hat{N}(z) & \hat{M}(z)%
\end{array}%
\right] \left[ 
\begin{array}{c}
u(z) \\ 
y(z)%
\end{array}%
\right] ,  \label{eq2-9a} \\
\text{SIR}& :\left[ 
\begin{array}{c}
u(z) \\ 
y(z)%
\end{array}%
\right] =\left[ 
\begin{array}{c}
M(z) \\ 
N(z)%
\end{array}%
\right] v(z).  \label{eq2-9b}
\end{align}%
In our subsequent study, the so-called normalised SKR and SIR play an
important role, which are denoted by $K_{G}$ and $I_{G}$ and defined by%
\begin{align*}
K_{G}(z)K_{G}^{\sim }(z)& =\hat{N}_{0}(z)\hat{N}_{0}^{\sim }(z)+\hat{M}%
_{0}(z)\hat{M}_{0}^{\sim }(z)=I, \\
I_{G}^{\sim }(z)I_{G}(z)& =M_{0}^{\sim }(z)M_{0}(z)+N_{0}^{\sim
}(z)N_{0}(z)=I,
\end{align*}%
where $\left( \hat{M}_{0},\hat{N}_{0}\right) $ and $\left(
M_{0},N_{0}\right) $ are LCP and RCP with $L=L_{0},W=W_{0},F=F_{0}$ and $%
V=V_{0},$ as given below \cite{Hoffmann1996}:%
\begin{align}
L& =L_{0}=\left( BD^{T}+APC^{T}\right) \left( I+DD^{T}+CPC^{T}\right) ^{-1},
\label{eq2-10a} \\
W& =W_{0}=U\left( I+DD^{T}+CPC^{T}\right)
^{-1/2},UU^{T}=I\Longleftrightarrow W_{0}\left( I+DD^{T}+CPC^{T}\right)
W_{0}^{T}=I,  \label{eq2-10b} \\
F& =F_{0}=-\left( I+D^{T}D+B^{T}QB\right) ^{-1}\left( D^{T}C+B^{T}QA\right) ,
\label{eq2-10c} \\
V& =V_{0}=\left( I+D^{T}D+B^{T}QB\right) ^{-1/2}\Gamma ,\Gamma ^{T}\Gamma
=I\Longleftrightarrow V_{0}^{T}\left( I+D^{T}D+B^{T}QB\right) V_{0}=I.
\label{eq2-10d}
\end{align}%
Here, $P>0,Q>0$ solve the following Riccati equations respectively,%
\begin{align*}
P& =APA^{T}+BB^{T}-\left( BD^{T}+APC^{T}\right) \left(
I+DD^{T}+CPC^{T}\right) ^{-1}\left( BD^{T}+APC^{T}\right) ^{T}, \\
Q& =A^{T}QA+C^{T}C-\left( D^{T}C+B^{T}QA\right) ^{T}\left(
I+D^{T}D+B^{T}QB\right) ^{-1}\left( D^{T}C+B^{T}QA\right) .
\end{align*}

\begin{Rem}
Hereafter, we may drop out the domain variable $z$\ or $k$ when there is no
risk of confusion.
\end{Rem}

The $\mathcal{H}_{2}$ input and output vectors $\left[ 
\begin{array}{c}
u \\ 
y%
\end{array}%
\right] $ satisfying (\ref{eq2-7a})-(\ref{eq2-7b}) or generated by $v\in 
\mathcal{H}_{2}$ according to (\ref{eq2-8a})-(\ref{eq2-8b}) build subspaces
in $\mathcal{H}_{2}$. For our purpose, the following definitions of kernel
and image subspaces are introduced.

\begin{Def}
Given the model (\ref{eq2-1}) and the corresponding LCP and RCP $\left( \hat{%
M},\hat{N}\right) $ and $\left( M,N\right) ,$ the $\mathcal{H}_{2}$ subspace 
$\mathcal{K}_{G}$ defined by 
\begin{equation}
\mathcal{K}_{G}=\left\{ \left[ 
\begin{array}{c}
u \\ 
y%
\end{array}%
\right] \in \mathcal{H}_{2}:\left[ 
\begin{array}{cc}
-\hat{N} & \hat{M}%
\end{array}%
\right] \left[ 
\begin{array}{c}
u \\ 
y%
\end{array}%
\right] =0\hspace{-2pt}\right\}  \label{eq2-11a}
\end{equation}%
is called kernel subspace of $G,$ and the $\mathcal{H}_{2}$ subspace $%
\mathcal{I}_{G}$ defined by%
\begin{equation}
\mathcal{I}_{G}=\left\{ \left[ 
\begin{array}{c}
u \\ 
y%
\end{array}%
\right] \in \mathcal{H}_{2}:\left[ 
\begin{array}{c}
u \\ 
y%
\end{array}%
\right] =\left[ 
\begin{array}{c}
M \\ 
N%
\end{array}%
\right] v,v\in \mathcal{H}_{2}\right\}  \label{eq2-11b}
\end{equation}%
is image subspace.
\end{Def}

It is evident that the kernel and image subspaces $\mathcal{K}_{G}$ and $%
\mathcal{I}_{G}$ consist of all (bounded) input and output pairs $(u,y).$
They are closed subspace in $\mathcal{H}_{2}$ \cite{Vinnicombe-book}.

\subsection{Orthogonal projection and gap metric}

An orthogonal projection on a subspace $\mathcal{V},$ denoted by $\mathcal{P}%
_{\mathcal{V}},$ in Hilbert space endowed with the inner product, 
\begin{equation}
\left\langle x,y\right\rangle =\sum\limits_{k=0}^{\infty
}x^{T}(k)y(k),x,y\in \mathcal{H}_{2},  \label{eq2-12a}
\end{equation}%
is a linear operator satisfying \cite{Kato_book}%
\begin{equation}
x,y\in \mathcal{V},\mathcal{P}_{\mathcal{V}}^{2}=\mathcal{P}_{\mathcal{V}%
},\left\langle \mathcal{P}_{\mathcal{V}}x,y\right\rangle =\left\langle x,%
\mathcal{P}_{\mathcal{V}}y\right\rangle .  \label{eq2-12}
\end{equation}%
The following well-known properties of an orthogonal projection are of
importance for our subsequent study \cite{Kato_book,Feintuch_book}:

\begin{itemize}
\item given $x\in \mathcal{H}_{2},$ the (orthogonal) projection of $x$ onto $%
\mathcal{V},$ $\mathcal{P}_{\mathcal{V}}x,$ satisfies%
\begin{equation}
\left\langle \mathcal{P}_{\mathcal{V}}x,x-\mathcal{P}_{\mathcal{V}%
}x\right\rangle =0.  \label{eq2-13}
\end{equation}%
Often, $\mathcal{P}_{\mathcal{V}}x$ serves as an estimate for $x,$ and in
this context, $x-\mathcal{P}_{\mathcal{V}}x$ is understood as the estimation
error. That means, the projection-induced estimate is orthogonal to the
estimation error;

\item relation (\ref{eq2-13}) can also be equivalently expressed by%
\begin{equation*}
x=\mathcal{P}_{\mathcal{V}}x+\mathcal{P}_{\mathcal{V}^{\bot }}x,
\end{equation*}%
where $\mathcal{V}^{\bot }$ is the orthogonal complement of $\mathcal{V};$

\item and given $y\in \mathcal{H}_{2},\forall x\in \mathcal{V}\in \mathcal{H}%
_{2},$%
\begin{equation}
\left\langle y-x,y-x\right\rangle =\left\Vert y-x\right\Vert _{2}^{2}\geq
\left\Vert y-\mathcal{P}_{\mathcal{V}}y\right\Vert _{2}^{2}.  \label{eq2-14}
\end{equation}
\end{itemize}

Given a closed subspace $\mathcal{V}\in \mathcal{H}_{2}$ and a vector $y\in 
\mathcal{H}_{2},$ the distance between $y$ and $\mathcal{V},dist\left( y,%
\mathcal{V}\right) ,$ is defined as%
\begin{equation*}
dist\left( y,\mathcal{V}\right) =\inf_{x\in \mathcal{V}}\left\Vert
y-x\right\Vert _{2},
\end{equation*}%
which, following (\ref{eq2-14}), can be computed as 
\begin{equation*}
dist\left( y,\mathcal{V}\right) =\left( \mathcal{I}-\mathcal{P}_{\mathcal{V}%
}\right) y=\mathcal{P}_{\mathcal{V}^{\bot }}y.
\end{equation*}%
Here, $\mathcal{I}$ is the unit operator.

\bigskip

In order to measure the similarity of two (closed) subspaces in Hilbert
space $\mathcal{H}$, the concept of gap metric is established \cite%
{Kato_book,Feintuch_book}. Given two closed subspaces $\mathcal{V},\mathcal{U%
}\in \mathcal{H},$ the gap metric between them is defined by 
\begin{align}
\delta \left( \mathcal{V},\mathcal{U}\right) & =\max \left\{ \vec{\delta}%
\left( \mathcal{V},\mathcal{U}\right) ,\vec{\delta}\left( \mathcal{U},%
\mathcal{V}\right) \right\} ,  \label{eq2-15a} \\
\vec{\delta}\left( \mathcal{V},\mathcal{U}\right) & =\sup_{\substack{ x\in 
\mathcal{V}  \\ \left\Vert x\right\Vert _{2}=1}}dist\left( x,\mathcal{U}%
\right) =\left\Vert \left( \mathcal{I}-\mathcal{P}_{\mathcal{U}}\right) 
\mathcal{P}_{\mathcal{V}}\right\Vert =\sup_{x\in \mathcal{V}}\inf_{y\in 
\mathcal{U}}\frac{\left\Vert x-y\right\Vert _{2}}{\left\Vert x\right\Vert
_{2}},  \label{eq2-15} \\
\vec{\delta}\left( \mathcal{U},\mathcal{V}\right) & =\sup_{\substack{ y\in 
\mathcal{U}  \\ \left\Vert y\right\Vert _{2}=1}}dist\left( y,\mathcal{V}%
\right) =\left\Vert \left( \mathcal{I}-\mathcal{P}_{\mathcal{V}}\right) 
\mathcal{P}_{\mathcal{U}}\right\Vert =\sup_{y\in \mathcal{U}}\inf_{x\in 
\mathcal{V}}\frac{\left\Vert y-x\right\Vert _{2}}{\left\Vert y\right\Vert
_{2}}.
\end{align}%
Here, $\vec{\delta}\left( \mathcal{V},\mathcal{U}\right) $ and $\vec{\delta}%
\left( \mathcal{U},\mathcal{V}\right) $ are called directed gap. The
following properties are well-known \cite{Kato_book,Feintuch_book} and
useful for our subsequent investigation:%
\begin{align*}
0& \leq \delta \left( \mathcal{V},\mathcal{U}\right) \leq 1, \\
\text{for }\delta \left( \mathcal{V},\mathcal{U}\right) & <1,\vec{\delta}%
\left( \mathcal{V},\mathcal{U}\right) =\vec{\delta}\left( \mathcal{U},%
\mathcal{V}\right) =\delta \left( \mathcal{V},\mathcal{U}\right) , \\
\text{for }\delta \left( \mathcal{V},\mathcal{U}\right) & =0,\mathcal{V}=%
\mathcal{U},\text{and }\delta \left( \mathcal{V},\mathcal{U}\right) =1,%
\mathcal{V}\bot \mathcal{U}.
\end{align*}

\subsection{Problem formulation}

With the aid of the defined image/kernel subspace and orthogonal projection
operator, we are now in the position to concretise fault detection problem
in terms of projection-based classification, as sketched in Figure 1. Given
system measurement vector $\left[ 
\begin{array}{c}
u \\ 
y%
\end{array}%
\right] \in \mathcal{H}_{2},$ find its orthogonal projection onto image
subspace $\mathcal{I}_{G},\mathcal{P}_{\mathcal{I}_{G}},$ and further
determine the distance $dist\left( \left[ 
\begin{array}{c}
u \\ 
y%
\end{array}%
\right] ,\mathcal{I}_{G}\right) $ so that a decision can be made on the
basis of the distance value with respect to an established threshold. In
this regard, the following problems should be solved at first:

\begin{itemize}
\item definition of orthogonal projection operator and computation of $%
dist\left( \left[ 
\begin{array}{c}
u \\ 
y%
\end{array}%
\right] ,\mathcal{I}_{G}\right) ,$

\item online realisation algorithm towards constructing a fault detection
system, and

\item determination of threshold. In our work, threshold is to be set to
guarantee that model uncertainties will not cause false alarms.
\end{itemize}

This work is essential to establish the intended projection-based fault
diagnosis framework. The first application in this framework is detecting
(parametric) faults in feedback control systems, a challenging and open
issue. As further topics, we will study fault detection and isolation,
formulated as binary and multi-class classification problems, in the
projection-based framework.

\bigskip

In our work, attention will be paid to comparison study with the
well-established observer-based fault diagnosis methods when possible. In
this context, we will propose a modified projection-based fault detection
scheme that is comparable with an observer-based fault detection system with
regard to online computation but delivers better detection performance.
Considering that realisation of the inner product defined in (\ref{eq2-12a})
requires, theoretically, data over $\left[ 0,\infty \right) ,$ a further
modified projection-based method will be developed, which allows an optimal
fault detection over finite time interval. This method will enable us to
realise a data-driven implementation of projection-based fault detection as
well. To illustrate and demonstrate the methods developed in our work,
experimental study on the laboratory three-tank system and some achieved
results will be finally presented.

\section{Basic fault detection methods based on orthogonal projection}

In this section, we apply orthogonal projection technique to achieving fault
detection. Our major focus is on optimally detecting faults in dynamic
systems with model uncertainties.

\subsection{Orthogonal projection-based residual generation}

It is a known result \cite{Georgiou88} that the orthogonal projection onto
the image subspace $\mathcal{I}_{G}$ is given by 
\begin{equation}
\mathcal{P}_{\mathcal{I}_{G}}=\mathcal{L}_{I_{G}}\mathcal{L}_{I_{G}}^{\ast },%
\mathcal{L}_{I_{G}}^{\ast }:\mathcal{H}_{2}\rightarrow \mathcal{H}_{2},%
\mathcal{L}_{I_{G}}^{\ast }=\mathcal{P}_{\mathcal{H}_{2}}\mathcal{L}%
_{I_{G}^{\sim }}.  \label{eq3-1}
\end{equation}%
Correspondingly, the projection of a data vector $\left[ 
\begin{array}{c}
u \\ 
y%
\end{array}%
\right] \in \mathcal{H}_{2}$ onto the image subspace is 
\begin{equation*}
p_{\mathcal{I}_{G}}=\mathcal{P}_{\mathcal{I}_{G}}\left[ 
\begin{array}{c}
u \\ 
y%
\end{array}%
\right] .
\end{equation*}%
The difference between $\left[ 
\begin{array}{c}
u \\ 
y%
\end{array}%
\right] $ and $p_{\mathcal{I}_{G}}$ as its estimate, 
\begin{equation}
\left[ 
\begin{array}{c}
u \\ 
y%
\end{array}%
\right] -p_{\mathcal{I}_{G}}=\left( \mathcal{I}-\mathcal{P}_{\mathcal{I}%
_{G}}\right) \left[ 
\begin{array}{c}
u \\ 
y%
\end{array}%
\right] =\left( \mathcal{I}-\mathcal{L}_{I_{G}}\mathcal{P}_{\mathcal{H}_{2}}%
\mathcal{L}_{I_{G}^{\sim }}\right) \left[ 
\begin{array}{c}
u \\ 
y%
\end{array}%
\right] ,  \label{eq3-2}
\end{equation}%
indicates how far the measurement (data) vector $\left[ 
\begin{array}{c}
u \\ 
y%
\end{array}%
\right] $ deviates from the nominal system dynamics expressed by the system
SIR. Thus, its $l_{2}$-norm is the distance between the data vector and the
image subspace. In this regard, we introduce the following definition.

\begin{Def}
\label{Def3-1}Given the model (\ref{eq2-1}) and the corresponding operator $%
\mathcal{L}_{I_{G}}$, system $\left( \mathcal{I}-\mathcal{P}_{\mathcal{I}%
_{G}}\right) \left[ 
\begin{array}{c}
u \\ 
y%
\end{array}%
\right] $ is called projection-based residual generator with output 
\begin{equation}
r_{\mathcal{I}_{G}}=\left[ 
\begin{array}{c}
u \\ 
y%
\end{array}%
\right] -p_{\mathcal{I}_{G}}=\left( \mathcal{I}-\mathcal{P}_{\mathcal{I}%
_{G}}\right) \left[ 
\begin{array}{c}
u \\ 
y%
\end{array}%
\right]  \label{eq3-3}
\end{equation}%
as projection-based residual. The $l_{2}$-norm of $r_{\mathcal{I}_{G}},$%
\begin{equation}
\left\Vert r_{\mathcal{I}_{G}}\right\Vert _{2}=dist\left( \left[ 
\begin{array}{c}
u \\ 
y%
\end{array}%
\right] ,\mathcal{I}_{G}\right)  \label{eq3-3a}
\end{equation}%
is the distance from $\left[ 
\begin{array}{c}
u \\ 
y%
\end{array}%
\right] $ to $\mathcal{I}_{G}$.
\end{Def}

It should be remarked that $\mathcal{I}-\mathcal{L}_{I_{G}}\mathcal{P}_{%
\mathcal{H}_{2}}\mathcal{L}_{I_{G}^{\sim }}$ defines the projection onto $%
\mathcal{I}_{G}^{\bot }$. That is%
\begin{eqnarray*}
\mathcal{I}_{G}^{\bot } &=&\left\{ r_{\mathcal{I}_{G}}:r_{\mathcal{I}%
_{G}}=\left( \mathcal{I}-\mathcal{L}_{I_{G}}\mathcal{P}_{\mathcal{H}_{2}}%
\mathcal{L}_{I_{G}^{\sim }}\right) \left[ 
\begin{array}{c}
u \\ 
y%
\end{array}%
\right] ,\left[ 
\begin{array}{c}
u \\ 
y%
\end{array}%
\right] \in \mathcal{H}_{2}\right\} , \\
\mathcal{H}_{2} &=&\mathcal{I}_{G}\oplus \mathcal{I}_{G}^{\bot }.
\end{eqnarray*}%
In other words, the residual subspace is in fact $\mathcal{I}_{G}^{\bot }$.
Under consideration of our study purpose, we prefer the term ``residual" over
projection onto the orthogonal complement of $\mathcal{I}_{G}$.

\bigskip

Remembering the natural relation between SKR and the residual generation as
well as the known fact that $\mathcal{K}_{G}=\mathcal{I}_{G}$ \cite%
{Vinnicombe-book}, it is of considerable interest to investigate relations
among the operator $\mathcal{L}_{I_{G}},$ the orthogonal projection $%
\mathcal{P}_{\mathcal{I}_{G}}$ and the system SKR.

\begin{Le}
\label{Cor3-1}Let $\mathcal{L}_{K_{G}}$ and $\mathcal{L}_{K_{G}^{\sim }}$ be
Laurent operators with symbol $K_{G}$ and $K_{G^{\sim }},$ respectively. It
holds%
\begin{gather}
\mathcal{L}_{K_{G}^{\sim }}\mathcal{L}_{K_{G}}+\mathcal{L}_{I_{G}}\mathcal{L}%
_{I_{G}^{\sim }}=\mathcal{I},  \label{eq3-41} \\
\mathcal{I}-\mathcal{P}_{\mathcal{H}_{2}}\mathcal{L}_{K_{G}^{\sim }}\mathcal{%
L}_{K_{G}}=\mathcal{P}_{\mathcal{I}_{G}}+\mathcal{P}_{\mathcal{H}_{2}}%
\mathcal{L}_{I_{G}}\mathcal{P}_{\mathcal{H}_{2}^{\bot }}\mathcal{L}%
_{I_{G}^{\sim }}.  \label{eq3-42}
\end{gather}
\end{Le}

\begin{proof}
Identity (\ref{eq3-41}) is the result of 
\begin{equation*}
\left[ 
\begin{array}{cc}
M_{0}^{\sim } & N_{0}^{\sim } \\ 
-\hat{N}_{0} & \hat{M}_{0}%
\end{array}%
\right] \left[ 
\begin{array}{cc}
M_{0} & -\hat{N}_{0}^{\sim } \\ 
N_{0} & \hat{M}_{0}^{\sim }%
\end{array}%
\right] =\left[ 
\begin{array}{cc}
M_{0} & -\hat{N}_{0}^{\sim } \\ 
N_{0} & \hat{M}_{0}^{\sim }%
\end{array}%
\right] \left[ 
\begin{array}{cc}
M_{0}^{\sim } & N_{0}^{\sim } \\ 
-\hat{N}_{0} & \hat{M}_{0}%
\end{array}%
\right] =\left[ 
\begin{array}{cc}
I & 0 \\ 
0 & I%
\end{array}%
\right] .
\end{equation*}%
It follows from (\ref{eq3-41}) that  
\begin{equation*}
\mathcal{P}_{\mathcal{H}_{2}}\mathcal{L}_{K_{G}^{\sim }}\mathcal{L}_{K_{G}}=%
\mathcal{P}_{\mathcal{H}_{2}}\left( \mathcal{I}-\mathcal{L}_{I_{G}}\mathcal{L%
}_{I_{G}^{\sim }}\right) ,
\end{equation*}%
which yields  
\begin{align*}
\mathcal{I}-\mathcal{P}_{\mathcal{H}_{2}}\mathcal{L}_{K_{G}^{\sim }}\mathcal{%
L}_{K_{G}}& =\mathcal{P}_{\mathcal{H}_{2}}\mathcal{L}_{I_{G}}\mathcal{L}%
_{I_{G}^{\sim }}=\mathcal{P}_{\mathcal{H}_{2}}\mathcal{L}_{I_{G}}\mathcal{P}%
_{\mathcal{H}_{2}}\mathcal{L}_{I_{G}^{\sim }}+\mathcal{P}_{\mathcal{H}_{2}}%
\mathcal{L}_{I_{G}}\mathcal{P}_{\mathcal{H}_{2}^{\bot }}\mathcal{L}%
_{I_{G}^{\sim }} \\
& =\mathcal{P}_{\mathcal{I}_{G}}+\mathcal{P}_{\mathcal{H}_{2}}\mathcal{L}%
_{I_{G}}\mathcal{P}_{\mathcal{H}_{2}^{\bot }}\mathcal{L}_{I_{G}^{\sim }}.
\end{align*}
\end{proof}

Equation (\ref{eq3-42}) illustrates that, although $\mathcal{I}-\mathcal{P}_{%
\mathcal{H}_{2}}\mathcal{L}_{K_{G}^{\sim }}\mathcal{L}_{K_{G}}$ is a self
adjoint operator that maps vectors in $\mathcal{H}_{2}$ to $\mathcal{H}_{2},$
it is not a projection operator. The operator $\mathcal{P}_{\mathcal{H}_{2}}%
\mathcal{L}_{I_{G}}\mathcal{P}_{\mathcal{H}_{2}^{\bot }}\mathcal{L}%
_{I_{G}^{\sim }}$ and relations (\ref{eq3-41})-(\ref{eq3-42}) will play an
important role in our subsequent work.

\subsection{Residual generators and implementation algorithms}

In this subsection, online computation issues of $\left\Vert r_{\mathcal{I}%
_{G}}\right\Vert _{2}$ are addressed. It follows from (\ref{eq3-1}) and (\ref%
{eq3-2}) that 
\begin{equation}
\left\Vert r_{\mathcal{I}_{G}}\right\Vert _{2}^{2}=\left\Vert \left[ 
\begin{array}{c}
u \\ 
y%
\end{array}%
\right] \right\Vert _{2}^{2}-\left\Vert \mathcal{L}_{I_{G}}\mathcal{P}_{%
\mathcal{H}_{2}}\mathcal{L}_{I_{G}^{\sim }}\left[ 
\begin{array}{c}
u \\ 
y%
\end{array}%
\right] \right\Vert _{2}^{2}=\left\Vert \left[ 
\begin{array}{c}
u \\ 
y%
\end{array}%
\right] \right\Vert _{2}^{2}-\left\Vert \mathcal{P}_{\mathcal{H}_{2}}%
\mathcal{L}_{I_{G}^{\sim }}\left[ 
\begin{array}{c}
u \\ 
y%
\end{array}%
\right] \right\Vert _{2}^{2}.  \label{eq3-43}
\end{equation}%
The relation 
\begin{equation*}
\mathcal{L}_{I_{G}^{\sim }}=\mathcal{P}_{\mathcal{H}_{2}}\mathcal{L}%
_{I_{G}^{\sim }}+\mathcal{P}_{\mathcal{H}_{2}^{\bot }}\mathcal{L}%
_{I_{G}^{\sim }}
\end{equation*}%
leads to 
\begin{gather}
\left\Vert \mathcal{P}_{\mathcal{H}_{2}}\mathcal{L}_{I_{G}^{\sim }}\left[ 
\begin{array}{c}
u \\ 
y%
\end{array}%
\right] \right\Vert _{2}^{2}=\left\Vert \mathcal{L}_{I_{G}^{\sim }}\left[ 
\begin{array}{c}
u \\ 
y%
\end{array}%
\right] \right\Vert _{2}^{2}-\left\Vert \mathcal{P}_{\mathcal{H}_{2}^{\bot }}%
\mathcal{L}_{I_{G}^{\sim }}\left[ 
\begin{array}{c}
u \\ 
y%
\end{array}%
\right] \right\Vert _{2}^{2}  \notag \\
\Longrightarrow \left\Vert r_{\mathcal{I}_{G}}\right\Vert
_{2}^{2}=\left\Vert \left[ 
\begin{array}{c}
u \\ 
y%
\end{array}%
\right] \right\Vert _{2}^{2}-\left\Vert \mathcal{L}_{I_{G}^{\sim }}\left[ 
\begin{array}{c}
u \\ 
y%
\end{array}%
\right] \right\Vert _{2}^{2}+\left\Vert \mathcal{P}_{\mathcal{H}_{2}^{\bot }}%
\mathcal{L}_{I_{G}^{\sim }}\left[ 
\begin{array}{c}
u \\ 
y%
\end{array}%
\right] \right\Vert _{2}^{2}.  \label{eq3-23}
\end{gather}%
Operator 
\begin{equation*}
\mathcal{H}_{I_{G}^{\sim }}:=\mathcal{P}_{\mathcal{H}_{2}^{\bot }}\mathcal{L}%
_{I_{G}^{\sim }}:\mathcal{H}_{2}\rightarrow \mathcal{H}_{2}^{\bot }
\end{equation*}%
is the so-called Hankel operator with symbol $I_{G}^{\sim }$ \cite{Francis87}
and serves as a filter. Given the state space representation of $I_{G}^{\sim
},$%
\begin{align}
\xi (k-1)& =\bar{A}\xi (k)+\bar{B}\left[ 
\begin{array}{c}
u(k) \\ 
y(k)%
\end{array}%
\right] ,\xi \in \mathbb{R}^{n},\bar{A}=\left( A+BF_{0}\right) ^{T},\bar{B}=%
\left[ 
\begin{array}{cc}
F_{0}^{T} & \left( C+DF\right) ^{T}%
\end{array}%
\right] ,  \label{eq3-20a} \\
\varsigma (k)& =\bar{C}\xi (k)+\bar{D}\left[ 
\begin{array}{c}
u(k) \\ 
y(k)%
\end{array}%
\right] \in \mathbb{R}^{p},\bar{C}=B^{T},\bar{D}=\left[ 
\begin{array}{cc}
V_{0}^{T} & \left( DV_{0}\right) ^{T}%
\end{array}%
\right]  \label{eq3-20b}
\end{align}%
with $\varsigma (k)$ as the output of system $I_{G}^{\sim }\left[ 
\begin{array}{c}
u \\ 
y%
\end{array}%
\right] ,$ the computation of $\left\Vert \mathcal{L}_{I_{G}^{\sim }}\left[ 
\begin{array}{c}
u \\ 
y%
\end{array}%
\right] \right\Vert _{2}$ is straightforward. Moreover, 
\begin{equation*}
\varsigma _{\mathcal{H}}:=\mathcal{H}_{I_{G}^{\sim }}\left[ 
\begin{array}{c}
u \\ 
y%
\end{array}%
\right]
\end{equation*}%
can be computed by means of discrete convolution as follows%
\begin{equation*}
\varsigma _{\mathcal{H}}(k)=-\sum\limits_{i=0}^{\infty }\bar{C}\bar{A}^{k+i}%
\bar{B}\left[ 
\begin{array}{c}
u(i) \\ 
y(i)%
\end{array}%
\right] ,k\in \left( -\infty ,0\right] .
\end{equation*}%
As demonstrated in \cite{Francis87}, $\varsigma _{\mathcal{H}}$ can also be
written as 
\begin{equation}
\varsigma _{\mathcal{H}}(k)=\left( \Psi _{o}\Psi _{c}\left[ 
\begin{array}{c}
u \\ 
y%
\end{array}%
\right] \right) (k),k\in \left( -\infty ,0\right] ,  \label{eq3-46}
\end{equation}%
where $\Psi _{o}$ and $\Psi _{c},$%
\begin{align*}
\Psi _{c}\left[ 
\begin{array}{c}
u \\ 
y%
\end{array}%
\right] & =-\sum\limits_{i=0}^{\infty }\bar{A}^{i}\bar{B}\left[ 
\begin{array}{c}
u(i) \\ 
y(i)%
\end{array}%
\right] ,\left( \Psi _{o}x\right) (k)=\bar{C}\bar{A}^{k}x(k),k\in \left(
-\infty ,0\right] , \\
& \Longrightarrow \mathcal{H}_{I_{G}^{\sim }}=\Psi _{o}\Psi _{c},
\end{align*}%
are controllability and observability operators, respectively.

\bigskip

Note that, according to (\ref{eq3-41}) given in Lemma \ref{Cor3-1}, $%
\left\Vert r_{\mathcal{I}_{G}}\right\Vert _{2}$ can be written as 
\begin{align}
\left\Vert r_{\mathcal{I}_{G}}\right\Vert _{2}^{2}& =\left\Vert \mathcal{L}%
_{K_{G}^{\sim }}\mathcal{L}_{K_{G}}\left[ 
\begin{array}{c}
u \\ 
y%
\end{array}%
\right] \right\Vert _{2}^{2}+\left\Vert \mathcal{P}_{\mathcal{H}_{2}^{\bot }}%
\mathcal{L}_{I_{G}^{\sim }}\left[ 
\begin{array}{c}
u \\ 
y%
\end{array}%
\right] \right\Vert _{2}^{2}  \notag \\
& =\left\Vert \mathcal{L}_{K_{G}}\left[ 
\begin{array}{c}
u \\ 
y%
\end{array}%
\right] \right\Vert _{2}^{2}+\left\Vert \mathcal{P}_{\mathcal{H}_{2}^{\bot }}%
\mathcal{L}_{I_{G}^{\sim }}\left[ 
\begin{array}{c}
u \\ 
y%
\end{array}%
\right] \right\Vert _{2}^{2}.  \label{eq3-24}
\end{align}%
Recalling that $\mathcal{L}_{K_{G}}\left[ 
\begin{array}{c}
u \\ 
y%
\end{array}%
\right] $ is equivalent to the observer-based residual generator, i.e.%
\begin{equation}
r_{0}=\mathcal{L}_{K_{G}}\left[ 
\begin{array}{c}
u \\ 
y%
\end{array}%
\right] =K_{G}\left[ 
\begin{array}{c}
u \\ 
y%
\end{array}%
\right] =\left[ 
\begin{array}{cc}
-\hat{N}_{0} & \hat{M}_{0}%
\end{array}%
\right] \left[ 
\begin{array}{c}
u \\ 
y%
\end{array}%
\right] ,  \label{eq2-7c}
\end{equation}%
it holds 
\begin{equation*}
\left\Vert \mathcal{L}_{K_{G}}\left[ 
\begin{array}{c}
u \\ 
y%
\end{array}%
\right] \right\Vert _{2}=\left\Vert r_{0}\right\Vert _{2},
\end{equation*}%
and thus 
\begin{equation}
\left\Vert r_{\mathcal{I}_{G}}\right\Vert _{2}^{2}=\left\Vert
r_{0}\right\Vert _{2}^{2}+\left\Vert \mathcal{P}_{\mathcal{H}_{2}^{\bot }}%
\mathcal{L}_{I_{G}^{\sim }}\left[ 
\begin{array}{c}
u \\ 
y%
\end{array}%
\right] \right\Vert _{2}^{2}.  \label{eq3-4}
\end{equation}%
Equation (\ref{eq3-4}) reveals the relation between the projection-based and
observer-based residual generation. The latter, as well-known, is the state
of the art technique for fault detection in dynamic systems. This fact
motivates us to have a close look at the term 
\begin{equation*}
\left\Vert \mathcal{L}_{I_{G}}\mathcal{P}_{\mathcal{H}_{2}^{\bot }}\mathcal{L%
}_{I_{G}^{\sim }}\left[ 
\begin{array}{c}
u \\ 
y%
\end{array}%
\right] \right\Vert _{2}=\left\Vert \mathcal{P}_{\mathcal{H}_{2}^{\bot }}%
\mathcal{L}_{I_{G}^{\sim }}\left[ 
\begin{array}{c}
u \\ 
y%
\end{array}%
\right] \right\Vert _{2}
\end{equation*}%
that marks the difference between $\left\Vert r_{\mathcal{I}_{G}}\right\Vert
_{2}$ and $\left\Vert r_{0}\right\Vert _{2}.$ Note that, on the one hand, 
\begin{equation}
\mathcal{L}_{I_{G}}\mathcal{P}_{\mathcal{H}_{2}^{\bot }}\mathcal{L}%
_{I_{G}^{\sim }}\left[ 
\begin{array}{c}
u \\ 
y%
\end{array}%
\right] \in \mathcal{L}_{2\text{ }}\text{and }\mathcal{L}_{K_{G}}\mathcal{L}%
_{I_{G}}=0.  \label{eq3-44}
\end{equation}%
Moreover, 
\begin{equation*}
\forall \left[ 
\begin{array}{c}
u \\ 
y%
\end{array}%
\right] \in \mathcal{I}_{G},\exists v\in \mathcal{H}_{2},\text{ s.t. }\left[ 
\begin{array}{c}
u \\ 
y%
\end{array}%
\right] =\mathcal{L}_{I_{G}}v,
\end{equation*}%
which leads to 
\begin{equation*}
\mathcal{L}_{I_{G}}\mathcal{P}_{\mathcal{H}_{2}^{\bot }}\mathcal{L}%
_{I_{G}^{\sim }}\left[ 
\begin{array}{c}
u \\ 
y%
\end{array}%
\right] =\mathcal{L}_{I_{G}}\mathcal{P}_{\mathcal{H}_{2}^{\bot }}\mathcal{L}%
_{I_{G}^{\sim }}\mathcal{L}_{I_{G}}v=0.
\end{equation*}%
On the other hand, for $\left[ 
\begin{array}{c}
u \\ 
y%
\end{array}%
\right] \in \mathcal{I}_{G}^{\bot }\subset \mathcal{H}_{2},$ it is possible
that 
\begin{equation}
\mathcal{L}_{I_{G}}\mathcal{P}_{\mathcal{H}_{2}^{\bot }}\mathcal{L}%
_{I_{G}^{\sim }}\left[ 
\begin{array}{c}
u \\ 
y%
\end{array}%
\right] \neq 0.  \label{eq3-45}
\end{equation}%
Relations (\ref{eq3-44})-(\ref{eq3-45}) showcase that

\begin{itemize}
\item changes caused by $\left[ 
\begin{array}{c}
u \\ 
y%
\end{array}%
\right] \in \mathcal{I}_{G}^{\bot }\subset \mathcal{H}_{2}$ and satisfying (%
\ref{eq3-45}) cannot be detected using an observer-based residual generator,
since they do not lead to any variation in $r_{0};$

\item in against, it is possible to detect these changes using $r_{\mathcal{I%
}_{G}},$ as described by (\ref{eq3-4}).
\end{itemize}

In the context of fault detection, it can thus be claimed that the residual
signal $r_{\mathcal{I}_{G}}$ is more sensitive to faults than $r_{0}.$

\bigskip

Next, we address the interpretation of $\left\Vert \mathcal{H}_{I_{G}^{\sim
}}\left[ 
\begin{array}{c}
u \\ 
y%
\end{array}%
\right] \right\Vert _{2}=\left\Vert \mathcal{P}_{\mathcal{H}_{2}^{\bot }}%
\mathcal{L}_{I_{G}^{\sim }}\left[ 
\begin{array}{c}
u \\ 
y%
\end{array}%
\right] \right\Vert _{2}$ from the fault detection aspect. To this end,
consider the relation%
\begin{equation*}
\left\Vert \mathcal{H}_{I_{G}^{\sim }}\left[ 
\begin{array}{c}
u \\ 
y%
\end{array}%
\right] \right\Vert _{2}=\left\Vert \mathcal{H}_{I_{G}^{\sim }}^{\ast
}\varsigma _{\mathcal{H}}\right\Vert _{2},\varsigma _{\mathcal{H}}\in 
\mathcal{H}_{2}^{\bot }
\end{equation*}%
where $\mathcal{H}_{I_{G}^{\sim }}^{\ast }$ is the adjoint of $\mathcal{H}%
_{I_{G}^{\sim }}.$ As demonstrated in \cite{Francis87}, based on (\ref%
{eq3-46}) $\mathcal{H}_{I_{G}^{\sim }}^{\ast }$ can be written as 
\begin{equation*}
\mathcal{H}_{I_{G}^{\sim }}^{\ast }:\mathcal{H}_{2}^{\bot }\rightarrow 
\mathcal{H}_{2},\mathcal{H}_{I_{G}^{\sim }}^{\ast }=\left( \Psi _{o}\Psi
_{c}\right) ^{\ast }=\Psi _{c}^{\ast }\Psi _{o}^{\ast },
\end{equation*}%
and correspondingly $\mathcal{H}_{I_{G}^{\sim }}^{\ast }\varsigma _{\mathcal{%
H}}$ yields%
\begin{align}
x_{0}& =\Psi _{o}^{\ast }\varsigma _{\mathcal{H}}=\sum\limits_{k=-\infty
}^{0}\left( \bar{A}^{T}\right) ^{k}\bar{C}^{T}\varsigma _{\mathcal{H}}(k),
\label{eq3-47a} \\
\left[ 
\begin{array}{c}
u(k) \\ 
y(k)%
\end{array}%
\right] & =\left( \Psi _{c}^{\ast }x_{0}\right) \left( k\right) =-\bar{B}%
^{T}\left( \bar{A}^{T}\right) ^{k}x_{0},k\in \left[ 0,\infty \right) .
\label{eq3-47b}
\end{align}%
In this regard, $\left\Vert \mathcal{P}_{\mathcal{H}_{2}^{\bot }}\mathcal{L}%
_{I_{G}^{\sim }}\left[ 
\begin{array}{c}
u \\ 
y%
\end{array}%
\right] \right\Vert _{2}$ can be interpreted as the influence of deviations
from the (nominal) image subspace in the past (i.e. over the time interval $%
\left( -\infty ,0\right] $), which affects the dynamics of the residual
generator $\left( \mathcal{I}-\mathcal{P}_{\mathcal{I}_{G}}\right) \left[ 
\begin{array}{c}
u \\ 
y%
\end{array}%
\right] $ in form of the response to the corresponding changes in initial
condition (i.e. via $x_{0}$).

\bigskip

Summarising the discussion on relation (\ref{eq3-4}) and the interpretation
of $\left\Vert \mathcal{P}_{\mathcal{H}_{2}^{\bot }}\mathcal{L}_{I_{G}^{\sim
}}\left[ 
\begin{array}{c}
u \\ 
y%
\end{array}%
\right] \right\Vert _{2}$, it can be concluded that the projection-based
residual generator $\left( \mathcal{I}-\mathcal{P}_{\mathcal{I}_{G}}\right) %
\left[ 
\begin{array}{c}
u \\ 
y%
\end{array}%
\right] $ is not only efficient in detecting existing faults (i.e. over the
time interval $\left[ 0,\infty \right) $) like the observer-based residual
generator $\mathcal{L}_{K_{G}}\left[ 
\begin{array}{c}
u \\ 
y%
\end{array}%
\right] ,$ but also more capable of detecting faults in the past (i.e. over
the time interval $\left( -\infty ,0\right] $).

\subsection{Threshold setting\label{subsection3-3}}

We now consider threshold setting issues for systems with model uncertainty
described by%
\begin{align}
G& =NM^{-1}=\left( N_{0}+\Delta _{N}\right) \left( M_{0}+\Delta _{M}\right)
^{-1},I_{G}=\left[ 
\begin{array}{c}
M \\ 
N%
\end{array}%
\right] =\left[ 
\begin{array}{c}
M_{0}+\Delta _{M} \\ 
N_{0}+\Delta _{N}%
\end{array}%
\right] =I_{G_{0}}+\Delta _{I},  \label{eq3-5} \\
I_{G_{0}}& =\left[ 
\begin{array}{c}
M_{0} \\ 
N_{0}%
\end{array}%
\right] ,\Delta _{I}=\left[ 
\begin{array}{c}
\Delta _{M} \\ 
\Delta _{N}%
\end{array}%
\right]  \notag
\end{align}%
with normalised SIR $I_{G}.$ It is assumed that 
\begin{equation}
\sup \left\Vert \Delta _{I}\right\Vert _{\infty }=\delta _{\Delta _{I}}<1,
\label{eq3-5a}
\end{equation}

\begin{Rem}
Hereafter, notation $G_{0}$ is adopted for the nominal system transfer
matrix, i.e. uncertainty- and fault-free system dynamics.
\end{Rem}

On the assumption of (\ref{eq3-5a}), threshold is set to prevent false
alarms \cite{Ding2008}, namely during fault-free operations 
\begin{equation}
J=\left\Vert r_{\mathcal{I}_{G}}\right\Vert _{2}\leq J_{th}.  \label{eq3-6}
\end{equation}%
Recall that for some $\Delta _{I},$%
\begin{equation*}
\mathcal{I}_{G}=\left\{ \left[ 
\begin{array}{c}
u \\ 
y%
\end{array}%
\right] :\left[ 
\begin{array}{c}
u \\ 
y%
\end{array}%
\right] =\left[ 
\begin{array}{c}
M \\ 
N%
\end{array}%
\right] v,v\in \mathcal{H}_{2}\hspace{-2pt}\right\}
\end{equation*}%
defines an image subspace that is obviously different from $\mathcal{I}%
_{G_{0}},$%
\begin{equation*}
\mathcal{I}_{G_{0}}=\left\{ \left[ 
\begin{array}{c}
u \\ 
y%
\end{array}%
\right] :\left[ 
\begin{array}{c}
u \\ 
y%
\end{array}%
\right] =\left[ 
\begin{array}{c}
M_{0} \\ 
N_{0}%
\end{array}%
\right] v,v\in \mathcal{H}_{2}\hspace{-2pt}\right\} .
\end{equation*}%
Defining furthermore the following subspace in $\mathcal{H}_{2},$%
\begin{equation}
\mathcal{I}_{G,\delta }:=\left\{ \mathcal{I}_{G}:\left\Vert \Delta
_{I}\right\Vert _{\infty }=\left\Vert I_{G}-\hspace{-2pt}I_{G_{0}}\right%
\Vert _{\infty }\leq \delta _{\Delta _{I}}\right\} ,\hspace{-3pt}
\label{eq3-10}
\end{equation}%
threshold setting problem (\ref{eq3-6}) can be equivalently written as%
\begin{equation}
J_{th}=\sup_{\left\Vert \Delta _{I}\right\Vert _{\infty }\leq \delta
_{\Delta _{I}}}J=\sup_{\mathcal{I}_{G}\in \mathcal{I}_{G,\delta }}\sup_{%
\left[ 
\begin{array}{c}
u \\ 
y%
\end{array}%
\right] \in \mathcal{I}_{G}}dist\left( \left[ 
\begin{array}{c}
u \\ 
y%
\end{array}%
\right] ,\mathcal{I}_{G_{0}}\right) .  \label{eq3-6a}
\end{equation}%
In this regard, it becomes apparent that definition and computation of a
metric to measure the difference between two image subspaces, $I_{G}$ and $%
\hspace{-2pt}I_{G_{0}},$ are helpful for solving the threshold setting
problem (\ref{eq3-6}). To this end, we adopt the concept of gap metric
defined in (\ref{eq2-15a})-(\ref{eq2-15}) that is widely applied in robust
control theory \cite{Vinnicombe-book,Feintuch_book}.

\begin{Def}
Let 
\begin{equation}
\mathcal{I}_{G_{i}}=\left\{ \left[ 
\begin{array}{c}
u_{i} \\ 
y_{i}%
\end{array}%
\right] :\left[ 
\begin{array}{c}
u_{i} \\ 
y_{i}%
\end{array}%
\right] =\left[ 
\begin{array}{c}
M_{i} \\ 
N_{i}%
\end{array}%
\right] v,v\in \mathcal{H}_{2}\hspace{-2pt}\right\} ,i=1,2,  \notag
\end{equation}%
be two image subspaces. The directed gap $\vec{\delta}_{\mathcal{I}}\left( 
\mathcal{I}_{G_{1}},\mathcal{I}_{G_{2}}\right) $ and gap metric $\delta _{%
\mathcal{I}}\left( \mathcal{I}_{G_{1}},\mathcal{I}_{G_{2}}\right) $ are
respectively defined%
\begin{eqnarray}
\vec{\delta}_{\mathcal{I}}\left( \mathcal{I}_{G_{1}},\mathcal{I}%
_{G_{2}}\right) &=&\sup_{\left[ 
\begin{array}{c}
u_{1} \\ 
y_{1}%
\end{array}%
\right] \in \mathcal{I}_{G_{1}}}\inf_{\left[ 
\begin{array}{c}
u_{2} \\ 
y_{2}%
\end{array}%
\right] \in \mathcal{I}_{G_{2}}}\frac{\left\Vert \left[ 
\begin{array}{c}
u_{1} \\ 
y_{1}%
\end{array}%
\right] -\left[ 
\begin{array}{c}
u_{2} \\ 
y_{2}%
\end{array}%
\right] \right\Vert _{2}}{\left\Vert \left[ 
\begin{array}{c}
u_{1} \\ 
y_{1}%
\end{array}%
\right] \right\Vert _{2}}  \label{eq3-7} \\
&=&\sup_{\left[ 
\begin{array}{c}
u_{1} \\ 
y_{1}%
\end{array}%
\right] \in \mathcal{I}_{G_{1}},\left\Vert \left[ 
\begin{array}{c}
u_{1} \\ 
y_{1}%
\end{array}%
\right] \right\Vert _{2}=1}dist\left( \left[ 
\begin{array}{c}
u_{1} \\ 
y_{1}%
\end{array}%
\right] ,\mathcal{I}_{G_{2}}\right) ,  \label{eq3-7b} \\
\delta _{\mathcal{I}}\left( \mathcal{I}_{G_{1}},\mathcal{I}_{G_{2}}\right)
&=&\max \left\{ \vec{\delta}_{\mathcal{I}}\left( \mathcal{I}_{G_{1}},%
\mathcal{I}_{G_{2}}\right) ,\vec{\delta}_{\mathcal{I}}\left( \mathcal{I}%
_{G_{2}},\mathcal{I}_{G_{1}}\right) \right\} .
\end{eqnarray}
\end{Def}

The computation of $\delta _{\mathcal{I}}\left( \mathcal{I}_{G_{i}},\mathcal{%
I}_{G_{j}}\right) $ was intensively investigated, and one of the key results
is that 
\begin{equation}
\vec{\delta}_{\mathcal{I}}\left( \mathcal{I}_{G_{i}},\mathcal{I}%
_{G_{j}}\right) =\inf_{Q\in \mathcal{H}_{\infty }}\left\Vert \left[ 
\begin{array}{c}
M_{i} \\ 
N_{i}%
\end{array}%
\right] -\left[ 
\begin{array}{c}
M_{j} \\ 
N_{j}%
\end{array}%
\right] Q\right\Vert _{\infty },  \label{eq3-8}
\end{equation}%
i.e. the gap metric can be calculated by solving the model matching problem
(MMP) on the right-hand side of (\ref{eq3-8}) \cite%
{Georgiou88,Georgiou&Smith90,Vinnicombe-book}. Below, we briefly delineate (%
\ref{eq3-8}) using the result given in Lemma \ref{Cor3-1}. It follows from (%
\ref{eq2-15}) that 
\begin{equation*}
\vec{\delta}_{\mathcal{I}}\left( \mathcal{I}_{G_{i}},\mathcal{I}%
_{G_{j}}\right) =\left\Vert \left( \mathcal{I}-\mathcal{P}_{G_{j}}\right) 
\mathcal{P}_{G_{i}}\right\Vert .
\end{equation*}%
Since, according to Lemma \ref{Cor3-1}, 
\begin{eqnarray*}
\mathcal{I}-\mathcal{P}_{\mathcal{I}_{G_{j}}} &=&\mathcal{I}-\mathcal{L}%
_{I_{G_{j}}}\mathcal{L}_{I_{G_{j}}^{\sim }}+\mathcal{L}_{I_{G_{j}}}\mathcal{P%
}_{\mathcal{H}_{2}^{\mathcal{\bot }}}\mathcal{L}_{I_{G_{j}}^{\sim }}=%
\mathcal{L}_{K_{G_{j}}^{\sim }}\mathcal{L}_{K_{G_{j}}}+\mathcal{L}%
_{I_{G_{j}}}\mathcal{P}_{\mathcal{H}_{2}^{\mathcal{\bot }}}\mathcal{L}%
_{I_{G_{j}}^{\sim }} \\
&=&\left[ 
\begin{array}{cc}
\mathcal{L}_{K_{G_{j}}^{\sim }} & \mathcal{L}_{I_{G_{j}}}%
\end{array}%
\right] \left[ 
\begin{array}{c}
\mathcal{L}_{K_{G_{j}}} \\ 
\mathcal{P}_{\mathcal{H}_{2}^{\mathcal{\bot }}}\mathcal{L}_{I_{G_{j}}^{\sim
}}%
\end{array}%
\right] ,
\end{eqnarray*}%
where $K_{G_{j}}$ is the normalised SKR of $G_{j},$ it yields%
\begin{equation*}
\vec{\delta}\left( \mathcal{I}_{G_{i}},\mathcal{I}_{G_{j}}\right)
=\left\Vert \left[ 
\begin{array}{cc}
\mathcal{L}_{K_{G_{j}}^{\sim }} & \mathcal{L}_{I_{G_{j}}}%
\end{array}%
\right] \left[ 
\begin{array}{c}
\mathcal{L}_{K_{G_{j}}} \\ 
\mathcal{P}_{\mathcal{H}_{2}^{\mathcal{\bot }}}\mathcal{L}_{I_{G_{j}}^{\sim
}}%
\end{array}%
\right] \mathcal{L}_{I_{G_{i}}}\mathcal{L}_{I_{G_{i}}}^{\ast }\right\Vert .
\end{equation*}%
Noting further 
\begin{equation*}
\mathcal{L}_{I_{G_{i}}}^{\ast }\mathcal{L}_{I_{G_{i}}}=\mathcal{I},\left[ 
\begin{array}{cc}
\mathcal{L}_{K_{G_{j}}^{\sim }} & \mathcal{L}_{I_{G_{j}}}%
\end{array}%
\right] ^{\ast }\left[ 
\begin{array}{cc}
\mathcal{L}_{K_{G_{j}}^{\sim }} & \mathcal{L}_{I_{G_{j}}}%
\end{array}%
\right] =\mathcal{I},
\end{equation*}%
it turns out%
\begin{equation}
\vec{\delta}\left( \mathcal{I}_{G_{i}},\mathcal{I}_{G_{j}}\right)
=\left\Vert \left[ 
\begin{array}{c}
\mathcal{L}_{K_{G_{i}}}\mathcal{L}_{I_{G_{i}}} \\ 
\mathcal{P}_{\mathcal{H}_{2}^{\mathcal{\bot }}}\mathcal{L}_{I_{G_{j}}^{\sim
}}\mathcal{L}_{I_{G_{i}}}%
\end{array}%
\right] \right\Vert .  \label{eq3-33}
\end{equation}%
In \cite{Georgiou&Smith90}, it is proved that 
\begin{equation}
\left\Vert \left[ 
\begin{array}{c}
\mathcal{L}_{K_{G_{i}}}\mathcal{L}_{I_{G_{i}}} \\ 
\mathcal{P}_{\mathcal{H}_{2}^{\mathcal{\bot }}}\mathcal{L}_{I_{G_{j}}^{\sim
}}\mathcal{L}_{I_{G_{i}}}%
\end{array}%
\right] \right\Vert =\inf_{Q\in \mathcal{H}_{\infty }}\left\Vert
I_{G_{i}}-I_{G_{j}}Q\right\Vert _{\infty }.  \label{eq3-34}
\end{equation}%
Concerning the computation of MMP (\ref{eq3-8}), there exist
well-established algorithms, see e.g. \cite{Francis87}.

\bigskip

Now, we are in the position to solve the threshold setting problem defined
in (\ref{eq3-6}) and its re-formulation (\ref{eq3-6a}). Since 
\begin{equation*}
\forall \left[ 
\begin{array}{c}
u \\ 
y%
\end{array}%
\right] \in \mathcal{I}_{G},dist\left( \left[ 
\begin{array}{c}
u \\ 
y%
\end{array}%
\right] ,\mathcal{I}_{G_{0}}\right) \leq \vec{\delta}_{\mathcal{I}}\left( 
\mathcal{I}_{G},\mathcal{I}_{G_{0}}\right) \left\Vert \left[ 
\begin{array}{c}
u \\ 
y%
\end{array}%
\right] \right\Vert _{2},
\end{equation*}%
it holds%
\begin{equation}
\sup_{\left[ 
\begin{array}{c}
u \\ 
y%
\end{array}%
\right] \in \mathcal{I}_{G}}dist\left( \left[ 
\begin{array}{c}
u \\ 
y%
\end{array}%
\right] ,\mathcal{I}_{G_{0}}\right) =\vec{\delta}_{\mathcal{I}}\left( 
\mathcal{I}_{G},\mathcal{I}_{G_{0}}\right) \left\Vert \left[ 
\begin{array}{c}
u \\ 
y%
\end{array}%
\right] \right\Vert _{2}.  \label{eq3-9}
\end{equation}%
Using the well-established result given in \cite%
{Georgiou&Smith90,Vinnicombe-book} that for $0<\delta _{\Delta _{I}}<1,$%
\begin{align}
\mathcal{I}_{G,\delta }& =\left\{ \mathcal{I}_{G}:\left\Vert \Delta
_{I}\right\Vert _{\infty }=\left\Vert I_{G}-\hspace{-2pt}I_{G_{0}}\right%
\Vert _{\infty }\leq \delta _{\Delta _{I}}\right\}  \notag \\
& =\left\{ \mathcal{I}_{G}:\delta _{\mathcal{I}}\left( \mathcal{I}_{G},%
\mathcal{I}_{G_{0}}\right) =\vec{\delta}_{\mathcal{I}}\left( \mathcal{I}_{G},%
\mathcal{I}_{G_{0}}\right) \leq \delta _{\Delta _{I}}\right\} ,
\end{align}%
we further have%
\begin{gather}
\sup_{\left\Vert \Delta _{I}\right\Vert _{\infty }\leq \delta _{\Delta
_{I}}}J=\sup_{\mathcal{I}_{G}\in \mathcal{I}_{G,\delta }}\sup_{\left[ 
\begin{array}{c}
u \\ 
y%
\end{array}%
\right] \in \mathcal{I}_{G}}dist\left( \left[ 
\begin{array}{c}
u \\ 
y%
\end{array}%
\right] ,\mathcal{I}_{G_{0}}\right)  \notag \\
=\sup_{\left\Vert \Delta _{I}\right\Vert _{\infty }\leq \delta _{\Delta
_{I}}}\vec{\delta}_{\mathcal{I}}\left( \mathcal{I}_{G},\mathcal{I}%
_{G_{0}}\right) \left\Vert \left[ 
\begin{array}{c}
u \\ 
y%
\end{array}%
\right] \right\Vert _{2}=\delta _{\Delta _{I}}\left\Vert \left[ 
\begin{array}{c}
u \\ 
y%
\end{array}%
\right] \right\Vert _{2}.  \label{eq3-11}
\end{gather}%
If $\delta _{\Delta _{I}}\left\Vert \left[ 
\begin{array}{c}
u \\ 
y%
\end{array}%
\right] \right\Vert _{2}$ serves as a threshold, the detection performance
is determined by the ratio $\frac{J_{N}}{\delta _{\Delta _{I}}}$ with the
normalised residual $J_{N},$ 
\begin{equation*}
J_{N}=\frac{\left\Vert r_{\mathcal{I}_{G}}\right\Vert _{2}}{\left\Vert \left[
\begin{array}{c}
u \\ 
y%
\end{array}%
\right] \right\Vert _{2}},
\end{equation*}%
since it holds 
\begin{equation*}
J-\delta _{\Delta _{I}}\left\Vert \left[ 
\begin{array}{c}
u \\ 
y%
\end{array}%
\right] \right\Vert _{2}>0\Longleftrightarrow \frac{J_{N}}{\delta _{\Delta
_{I}}}>1.
\end{equation*}%
It becomes clear that if $\frac{J_{N}}{\delta _{\Delta _{I}}}$ is larger
than but close to one, false alarms can be easily triggered, for instance,
by noises in the system. In order to enhance the detection robustness,
further improvement of threshold setting is made. Consider the relation%
\begin{equation*}
\forall \mathcal{I}_{G}\subset \mathcal{I}_{G,\delta },\left\Vert r_{%
\mathcal{I}_{G}}\right\Vert _{2}^{2}\leq \delta _{\Delta _{I}}^{2}\left\Vert %
\left[ 
\begin{array}{c}
u \\ 
y%
\end{array}%
\right] \right\Vert _{2}^{2}.
\end{equation*}%
It turns out that $\forall \mathcal{I}_{G}\subset \mathcal{I}_{G,\delta },$%
\begin{align}
\left\Vert r_{\mathcal{I}_{G}}\right\Vert _{2}^{2}& \leq \delta _{\Delta
_{I}}^{2}\left( \left\Vert r_{\mathcal{I}_{G}}\right\Vert
_{2}^{2}+\left\Vert \mathcal{P}_{\mathcal{I}_{G}}\left[ 
\begin{array}{c}
u \\ 
y%
\end{array}%
\right] \right\Vert _{2}^{2}\right) \Longleftrightarrow  \notag \\
\left\Vert r_{\mathcal{I}_{G}}\right\Vert _{2}^{2}& \leq \frac{\delta
_{\Delta _{I}}^{2}}{1-\delta _{\Delta _{I}}^{2}}\left\Vert \mathcal{P}_{%
\mathcal{I}_{G}}\left[ 
\begin{array}{c}
u \\ 
y%
\end{array}%
\right] \right\Vert _{2}^{2}.  \label{eq3-12}
\end{align}%
Notice that 
\begin{equation*}
\mathcal{P}_{\mathcal{I}_{G}}\left[ 
\begin{array}{c}
u \\ 
y%
\end{array}%
\right] \in \mathcal{I}_{G_{0}}\subset \mathcal{I}_{G,\delta }.
\end{equation*}%
As a result, we have

\begin{Theo}
\label{Theo3-1}Given the model (\ref{eq2-1}) with model uncertainty
satisfying (\ref{eq3-5}) and (\ref{eq3-5a}), and suppose that
projection-based residual $r_{\mathcal{I}_{G}}$ is used for the detection
purpose, then the corresponding threshold defined by (\ref{eq3-6}) is given
by 
\begin{align}
J_{th}& =\frac{\delta _{\Delta _{I}}}{\sqrt{1-\delta _{\Delta _{I}}^{2}}}%
\left\Vert \mathcal{P}_{\mathcal{I}_{G}}\left[ 
\begin{array}{c}
u \\ 
y%
\end{array}%
\right] \right\Vert _{2}  \label{eq3-13a} \\
& =\frac{\delta _{\Delta _{I}}}{\sqrt{1-\delta _{\Delta _{I}}^{2}}}\left(
\left\Vert \left[ 
\begin{array}{c}
u \\ 
y%
\end{array}%
\right] \right\Vert _{2}^{2}-\left\Vert r_{\mathcal{I}_{G}}\right\Vert
_{2}^{2}\right) ^{1/2}.  \label{eq3-13b}
\end{align}

\begin{proof}
Since the inequality (\ref{eq3-12}) holds for all $\mathcal{I}_{G}\in 
\mathcal{I}_{G,\delta },$ it is straightforward that 
\begin{equation*}
J_{th}=\sup_{\left\Vert \Delta _{I}\right\Vert _{\infty }\leq \delta
_{\Delta _{I}}}\left\Vert r_{\mathcal{I}_{G}}\right\Vert _{2}=\frac{\delta
_{\Delta _{I}}}{\sqrt{1-\delta _{\Delta _{I}}^{2}}}\left\Vert \mathcal{P}_{%
\mathcal{I}_{G}}\left[ 
\begin{array}{c}
u \\ 
y%
\end{array}%
\right] \right\Vert _{2}.
\end{equation*}%
The threshold setting (\ref{eq3-13b}) immediately follows from the relation 
\begin{align*}
\left\Vert \left[ 
\begin{array}{c}
u \\ 
y%
\end{array}%
\right] \right\Vert _{2}^{2}& =\left\Vert \mathcal{P}_{\mathcal{I}_{G}^{\bot
}}\left[ 
\begin{array}{c}
u \\ 
y%
\end{array}%
\right] \right\Vert _{2}^{2}+\left\Vert \mathcal{P}_{\mathcal{I}_{G}}\left[ 
\begin{array}{c}
u \\ 
y%
\end{array}%
\right] \right\Vert _{2}^{2} \\
& =\left\Vert r_{\mathcal{I}_{G}}\right\Vert _{2}^{2}+\left\Vert \mathcal{P}%
_{\mathcal{I}_{G}}\left[ 
\begin{array}{c}
u \\ 
y%
\end{array}%
\right] \right\Vert _{2}^{2}.
\end{align*}
\end{proof}
\end{Theo}

\bigskip

The threshold $J_{th}$ given in the above theorem is driven by the online
measurement and thus called adaptive threshold \cite{Ding2008}. More
importantly, the adaptive threshold (\ref{eq3-13a}) (equivalently (\ref%
{eq3-13b})) is more sensitive to faults, since, as indicated by (\ref%
{eq3-13b}), the threshold will decrease, as a fault occurs in the system and
thus the residual increases. This observation motivates us to deepen our
understanding of this important aspect. To this end, let 
\begin{equation}
J_{th,N}:=\sup_{\left\Vert \Delta _{I}\right\Vert _{\infty }\leq \delta
_{\Delta _{I}}}J_{N}=\frac{\delta _{\Delta _{I}}}{\sqrt{1-\delta _{\Delta
_{I}}^{2}}}\left( 1-\frac{\left\Vert r_{\mathcal{I}_{G}}\right\Vert _{2}^{2}%
}{\left\Vert \left[ 
\begin{array}{c}
u \\ 
y%
\end{array}%
\right] \right\Vert _{2}^{2}}\right) ^{1/2}  \label{eq3-25}
\end{equation}%
be the normalised threshold, and 
\begin{equation}
\mathcal{F}=\left\{ \left[ 
\begin{array}{c}
u \\ 
y%
\end{array}%
\right] \in \mathcal{H}_{2},\left\Vert r_{\mathcal{I}_{G}}\right\Vert
_{2}>\delta _{\Delta _{I}}\left\Vert \left[ 
\begin{array}{c}
u \\ 
y%
\end{array}%
\right] \right\Vert _{2}\right\} .  \label{eq3-20}
\end{equation}%
While $\mathcal{I}_{G,\delta }$ denotes the set of all data collected during
fault-free operations, $\mathcal{F}$ can be interpreted as data set
corresponding to faulty operations.

\begin{Theo}
\label{Theo3-2}Given the model (\ref{eq2-1}) with model uncertainty
satisfying (\ref{eq3-5}) and (\ref{eq3-5a}), and suppose that
projection-based residual $r_{\mathcal{I}_{G}}$ is generated, then we have%
\begin{eqnarray}
\forall \left[ 
\begin{array}{c}
u \\ 
y%
\end{array}%
\right] &\in &\mathcal{I}_{G}\subset \mathcal{I}_{G,\delta },J_{th,N}\geq
\delta _{\Delta _{I}},  \label{eq3-21} \\
\forall \left[ 
\begin{array}{c}
u \\ 
y%
\end{array}%
\right] &\in &\mathcal{F},J_{th,N}<\delta _{\Delta _{I}},  \label{eq3-22} \\
\forall \left[ 
\begin{array}{c}
u \\ 
y%
\end{array}%
\right] &\in &\mathcal{F},\frac{J_{N}}{J_{th,N}}>\frac{J_{N}}{\delta
_{\Delta _{I}}}>1.  \label{eq3-26}
\end{eqnarray}
\end{Theo}

\begin{proof}
Relation (\ref{eq3-21}) follows directly from%
\begin{equation*}
\forall \left[ 
\begin{array}{c}
u \\ 
y%
\end{array}%
\right] \in \mathcal{I}_{G}\subset \mathcal{I}_{G,\delta },1-\frac{%
\left\Vert r_{\mathcal{I}_{G}}\right\Vert _{2}^{2}}{\left\Vert \left[ 
\begin{array}{c}
u \\ 
y%
\end{array}%
\right] \right\Vert _{2}^{2}}\geq 1-\delta _{\Delta _{I}}^{2}.
\end{equation*}%
For all $\left[ 
\begin{array}{c}
u \\ 
y%
\end{array}%
\right] $ belonging to $\mathcal{F},$ definition (\ref{eq3-20}) results in%
\begin{gather*}
1-\frac{\left\Vert r_{\mathcal{I}_{G}}\right\Vert _{2}^{2}}{\left\Vert \left[
\begin{array}{c}
u \\ 
y%
\end{array}%
\right] \right\Vert _{2}^{2}}<1-\delta _{\Delta _{I}}^{2}\Longleftrightarrow
J_{th,N}<\delta _{\Delta _{I}}, \\
\frac{J_{N}}{J_{th,N}}=\frac{J_{N}}{\delta _{\Delta _{I}}}\left( \frac{%
1-\delta _{\Delta _{I}}^{2}}{1-\frac{\left\Vert r_{\mathcal{I}%
_{G}}\right\Vert _{2}^{2}}{\left\Vert \left[ 
\begin{array}{c}
u \\ 
y%
\end{array}%
\right] \right\Vert _{2}^{2}}}\right) ^{1/2}>\frac{J_{N}}{\delta _{\Delta
_{I}}}>1.
\end{gather*}%
The theorem is thus proved. 
\end{proof}

\bigskip

From (\ref{eq3-21})-(\ref{eq3-22}) in Theorem \ref{Theo3-2} it becomes
apparent that

\begin{itemize}
\item during fault-free operations, the normalised threshold $J_{th,N}$ is
higher than or equal to the upper-bound of the uncertainty $\delta _{\Delta
_{I}}$, and

\item it decreases, as a fault defined in sense of (\ref{eq3-20}) occurs in
the system, and becomes smaller than $\delta _{\Delta _{I}}.$

\item More importantly, (\ref{eq3-26}) indicates that all these faults can
be detected with enhanced robustness thanks to the larger ratio $\frac{J_{N}%
}{J_{th,N}}.$
\end{itemize}

These properties reveal that the proposed projection-based threshold setting
is of higher robustness, which is useful to reduce false alarms caused by
noises. In Section 6, it will be demonstrated that the adaptive threshold
setting (\ref{eq3-13a}) or the normalised threshold (\ref{eq3-25}) results
in better detection performance than the threshold setting adopted in the
observer-based detection schemes \cite{LD-Automatica-2020,Li-IEEETCST2020}.

\section{Project-based fault detection in feedback control systems}

Due to wide integration of feedback control loops in automatic control
systems, fault detection in feedback control loops draws special attention
in research of model-based fault detection in dynamic systems. In this
section, we describe two fault detection schemes for feedback control
systems, which are developed respectively on the basis of (i) projection
onto the image subspace and (ii) projection onto the closed-loop image
subspace of the feedback control system under consideration. This work also
serves as examples for demonstrating the application of the basic
projection-based fault detection scheme presented in the previous section.

\subsection{System models and closed-loop dynamics}

Consider a feedback control loop 
\begin{equation}
y(z)=G(z)u(z),u(z)=K(z)y(z)+v(z)  \label{eq4-1}
\end{equation}%
with a dynamic output controller $K$ and reference vector $v.$ It is a
well-known result that all stabilising controllers can be parameterised by%
\begin{align}
K(z)& =-\left( X_{0}(z)-Q(z)\hat{N}_{0}(z)\right) ^{-1}\left( Y_{0}(z)+Q(z)%
\hat{M}_{0}(z)\right)  \label{eq4-2a} \\
& =-\left( \hat{Y}_{0}(z)+M_{0}(z)Q(z)\right) \left( \hat{X}%
_{0}(z)-N_{0}(z)Q(z)\right) ^{-1},  \label{eq4-2b}
\end{align}%
where $Q(z)\in \mathcal{RH}_{\infty }$ is the so-called parameter system,
and $\left( \hat{M}_{0},\hat{N}_{0}\right) $ and $\left( M_{0},N_{0}\right) $
are the coprime pairs of the nominal transfer function $G_{0},$ which,
together with $\left( \hat{X}_{0},\hat{Y}_{0}\right) $ and $\left(
X_{0},Y_{0}\right) ,$ are given in (\ref{eq2-4a})-(\ref{eq2-6b}) \cite%
{Zhou98}. Without loss of generality, it is assumed that $F=F_{0},L=L_{0}$.
Moreover, the extended form of Bezout identity (\ref{eq2-5}) holds:%
\begin{gather}
\left[ 
\begin{array}{cc}
M_{0} & \text{ }U_{0} \\ 
N_{0} & \text{ }V_{0}%
\end{array}%
\right] \left[ 
\begin{array}{cc}
\hat{V}_{0} & \text{ }-\hat{U}_{0} \\ 
-\hat{N}_{0} & \text{ }\hat{M}_{0}%
\end{array}%
\right] =\left[ 
\begin{array}{cc}
I\text{ } & 0\text{ } \\ 
0\text{ } & I\text{ }%
\end{array}%
\right] ,  \label{eq2-5a} \\
\left[ 
\begin{array}{cc}
\hat{V}_{0} & \hat{U}_{0}%
\end{array}%
\right] =\left[ 
\begin{array}{cc}
X_{0}-Q\hat{N}_{0} & \text{ }-Y_{0}-Q\hat{M}_{0}%
\end{array}%
\right] ,\left[ 
\begin{array}{c}
U_{0} \\ 
V_{0}%
\end{array}%
\right] =\left[ 
\begin{array}{c}
-\hat{Y}_{0}-M_{0}Q \\ 
\hat{X}_{0}-N_{0}Q%
\end{array}%
\right] .  \notag
\end{gather}%
It turns out%
\begin{gather}
\left[ 
\begin{array}{c}
u \\ 
y%
\end{array}%
\right] =\left[ 
\begin{array}{cc}
I & -K \\ 
-G_{0} & I%
\end{array}%
\right] ^{-1}\left[ 
\begin{array}{c}
I \\ 
0%
\end{array}%
\right] v=\left[ 
\begin{array}{cc}
\hat{V}_{0} & \hat{U}_{0} \\ 
-\hat{N}_{0} & \text{ }\hat{M}_{0}%
\end{array}%
\right] ^{-1}\left[ 
\begin{array}{c}
\hat{V}_{0} \\ 
0%
\end{array}%
\right] v  \notag \\
=\left[ 
\begin{array}{cc}
M_{0} & -U_{0} \\ 
N_{0} & V_{0}%
\end{array}%
\right] \left[ 
\begin{array}{c}
\hat{V}_{0} \\ 
0%
\end{array}%
\right] v=\left[ 
\begin{array}{c}
M_{0} \\ 
N_{0}%
\end{array}%
\right] \hat{v},\hat{v}=\hat{V}_{0}v.  \label{eq4-3}
\end{gather}%
Now, suppose that the plant is corrupted with uncertainty (\ref{eq3-5})-(\ref%
{eq3-5a}),%
\begin{equation}
G=\left( N_{0}+\Delta _{N}\right) \left( M_{0}+\Delta _{M}\right)
^{-1},I_{G}=I_{G_{0}}+\Delta _{I}.  \label{eq4-5}
\end{equation}%
The closed-loop dynamics is governed by 
\begin{equation}
\left[ 
\begin{array}{c}
u \\ 
y%
\end{array}%
\right] =\left[ 
\begin{array}{cc}
I & -K \\ 
-G & I%
\end{array}%
\right] ^{-1}\left[ 
\begin{array}{c}
I \\ 
0%
\end{array}%
\right] v.  \label{eq4-6}
\end{equation}%
It is assumed that 
\begin{equation}
\left\Vert \Delta _{I,c}\right\Vert _{\infty }\leq b<1,\Delta _{I,c}=\left[ 
\begin{array}{cc}
\hat{V}_{0} & \text{ }-\hat{U}_{0} \\ 
-\hat{N}_{0} & \text{ }\hat{M}_{0}%
\end{array}%
\right] \left[ 
\begin{array}{c}
\Delta _{M} \\ 
\Delta _{N}%
\end{array}%
\right] .  \label{eq4-6a}
\end{equation}%
It is known that (\ref{eq4-6a}) is a sufficient condition for the
closed-loop stability \cite{Vinnicombe-book}.

\bigskip

The following two lemmas are useful for our subsequent study.

\begin{Le}
\label{Le4-1}Given feedback control loop (\ref{eq4-6}) with the plant model (%
\ref{eq4-5}) and control law (\ref{eq4-2a})-(\ref{eq4-2b}), it holds%
\begin{equation}
\left[ 
\begin{array}{c}
u \\ 
y%
\end{array}%
\right] =\left[ 
\begin{array}{c}
M \\ 
N%
\end{array}%
\right] \left( I+\Delta _{1}\right) ^{-1}\hat{v},\Delta _{1}=\left[ 
\begin{array}{cc}
\hat{V}_{0} & -\hat{U}_{0}%
\end{array}%
\right] \left[ 
\begin{array}{c}
\Delta _{M} \\ 
\Delta _{N}%
\end{array}%
\right] \in \mathcal{H}_{\infty }.
\end{equation}
\end{Le}

The proof of the above lemma is given in Appendix.

\begin{Le}
\label{Le4-2}\cite{Georgiou&Smith90} Given (\ref{eq4-6a}) and let 
\begin{equation}
\left[ 
\begin{array}{cc}
\hat{V}_{0} & \text{ }-\hat{U}_{0} \\ 
-\hat{N}_{0} & \text{ }\hat{M}_{0}%
\end{array}%
\right] \left[ 
\begin{array}{c}
\Delta _{M} \\ 
\Delta _{N}%
\end{array}%
\right] =\left[ 
\begin{array}{c}
\Delta _{1} \\ 
\Delta _{2}%
\end{array}%
\right] .  \label{eq4-15}
\end{equation}%
Then it holds 
\begin{equation}
\left\Vert \Delta _{2}\left( I+\Delta _{1}\right) ^{-1}\right\Vert _{\infty
}\leq \frac{b}{\sqrt{1-b^{2}}}.  \label{eq4-7}
\end{equation}
\end{Le}

\subsection{Application of the projection-based fault detection scheme \label%
{subsection4-2}}

Based on the closed-loop dynamics, we first present a fault detection scheme
based on the projection onto the system image subspace. Notice that 
\begin{equation*}
\hat{v}=\hat{V}_{0}v,\hat{V}_{0}\in \mathcal{RH}_{\infty },\hat{V}%
_{0}^{-1}\in \mathcal{RL}_{\infty }.
\end{equation*}%
Hence, there exists an invertible $R_{0}(z)\in \mathcal{RH}_{\infty }$ so
that 
\begin{equation*}
\left( \hat{V}_{0}R_{0}\right) ^{\sim }\hat{V}_{0}R_{0}=I\Longrightarrow I_{%
\bar{G}_{0}}=\left[ 
\begin{array}{c}
\bar{M}_{0} \\ 
\bar{N}_{0}%
\end{array}%
\right] =\left[ 
\begin{array}{c}
M_{0} \\ 
N_{0}%
\end{array}%
\right] \hat{V}_{0}R_{0}
\end{equation*}%
is a normalised SIR. Define 
\begin{equation}
\mathcal{I}_{\bar{G}_{0}}=:\left\{ \left[ 
\begin{array}{c}
u \\ 
y%
\end{array}%
\right] :\left[ 
\begin{array}{c}
u \\ 
y%
\end{array}%
\right] =\left[ 
\begin{array}{c}
\bar{M}_{0} \\ 
\bar{N}_{0}%
\end{array}%
\right] v,v\in \mathcal{H}_{2}\right\} .  \label{eq4-9}
\end{equation}%
It follows from the results in Subsection 3.1 that the corresponding
orthogonal projection is%
\begin{equation*}
\mathcal{P}_{\mathcal{I}_{\bar{G}_{0}}}=\mathcal{L}_{I_{\bar{G}_{0}}}%
\mathcal{L}_{I_{\bar{G}_{0}}}^{\ast },p_{\mathcal{I}_{\bar{G}_{0}}}=\mathcal{%
P}_{\mathcal{I}_{\bar{G}_{0}}}\left[ 
\begin{array}{c}
u \\ 
y%
\end{array}%
\right] ,
\end{equation*}%
and the projection-based residual is generated by%
\begin{equation}
r_{\mathcal{I}_{\bar{G}_{0}}}=\left[ 
\begin{array}{c}
u \\ 
y%
\end{array}%
\right] -p_{\mathcal{I}_{\bar{G}_{0}}}=\left( \mathcal{I}-\mathcal{P}_{%
\mathcal{I}_{\bar{G}_{0}}}\right) \left[ 
\begin{array}{c}
u \\ 
y%
\end{array}%
\right] .  \label{eq4-13}
\end{equation}%
Next, we determine the threshold following the procedure introduced in
Subsection \ref{subsection3-3}. Since $\forall \left( \Delta _{M},\Delta
_{N}\right) ,\left( I+\Delta _{1}\right) ,\left( I+\Delta _{1}\right)
^{-1}\in \mathcal{H}_{\infty },$ it follows from Lemma \ref{Le4-1} that%
\begin{equation}
\mathcal{I}_{G}:=\left\{ \left[ 
\begin{array}{c}
u \\ 
y%
\end{array}%
\right] :\left[ 
\begin{array}{c}
u \\ 
y%
\end{array}%
\right] =\left[ 
\begin{array}{c}
M \\ 
N%
\end{array}%
\right] \left( I+\Delta _{1}\right) ^{-1}\hat{V}_{0}R_{0}v,v\in \mathcal{H}%
_{2}\right\}  \label{eq4-8}
\end{equation}%
builds a closed $\mathcal{H}_{2}$ subspace and is different from $\mathcal{I}%
_{\bar{G}_{0}}.$ For our purpose, the difference 
\begin{equation*}
\left[ 
\begin{array}{c}
M \\ 
N%
\end{array}%
\right] \left( I+\Delta _{1}\right) ^{-1}\hat{V}_{0}R_{0}-\left[ 
\begin{array}{c}
M_{0} \\ 
N_{0}%
\end{array}%
\right] \hat{V}_{0}R_{0}
\end{equation*}%
should be specified. By means of the following steps, 
\begin{gather}
\left[ 
\begin{array}{c}
\left( M_{o}+\Delta _{M}\right) \left( I+\Delta _{1}\right) ^{-1} \\ 
\left( N_{o}+\Delta _{N}\right) \left( I+\Delta _{1}\right) ^{-1}%
\end{array}%
\right] -\left[ 
\begin{array}{c}
M_{0} \\ 
N_{0}%
\end{array}%
\right] =\left( \left[ 
\begin{array}{c}
\Delta _{M} \\ 
\Delta _{N}%
\end{array}%
\right] -\left[ 
\begin{array}{c}
M_{0} \\ 
N_{0}%
\end{array}%
\right] \Delta _{1}\right) \left( I+\Delta _{1}\right) ^{-1}  \notag \\
=\left( I-\left[ 
\begin{array}{c}
M_{0} \\ 
N_{0}%
\end{array}%
\right] \left[ 
\begin{array}{cc}
\hat{V}_{0} & -\hat{U}_{0}%
\end{array}%
\right] \right) \left[ 
\begin{array}{c}
\Delta _{M} \\ 
\Delta _{N}%
\end{array}%
\right] \left( I+\Delta _{1}\right) ^{-1}  \notag \\
=\left[ 
\begin{array}{c}
U_{0} \\ 
V_{0}%
\end{array}%
\right] \left[ 
\begin{array}{cc}
-\hat{N}_{0} & \text{ }\hat{M}_{0}%
\end{array}%
\right] \left[ 
\begin{array}{c}
\Delta _{M} \\ 
\Delta _{N}%
\end{array}%
\right] \left( I+\Delta _{1}\right) ^{-1}=\left[ 
\begin{array}{c}
U_{0} \\ 
V_{0}%
\end{array}%
\right] \Delta _{2}\left( I+\Delta _{1}\right) ^{-1},  \label{eq4-20}
\end{gather}%
we have 
\begin{gather*}
\left[ 
\begin{array}{c}
M \\ 
N%
\end{array}%
\right] \left( I+\Delta _{1}\right) ^{-1}\hat{V}_{0}-\left[ 
\begin{array}{c}
M_{0} \\ 
N_{0}%
\end{array}%
\right] \hat{V}_{0}=\left[ 
\begin{array}{c}
U_{0} \\ 
V_{0}%
\end{array}%
\right] \Delta _{2}\left( I+\Delta _{1}\right) ^{-1}\hat{V}_{0} \\
\Longrightarrow \left[ 
\begin{array}{c}
M \\ 
N%
\end{array}%
\right] \left( I+\Delta _{1}\right) ^{-1}\hat{V}_{0}R_{0}-\left[ 
\begin{array}{c}
\bar{M}_{0} \\ 
\bar{N}_{0}%
\end{array}%
\right] =\left[ 
\begin{array}{c}
U_{0} \\ 
V_{0}%
\end{array}%
\right] \Delta _{2}\left( I+\Delta _{1}\right) ^{-1}\hat{V}_{0}R_{0} \\
=:I_{\bar{G}}-I_{\bar{G}_{0}}=\Delta _{I_{\bar{G}_{0}}}.
\end{gather*}%
By Lemma \ref{Le4-2}, 
\begin{equation*}
\left\Vert \Delta _{2}\left( I+\Delta _{1}\right) ^{-1}\right\Vert _{\infty
}\leq \frac{b}{\sqrt{1-b^{2}}},
\end{equation*}%
it holds 
\begin{equation}
\left\Vert \Delta _{I_{\bar{G}_{0}}}\right\Vert _{\infty }\leq \frac{\gamma b%
}{\sqrt{1-b^{2}}},\gamma =\left\Vert \left[ 
\begin{array}{c}
U_{0} \\ 
V_{0}%
\end{array}%
\right] \right\Vert _{\infty }.  \label{eq4-21}
\end{equation}%
On the assumption that $\left\Vert \Delta _{I_{\bar{G}_{0}}}\right\Vert
_{\infty }<1,$ applying the result in Theorem \ref{Theo3-1} yields 
\begin{equation}
J_{th}=\frac{\gamma b}{\sqrt{1-\left( 1+\gamma ^{2}\right) b^{2}}}\left(
\left\Vert \left[ 
\begin{array}{c}
u \\ 
y%
\end{array}%
\right] \right\Vert _{2}^{2}-\left\Vert r_{\mathcal{I}_{\bar{G}%
_{0}}}\right\Vert _{2}^{2}\right) ^{1/2}.  \label{eq4-10}
\end{equation}

\subsection{A closed-loop image subspace projection-based fault detection
scheme}

Consider the nominal closed-loop dynamics (\ref{eq4-3})%
\begin{equation*}
\left[ 
\begin{array}{c}
u \\ 
y%
\end{array}%
\right] =\left[ 
\begin{array}{c}
M_{0} \\ 
N_{0}%
\end{array}%
\right] \hat{v},
\end{equation*}%
which can be, for instance, used to generate an extended residual vector, 
\begin{equation*}
\left[ 
\begin{array}{c}
r_{u} \\ 
r_{y}%
\end{array}%
\right] =\left[ 
\begin{array}{c}
u \\ 
y%
\end{array}%
\right] -\left[ 
\begin{array}{c}
M_{0} \\ 
N_{0}%
\end{array}%
\right] \hat{v}\in \mathcal{H}_{2}^{p+m},
\end{equation*}%
for the purpose of detecting cyber-attacks on feedback control systems and
recovering control performance degradation \cite{Ding2020,DLautomatica2022}.
This motivates us to propose a fault detection scheme based on the
projection onto the image subspace of the feedback control system defined in
the sequel.

\bigskip

It is obvious that $\left[ 
\begin{array}{c}
M_{0} \\ 
N_{0}%
\end{array}%
\right] \in \mathcal{RH}_{\infty }$ and $\left( I,\left[ 
\begin{array}{c}
M_{0} \\ 
N_{0}%
\end{array}%
\right] \right) $ builds a RCP of the transfer matrix $G_{c}^{0}(z),$%
\begin{equation}
y=G_{c}^{0}\hat{v},G_{c}^{0}=\left[ 
\begin{array}{c}
M_{0} \\ 
N_{0}%
\end{array}%
\right] \Longrightarrow \left[ 
\begin{array}{c}
\hat{v} \\ 
u \\ 
y%
\end{array}%
\right] =\left[ 
\begin{array}{c}
I \\ 
M_{0} \\ 
N_{0}%
\end{array}%
\right] \hat{v}  \label{eq4-14}
\end{equation}%
with the sub-index $c$ standing for closed-loop. The state space model of $%
G_{c}^{0}$ is given by%
\begin{equation*}
G_{c}^{0}:=\left( A_{c},B_{c},C_{c},D_{c}\right) =\left( A+BF_{0},BV_{0}, 
\left[ 
\begin{array}{c}
F_{0} \\ 
C+DF_{0}%
\end{array}%
\right] ,\left[ 
\begin{array}{c}
V_{0} \\ 
DV_{0}%
\end{array}%
\right] \right) .
\end{equation*}%
Based on it, we have the following normalised SIR of $G_{c}^{0}$: 
\begin{gather*}
\left[ 
\begin{array}{c}
\hat{v} \\ 
u \\ 
y%
\end{array}%
\right] =I_{G_{c}^{0}}\hat{v}=\left[ 
\begin{array}{c}
M_{0,c} \\ 
N_{0,c}%
\end{array}%
\right] \hat{v}, \\
M_{0,c}=\left( A_{c}+B_{c}F_{c},B_{c}V_{c},F_{c},V_{c}\right)
,N_{0,c}=\left(
A_{c}+B_{c}F_{c},B_{c}V_{c},C_{c}+D_{c}F_{c},D_{c}V_{c}\right) , \\
F_{c}=-\left( I+D_{c}^{T}D_{c}+B_{c}^{T}XB_{c}\right) ^{-1}\Phi ,\Phi
=D_{c}^{T}C_{c}+B_{c}^{T}XA_{c}, \\
V_{c}=\left( I+D_{c}^{T}D_{c}+B_{c}^{T}XB_{c}\right) ^{-1/2}\Gamma ,\Gamma
^{T}\Gamma =I, \\
X=A_{c}^{T}XA_{c}+C_{c}^{T}C_{c}-\Phi ^{T}\left(
I+D_{c}^{T}D_{c}+B_{c}^{T}XB_{c}\right) ^{-1}\Phi ,X>0.
\end{gather*}%
It is well-known that between any two RCPs there exists a one-to-one mapping 
\cite{Ding2020}. In our case, it holds%
\begin{equation}
I_{G_{c}^{0}}=\left[ 
\begin{array}{c}
M_{0,c} \\ 
N_{0,c}%
\end{array}%
\right] =\left[ 
\begin{array}{c}
I \\ 
M_{0} \\ 
N_{0}%
\end{array}%
\right] M_{0,c},N_{0,c}=\left[ 
\begin{array}{c}
M_{0} \\ 
N_{0}%
\end{array}%
\right] M_{0,c}.  \label{eq4-17}
\end{equation}

\begin{Def}
Given the close-loop model (\ref{eq4-14}) and the corresponding RCP $\left(
M_{0,c},N_{0,c}\right) ,$ the $\mathcal{H}_{2}$ subspace $\mathcal{I}%
_{G_{c}^{0}}$ defined by 
\begin{equation}
\mathcal{I}_{G_{c}^{0}}=\left\{ \left[ 
\begin{array}{c}
\hat{v} \\ 
u \\ 
y%
\end{array}%
\right] :\left[ 
\begin{array}{c}
\hat{v} \\ 
u \\ 
y%
\end{array}%
\right] =\left[ 
\begin{array}{c}
M_{0,c} \\ 
N_{0,c}%
\end{array}%
\right] \hat{v},\hat{v}\in \mathcal{H}_{2}\hspace{-2pt}\right\}
\label{eq4-16}
\end{equation}%
is called image subspace of the closed-loop.
\end{Def}

\begin{Rem}
In the above definition, it is assumed that $\hat{V}_{0}^{-1}\in \mathcal{RH}%
_{\infty }.$ Otherwise, $\hat{v}$ will be substituted by $v$ and $\left[ 
\begin{array}{c}
M_{0} \\ 
N_{0}%
\end{array}%
\right] \hat{V}_{0}$ is handled as done in the previous subsection without
loss of generality.
\end{Rem}

Applying the basic projection-based residual generation scheme introduced in
Section 3 results in%
\begin{equation}
r_{c}=\left[ 
\begin{array}{c}
\hat{v} \\ 
u \\ 
y%
\end{array}%
\right] -p_{\mathcal{I}_{G_{c}^{0}}},p_{\mathcal{I}_{G_{c}^{0}}}=\mathcal{P}%
_{\mathcal{I}_{G_{c}^{0}}}\left[ 
\begin{array}{c}
\hat{v} \\ 
u \\ 
y%
\end{array}%
\right] =\left( I-\mathcal{L}_{I_{G_{c}^{0}}}\mathcal{L}_{I_{G_{c}^{0}}}^{%
\ast }\right) \left[ 
\begin{array}{c}
\hat{v} \\ 
u \\ 
y%
\end{array}%
\right]  \label{eq4-18}
\end{equation}%
with the residual vector $r_{c}.$ In the next step, the threshold $J_{th,c},$%
\begin{equation*}
J_{th,c}=\sup_{\left\Vert \Delta _{I,c}\right\Vert _{\infty }\leq
b}\left\Vert r_{c}\right\Vert _{2},
\end{equation*}%
is to be determined. It follows from Lemma \ref{Le4-1} that in case of
uncertainty $\Delta _{I},$%
\begin{align*}
\left[ 
\begin{array}{c}
\hat{v} \\ 
u \\ 
y%
\end{array}%
\right] & =\left[ 
\begin{array}{c}
I \\ 
M\left( I+\Delta _{1}\right) ^{-1} \\ 
N\left( I+\Delta _{1}\right) ^{-1}%
\end{array}%
\right] \hat{v}=\left[ 
\begin{array}{c}
I \\ 
\left( M_{0}+\Delta _{M}\right) \left( I+\Delta _{1}\right) ^{-1} \\ 
\left( N_{0}+\Delta _{N}\right) \left( I+\Delta _{1}\right) ^{-1}%
\end{array}%
\right] \hat{v} \\
& =\left( \left[ 
\begin{array}{c}
M_{0,c} \\ 
N_{0,c}%
\end{array}%
\right] +\left[ 
\begin{array}{c}
0 \\ 
\Delta _{N_{0,c}}%
\end{array}%
\right] \right) M_{0,c}^{-1}\hat{v}, \\
\Delta _{N_{0,c}}& =\left( \left[ 
\begin{array}{c}
\left( M_{0}+\Delta _{M}\right) \left( I+\Delta _{1}\right) ^{-1} \\ 
\left( N_{0}+\Delta _{N}\right) \left( I+\Delta _{1}\right) ^{-1}%
\end{array}%
\right] -\left[ 
\begin{array}{c}
M_{0} \\ 
N_{0}%
\end{array}%
\right] \right) M_{0,c}.
\end{align*}%
Recalling (\ref{eq4-20}), it turns out%
\begin{gather*}
\left[ 
\begin{array}{c}
\left( M_{o}+\Delta _{M}\right) \left( I+\Delta _{1}\right) ^{-1} \\ 
\left( N_{o}+\Delta _{N}\right) \left( I+\Delta _{1}\right) ^{-1}%
\end{array}%
\right] -\left[ 
\begin{array}{c}
M_{0} \\ 
N_{0}%
\end{array}%
\right] =\left[ 
\begin{array}{c}
U_{0} \\ 
V_{0}%
\end{array}%
\right] \Delta _{2}\left( I+\Delta _{1}\right) ^{-1} \\
\Longrightarrow \Delta _{N_{0,c}}=\left[ 
\begin{array}{c}
U_{0} \\ 
V_{0}%
\end{array}%
\right] \Delta _{2}\left( I+\Delta _{1}\right) ^{-1}M_{0,c}.
\end{gather*}%
Consider that $M_{0,c},M_{0,c}^{-1}\in \mathcal{RH}_{\infty }$ and thus%
\begin{align*}
\mathcal{I}_{G_{c}}& =\left\{ \left[ 
\begin{array}{c}
\hat{v} \\ 
u \\ 
y%
\end{array}%
\right] :\left[ 
\begin{array}{c}
\hat{v} \\ 
u \\ 
y%
\end{array}%
\right] =\left[ 
\begin{array}{c}
I \\ 
M\left( I+\Delta _{1}\right) ^{-1} \\ 
N\left( I+\Delta _{1}\right) ^{-1}%
\end{array}%
\right] \hat{v},\hat{v}\in \mathcal{H}_{2}\hspace{-2pt}\right\} \\
& =\left\{ \left[ 
\begin{array}{c}
\hat{v} \\ 
u \\ 
y%
\end{array}%
\right] :\left[ 
\begin{array}{c}
\hat{v} \\ 
u \\ 
y%
\end{array}%
\right] =\left( \left[ 
\begin{array}{c}
M_{0,c} \\ 
N_{0,c}%
\end{array}%
\right] +\left[ 
\begin{array}{c}
0 \\ 
\Delta _{N_{0,c}}%
\end{array}%
\right] \right) \hat{v},\hat{v}\in \mathcal{H}_{2}\hspace{-2pt}\right\} .
\end{align*}%
Now, let 
\begin{equation}
I_{G_{c}}=\left[ 
\begin{array}{c}
M_{0,c} \\ 
N_{0,c}%
\end{array}%
\right] +\left[ 
\begin{array}{c}
0 \\ 
\Delta _{N_{0,c}}%
\end{array}%
\right] \Longrightarrow I_{G_{c}}-I_{G_{c}^{0}}=\left[ 
\begin{array}{c}
0 \\ 
\Delta _{N_{0,c}}%
\end{array}%
\right] =:\Delta _{I_{G_{c}}}.
\end{equation}%
By Lemma \ref{Le4-2}, it holds%
\begin{gather}
\left\Vert \Delta _{I_{G_{c}}}\right\Vert _{\infty }=\left\Vert \left[ 
\begin{array}{c}
U_{0} \\ 
V_{0}%
\end{array}%
\right] \Delta _{2}\left( I+\Delta _{1}\right) ^{-1}M_{0,c}\right\Vert
_{\infty }\leq \frac{\gamma b}{\sqrt{1-b^{2}}}=\varepsilon ,  \notag \\
\left\Vert \left[ 
\begin{array}{c}
U_{0} \\ 
V_{0}%
\end{array}%
\right] \right\Vert _{\infty }\left\Vert M_{0,c}\right\Vert _{\infty }\leq
\left\Vert \left[ 
\begin{array}{c}
U_{0} \\ 
V_{0}%
\end{array}%
\right] \right\Vert _{\infty }=\gamma .  \label{eq4-19}
\end{gather}%
In the sequel, it is assumed that $\varepsilon <1.$

\bigskip

Next, let $\delta \left( \mathcal{I}_{G_{c}},\mathcal{I}_{G_{c}^{0}}\right) $
be the gap metric between $\mathcal{I}_{G_{c}}$ and $\mathcal{I}%
_{G_{c}^{0}}, $ and define%
\begin{eqnarray*}
\mathcal{I}_{G_{c},\Delta _{I_{G_{c}}}}(\varepsilon ) &=&\left\{ \mathcal{I}%
_{G_{c}}:\left\Vert \Delta _{I_{G_{c}}}\right\Vert _{\infty }=\left\Vert
I_{G_{c}}-I_{G_{c}^{0}}\right\Vert _{\infty }\leq \varepsilon \right\} , \\
\mathcal{I}_{G_{c},\delta }(\varepsilon ) &=&\left\{ \mathcal{I}%
_{G_{c}}:\delta \left( \mathcal{I}_{G_{c}},\mathcal{I}_{G_{c}^{0}}\right)
\leq \varepsilon \right\} .
\end{eqnarray*}%
Since, following the well-established result given in \cite%
{Georgiou&Smith90,Vinnicombe-book}, 
\begin{equation*}
\mathcal{I}_{G_{c},\Delta _{I_{G_{c}}}}(\varepsilon )=\mathcal{I}%
_{G_{c},\delta }(\varepsilon ),
\end{equation*}%
we finally have the following theorem.

\begin{Theo}
\label{Theo4-1}Given the control loop (\ref{eq4-1}) with control law (\ref%
{eq4-2a})-(\ref{eq4-2b}) and uncertainty $\Delta _{I,c}$ satisfying (\ref%
{eq4-6a}). Suppose that projection-based residual generator (\ref{eq4-18})
is used with residual vector $r_{c}$, then the corresponding threshold is
given by 
\begin{equation}
J_{th,c}=\sup_{\left\Vert \Delta _{I,c}\right\Vert _{\infty }\leq
b}\left\Vert r_{c}\right\Vert _{2}=\frac{\gamma b}{\sqrt{1-\left( 1+\gamma
^{2}\right) b^{2}}}\left( \left\Vert \left[ 
\begin{array}{c}
\hat{v} \\ 
u \\ 
y%
\end{array}%
\right] \right\Vert _{2}^{2}-\left\Vert r_{c}\right\Vert _{2}^{2}\right)
^{1/2},
\end{equation}%
where $\gamma $ is a constant given in (\ref{eq4-19}).
\end{Theo}

\begin{proof}
The proof is analogue to the one of Theorem \ref{Theo3-1} and thus omitted.
\end{proof}

\bigskip

It is noteworthy that, for both detection schemes presented in this section,
the influence of controller design on the detection performance can be
clearly seen from the condition (\ref{eq4-6a}) and threshold setting. This
reveals that, in order to enhance the fault detectability (by reducing the
threshold), the norms of the LCP and RCP of the controller, $\left\Vert %
\left[ 
\begin{array}{cc}
\hat{V}_{0} & \text{ }-\hat{U}_{0}%
\end{array}%
\right] \right\Vert _{\infty }$ and $\left\Vert \left[ 
\begin{array}{c}
U_{0} \\ 
V_{0}%
\end{array}%
\right] \right\Vert _{\infty },$ are to be set as small as possible. It is
well-known from robust control theory that reducing $\left\Vert \left[ 
\begin{array}{cc}
\hat{V}_{0} & \text{ }-\hat{U}_{0}%
\end{array}%
\right] \right\Vert _{\infty }$ or/and $\left\Vert \left[ 
\begin{array}{c}
U_{0} \\ 
V_{0}%
\end{array}%
\right] \right\Vert _{\infty }$ increases the stability margin. This
observation coincides with the result reported in \cite%
{LD-Automatica-2020,Li-IEEETCST2020}.

\section{Fault classification issues}

In the observer-based fault diagnosis framework, fault detection and
isolation are two major tasks and mainly handled with the aid of system
analysis and observer design. They can be addressed in the context of fault
classification as well. Roughly speaking, the task of fault classification
is to determine to which class a fault belongs. In this section,
projection-based methods and algorithms are applied to dealing with fault
classification issues.

\subsection{Fault detection: a binary classification scheme}

In the previous sections, we have studied the basic projection-based fault
detection methods whose basis is the nominal system model and information
about model uncertainties. This is a typical one-class classification
problem. When a system model (or data) for faulty operations, in addition to
the nominal model, exists, fault detection can be achieved using both
models. This is a binary fault classification task. It is evident that
binary classification solutions, thanks to additional information, could
considerably enhance the fault detection performance. Additionally, a
multi-class classification problem towards fault isolation can be
reformulated as a bank of binary fault classification sub-problems. In this
subsection, binary fault classification issues are addressed in the
framework of projection-based methods.

\bigskip

Consider $G_{i}\in \mathcal{RL}_{\infty }^{m\times p},i=0,1,$ and let 
\begin{equation*}
K_{G_{i}}=\left[ 
\begin{array}{cc}
-\hat{N}_{i} & \hat{M}_{i}%
\end{array}%
\right] ,I_{G_{i}}=\left[ 
\begin{array}{c}
M_{i} \\ 
N_{i}%
\end{array}%
\right]
\end{equation*}%
denote the corresponding normalised SKR and SIR of $G_{i},$ and $\mathcal{K}%
_{G_{i}},\mathcal{I}_{G_{i}}$ be the kernel and image subspaces, 
\begin{align*}
\mathcal{K}_{G_{i}}& =\left\{ \left[ 
\begin{array}{c}
u \\ 
y%
\end{array}%
\right] \in \mathcal{H}_{2}:\mathcal{K}_{i}\left[ 
\begin{array}{c}
u \\ 
y%
\end{array}%
\right] =0\right\} , \\
\mathcal{I}_{G_{i}}& =\left\{ \left[ 
\begin{array}{c}
u \\ 
y%
\end{array}%
\right] :\left[ 
\begin{array}{c}
u \\ 
y%
\end{array}%
\right] =\mathcal{I}_{i}v,v\in \mathcal{H}_{2}\right\} .
\end{align*}%
It is supposed that the system operates either in the nominal or faulty
state, and the nominal and faulty models are represented by sub-index $0$
and sub-index $1$, respectively, i.e. any data vector $\left[ 
\begin{array}{c}
u \\ 
y%
\end{array}%
\right] $ belongs to $\mathcal{I}_{G_{0}}\cup \mathcal{I}_{G_{1}}.$

\bigskip

For the fault detection purpose, we now construct two projection-based
residual generators, 
\begin{align}
\mathcal{P}_{\mathcal{I}_{G_{i}}}& =\mathcal{L}_{I_{G_{i}}}\mathcal{L}%
_{I_{G_{i}}}^{\ast },p_{\mathcal{I}_{G_{i}}}=\mathcal{P}_{\mathcal{I}%
_{G_{i}}}\left[ 
\begin{array}{c}
u \\ 
y%
\end{array}%
\right] ,i=0,1,  \notag \\
r_{\mathcal{I}_{G_{i}}}& =\left[ 
\begin{array}{c}
u \\ 
y%
\end{array}%
\right] -p_{\mathcal{I}_{G_{i}}}=\left( \mathcal{I}-\mathcal{L}_{I_{G_{i}}}%
\mathcal{L}_{I_{G_{i}}}^{\ast }\right) \left[ 
\begin{array}{c}
u \\ 
y%
\end{array}%
\right] .  \label{eq5-2}
\end{align}%
It is evident that, without considering model uncertainties, 
\begin{equation*}
\left\{ 
\begin{array}{l}
r_{\mathcal{I}_{G_{0}}}=0,\text{ fault-free,} \\ 
r_{\mathcal{I}_{G_{1}}}=0,\text{ faulty.}%
\end{array}%
\right.
\end{equation*}%
Consequently, the (ideal) detection logic seems to be%
\begin{equation}
\left\{ 
\begin{array}{l}
\left\Vert r_{\mathcal{I}_{G_{0}}}\right\Vert _{2}\leq J_{th,0}\text{ }%
\Rightarrow \text{fault-free, } \\ 
\left\Vert r_{\mathcal{I}_{G_{1}}}\right\Vert _{2}\leq J_{th,1}\Rightarrow 
\text{faulty,}%
\end{array}%
\right.  \label{eq5-3}
\end{equation}%
with $J_{th,0}$ and $J_{th,1}$ denoting the corresponding thresholds.

\bigskip

In order to study threshold setting for the residual generator (\ref{eq5-2}%
), next, we analyse relations between $\mathcal{P}_{\mathcal{I}%
_{G_{i}}},i=0,1,$ and, associated with it, the dynamics of $r_{\mathcal{I}%
_{G_{i}}}$. Aiming at addressing practical cases, it is assumed that 
\begin{equation}
0<\vec{\delta}\left( \mathcal{I}_{G_{0}},\mathcal{I}_{G_{1}}\right) =\vec{%
\delta}\left( \mathcal{I}_{G_{1}},\mathcal{I}_{G_{0}}\right) =\delta \left( 
\mathcal{I}_{G_{0}},\mathcal{I}_{G_{1}}\right) <1.  \label{eq5-4}
\end{equation}%
The following theorem given by \cite{Feintuch_book} (Theorem 9.1.4) is
essential for our study.

\begin{Theo}
\label{Theo5-1}Suppose that $\mathcal{V}_{1}$ and $\mathcal{V}_{2}$ are two
closed subspaces of $\mathcal{H}.$ Then $\delta \left( \mathcal{V}_{1},%
\mathcal{V}_{2}\right) <1$ if and only if operator $\mathcal{P}_{\mathcal{V}%
_{1}}:\mathcal{V}_{2}$ $\rightarrow \mathcal{V}_{1}$ is one-to-one and onto.
\end{Theo}

\begin{Rem}
The proof given in \cite{Feintuch_book} shows that 
\begin{align*}
\left\Vert \mathcal{P}_{\mathcal{V}_{1}}-\mathcal{P}_{\mathcal{V}%
_{2}}\right\Vert & =\delta \left( \mathcal{V}_{1},\mathcal{V}_{2}\right)
<1\Longrightarrow \text{operator }\left( \mathcal{I}+\mathcal{P}_{\mathcal{V}%
_{2}}-\mathcal{P}_{\mathcal{V}_{1}}\right) \text{ is invertible,} \\
& \Longrightarrow \mathcal{P}_{\mathcal{V}_{1}}\mathcal{H}=\mathcal{P}_{%
\mathcal{V}_{1}}\left( \mathcal{I}+\mathcal{P}_{\mathcal{V}_{2}}-\mathcal{P}%
_{\mathcal{V}_{1}}\right) \mathcal{H}=\mathcal{P}_{\mathcal{V}_{1}}\mathcal{P%
}_{\mathcal{V}_{2}}\mathcal{H}.
\end{align*}%
As a result, $\mathcal{P}_{\mathcal{V}_{1}}$ maps $\mathcal{V}_{2}$ to $%
\mathcal{V}_{1}$ and it is one-to-one.
\end{Rem}

Applying this theorem to our case results in the following corollary.

\begin{Corol}
\label{Co5-1}Given $\mathcal{I}_{G_{i}},i=0,1,$ and the corresponding
projection operators $\mathcal{P}_{\mathcal{I}_{G_{i}}},$ then 
\begin{gather}
\mathcal{I}_{G_{i}}\cap \mathcal{I}_{G_{j}}\neq \mathcal{O},  \label{eq5-6}
\\
\exists \left[ 
\begin{array}{c}
u \\ 
y%
\end{array}%
\right] \left( \neq 0\right) ,\left[ 
\begin{array}{c}
u \\ 
y%
\end{array}%
\right] \in \mathcal{I}_{G_{i}},i=0,1\Longrightarrow r_{\mathcal{I}%
_{G_{j}}}=0,j\neq i,  \label{eq5-7} \\
\forall \left[ 
\begin{array}{c}
u \\ 
y%
\end{array}%
\right] \in \mathcal{I}_{G_{j}},\left\Vert r_{\mathcal{I}_{G_{i}}}\right%
\Vert _{2}=dist\left( \left[ 
\begin{array}{c}
u \\ 
y%
\end{array}%
\right] ,\mathcal{I}_{G_{i}}\right)  \notag \\
=\sqrt{\left\Vert \left[ 
\begin{array}{c}
u \\ 
y%
\end{array}%
\right] \right\Vert _{2}^{2}-\left\Vert \mathcal{P}_{\mathcal{I}_{G_{i}}}%
\mathcal{P}_{\mathcal{I}_{G_{j}}}\left[ 
\begin{array}{c}
u \\ 
y%
\end{array}%
\right] \right\Vert _{2}^{2}}<\left\Vert \left[ 
\begin{array}{c}
u \\ 
y%
\end{array}%
\right] \right\Vert _{2},  \label{eq5-8}
\end{gather}%
where $\mathcal{O}$ denotes empty set.
\end{Corol}

\begin{proof}
The claim (\ref{eq5-6}) is an immediate result of Theorem \ref{Theo5-1}.
From (\ref{eq5-6}) follows (\ref{eq5-7}). Equation in (\ref{eq5-8}) is
straightforward, because%
\begin{align*}
\left\Vert \left[ 
\begin{array}{c}
u \\ 
y%
\end{array}%
\right] \right\Vert _{2}^{2}& =\left\Vert \mathcal{P}_{\mathcal{I}_{G_{i}}}%
\left[ 
\begin{array}{c}
u \\ 
y%
\end{array}%
\right] \right\Vert _{2}^{2}+\left\Vert r_{\mathcal{I}_{G_{i}}}\right\Vert
_{2}^{2}, \\
\forall \left[ 
\begin{array}{c}
u \\ 
y%
\end{array}%
\right] & \in \mathcal{I}_{G_{j}},\left[ 
\begin{array}{c}
u \\ 
y%
\end{array}%
\right] =\mathcal{P}_{\mathcal{I}_{G_{j}}}\left[ 
\begin{array}{c}
u \\ 
y%
\end{array}%
\right] .
\end{align*}%
The inequality is due to the fact that 
\begin{equation*}
\forall \left[ 
\begin{array}{c}
u \\ 
y%
\end{array}%
\right] ,\mathcal{P}_{\mathcal{I}_{G_{i}}}\mathcal{P}_{\mathcal{I}_{G_{j}}}%
\left[ 
\begin{array}{c}
u \\ 
y%
\end{array}%
\right] \neq 0.
\end{equation*}
\end{proof}

Observe that (\ref{eq5-8}) implies, 
\begin{equation*}
\forall \left[ 
\begin{array}{c}
u \\ 
y%
\end{array}%
\right] \in \mathcal{I}_{G_{j}},p_{\mathcal{I}_{G_{i}}}=\mathcal{P}_{%
\mathcal{I}_{G_{i}}}\left[ 
\begin{array}{c}
u \\ 
y%
\end{array}%
\right] \in \mathcal{I}_{G_{i}}\cap \mathcal{I}_{G_{j}},
\end{equation*}%
that is, the projection of $\left[ 
\begin{array}{c}
u \\ 
y%
\end{array}%
\right] $ belonging to $\mathcal{I}_{G_{j}}$ onto $\mathcal{I}_{G_{i}}$\ is
a vector in the intersection set of $\mathcal{I}_{G_{j}}$ and $\mathcal{I}%
_{G_{i}}.$ Accordingly, $\left\Vert r_{\mathcal{I}_{G_{i}}}\right\Vert _{2}$
reaches its largest value at $\left[ 
\begin{array}{c}
u \\ 
y%
\end{array}%
\right] \in \mathcal{I}_{G_{j}},$ whose projection onto $\mathcal{I}%
_{G_{i}}\cap \mathcal{I}_{G_{j}}$ is at the smallest value. 

\bigskip 

Consider\ that 
\begin{eqnarray}
\mathcal{I}_{G_{i}} &=&\left( \mathcal{I}-\mathcal{P}_{\mathcal{I}%
_{G_{j}}}\right) \mathcal{I}_{G_{i}}+\mathcal{P}_{\mathcal{I}_{G_{j}}}%
\mathcal{I}_{G_{i}},  \label{eq5-8a} \\
\mathcal{P}_{\mathcal{I}_{G_{j}}}\mathcal{I}_{G_{i}} &=&\mathcal{I}%
_{G_{i}}\cap \mathcal{I}_{G_{j}}\neq \mathcal{O},\left( \mathcal{I}-\mathcal{%
P}_{\mathcal{I}_{G_{j}}}\right) \mathcal{I}_{G_{i}}\subseteq \mathcal{I}%
_{G_{j}}^{\perp }.
\end{eqnarray}%
This means that the whole process data subspace, $\mathcal{I}_{G}=\mathcal{I}%
_{G_{0}}\cup \mathcal{I}_{G_{1}},$ consists of three sets, 
\begin{align}
\mathcal{I}_{G}& =\mathcal{S}_{I}\oplus \mathcal{S}_{II}\oplus \mathcal{S}%
_{III},  \label{eq10-9} \\
\mathcal{S}_{II}& =\mathcal{I}_{G_{0}}\cap \mathcal{I}_{G_{1}},  \notag \\
\mathcal{S}_{I}& =\left\{ \left[ 
\begin{array}{c}
u \\ 
y%
\end{array}%
\right] ,\left[ 
\begin{array}{c}
u \\ 
y%
\end{array}%
\right] \in \mathcal{I}_{G_{0}}\cap \mathcal{I}_{G_{1}}^{\perp }\right\} , 
\notag \\
\mathcal{S}_{III}& =\left\{ \left[ 
\begin{array}{c}
u \\ 
y%
\end{array}%
\right] ,\left[ 
\begin{array}{c}
u \\ 
y%
\end{array}%
\right] \in \mathcal{I}_{G_{1}}\cap \mathcal{I}_{G_{0}}^{\perp }\right\} . 
\notag
\end{align}%
Correspondingly, the dynamics of residuals $r_{\mathcal{I}_{G_{i}}},i=0,1,$
can be divided into three ranges indicating different operation states of
the system:%
\begin{equation}
\left\{ 
\begin{array}{l}
\text{I. \ \ }0\leq \left\Vert r_{\mathcal{I}_{G_{0}}}\right\Vert _{2}\leq
J_{th,0}\text{ and }J_{th,1}<\left\Vert r_{\mathcal{I}_{G_{1}}}\right\Vert
_{2}\leq \delta \left( \mathcal{I}_{G_{0}},\mathcal{I}_{G_{1}}\right)
\left\Vert \left[ 
\begin{array}{c}
u \\ 
y%
\end{array}%
\right] \right\Vert _{2}, \\ 
\text{II. \ }0\leq \left\Vert r_{\mathcal{I}_{G_{0}}}\right\Vert _{2}\leq
J_{th,0}\text{ and }0\leq \left\Vert r_{\mathcal{I}_{G_{1}}}\right\Vert
_{2}\leq J_{th,1}, \\ 
\text{III. }J_{th,0}<\left\Vert r_{\mathcal{I}_{G_{0}}}\right\Vert _{2}\leq
\delta \left( \mathcal{I}_{G_{0}},\mathcal{I}_{G_{1}}\right) \left\Vert %
\left[ 
\begin{array}{c}
u \\ 
y%
\end{array}%
\right] \right\Vert _{2}\text{and }0\leq \left\Vert r_{\mathcal{I}%
_{G_{1}}}\right\Vert _{2}\leq J_{th,1}.%
\end{array}%
\right.   \label{eq5-9}
\end{equation}%
Here, the thresholds $J_{th,0}$ and $J_{th,1}$ are introduced to ensure a
reliable detection in case of model uncertainties and thus can be determined
using the gap metric schemes proposed in the previous sections. It is
evident that Cases I and III, corresponding to $\mathcal{S}_{I}$ and $%
\mathcal{S}_{III},$ respectively, indicate fault-free and faulty operations,
respectively. Case II is the result of (\ref{eq5-7}), which raises our
interest for a reasonable and convincing interpretation.

\bigskip

It is well-known that faults caused by e.g. ageing is a longtime process. It
begins with incipient degradation that does not affect the system dynamics
significantly, and thus is interpreted as the transitional phase from the
fault-free to faulty operation. Accordingly, the operations in this range
can be called incipient fault. This is one interpretation of Case II. A
further possible explanation for Case II is the so-called intermittent
faults which repeatedly occur in the process over a time interval and then
disappear. Their emergence and disappearance are activated, for instance, by
certain system operation conditions. That means, in Case II, for some input
signals (and so the corresponding outputs), the system may operate normally,
and then, as a response to the change of input signals, works in the faulty
operation. In summary, it can be concluded that, if the system operates in
range II, warning should be triggered to call operator's attention for
possible faults.

\bigskip

In order to gain a deeper insight into the above discussion, we would like
to introduce some useful system theoretic aspects in the sequel. The
following theorem is a summary of some relevant results in \cite%
{Feintuch_book} (Section 9.1).

\begin{Theo}
\label{Theo5-2}Suppose that $\mathcal{V}_{1}$ and $\mathcal{V}_{2}$ are two
closed subspaces of $\mathcal{H},$ $\delta \left( \mathcal{V}_{1},\mathcal{V}%
_{2}\right) <1,$ and the operator%
\begin{equation*}
\mathcal{M}_{21}:\mathcal{V}_{2}\rightarrow \mathcal{V}_{1},\forall x\in 
\mathcal{V}_{2},\mathcal{M}_{21}x=\mathcal{P}_{\mathcal{V}_{1}}x.
\end{equation*}%
is thus invertible. Then, it holds%
\begin{align*}
\mathcal{V}_{2}& =\left( \mathcal{I}\left( \mathcal{V}_{1}\right) +\mathcal{X%
}\right) \mathcal{V}_{1}, \\
\mathcal{X}& :\mathcal{V}_{1}\rightarrow \mathcal{V}_{1}^{\bot },\mathcal{X}%
=\left( \mathcal{I-P}_{\mathcal{V}_{1}}\right) \mathcal{M}_{21}^{-1},
\end{align*}%
where $\mathcal{I}\left( \mathcal{V}_{1}\right) $ is the identity operator
restricted to $\mathcal{V}_{1}.$
\end{Theo}

According to this theorem, we have, for instance, 
\begin{equation}
\mathcal{I}_{G_{1}}=\left( \mathcal{I}\left( \mathcal{I}_{G_{0}}\right) +%
\mathcal{X}\right) \mathcal{I}_{G_{0}},\mathcal{X}=\left( \mathcal{I-P}_{%
\mathcal{I}_{G_{0}}}\right) \mathcal{M}_{10}^{-1}.  \label{eq5-17}
\end{equation}%
Equation (\ref{eq5-17}) reveals the relation between fault-free and faulty
subspaces. The subspace in $\mathcal{I}_{G_{1}}$, which is orthogonal to the
(fault-free) subspace $\mathcal{I}_{G_{0}}$ and thus the corresponding
residual dynamics is represented by Case III in (\ref{eq5-9}), is modelled
by $\mathcal{XI}_{G_{0}},\mathcal{X}:\mathcal{I}_{G_{0}}\rightarrow \mathcal{%
I}_{G_{0}}^{\bot }.$ The transitional phase from the fault-free to faulty
operation (Case II) is represented by $\mathcal{I}_{G_{1}}-\mathcal{XI}%
_{G_{0}}.$ Although it is difficult to give an analytical form of $\mathcal{I%
}_{G_{1}}-\mathcal{XI}_{G_{0}},$ it is possible to characterise this
subspace by means of simulation or data, for instance, using the so-called
randomised algorithm technique \cite{DLKautomatica2019}. This would be
helpful to understand the mechanism of incipient and intermittent faults.

\bigskip

As a summary of this subsection, Figure 2 sketches the system operation
state and the corresponding detection logic. 
\begin{figure}[h]
\centering\includegraphics[width=13.5cm,height=5cm]{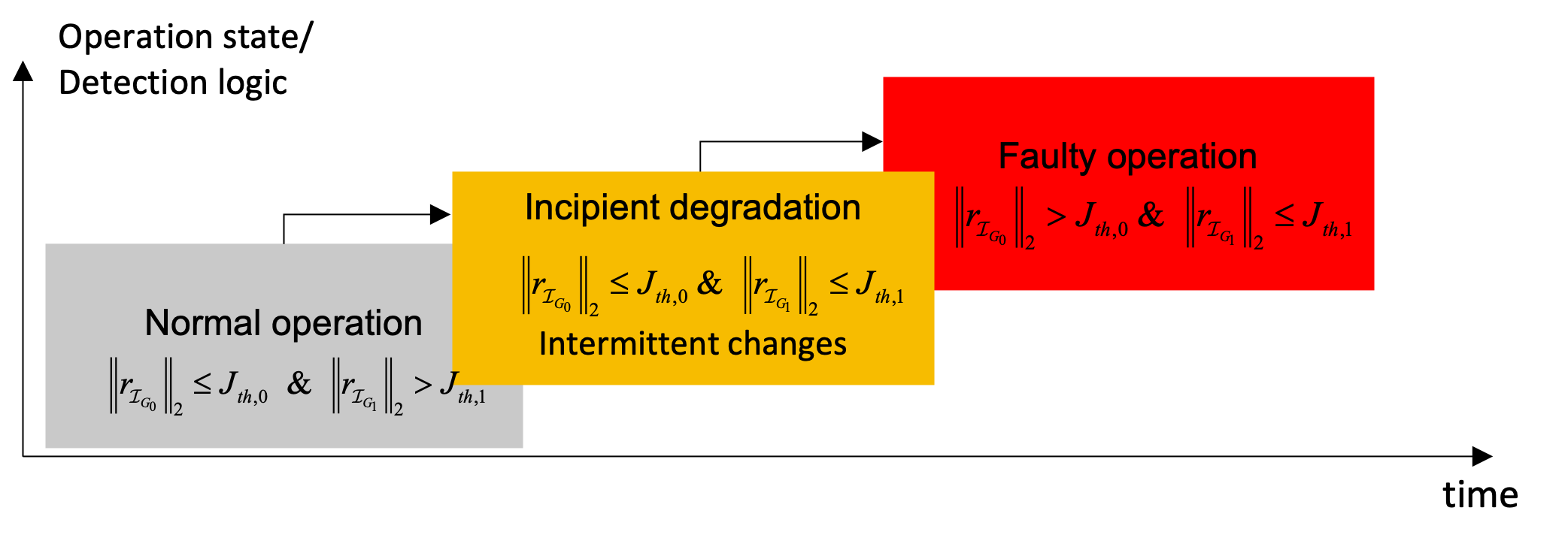}
\caption{Operation states and detection logic}
\end{figure}

\subsection{Fault isolation: multi-class classification}

A straightforward extension of the binary classification method introduced
in the previous subsection to multi-class classification provides us with a
projection-based method for fault isolation. To this end, we first define
fault classes under consideration. Let $G_{i}\in \mathcal{RL}_{\infty
}^{m\times p},i=1,\cdots ,M,$ be (uncertain) transfer functions representing
the $M$ faulty classes of the nominal system (\ref{eq2-1}), and $\mathcal{K}%
_{G_{i}},\mathcal{I}_{G_{i}}$ denote the corresponding kernel and image
subspaces, 
\begin{align*}
\mathcal{K}_{G_{i}}& =\left\{ \left[ 
\begin{array}{c}
u \\ 
y%
\end{array}%
\right] \in \mathcal{H}_{2}:K_{G_{i}}\left[ 
\begin{array}{c}
u \\ 
y%
\end{array}%
\right] =0\right\} , \\
\mathcal{I}_{G_{i}}& =\left\{ \left[ 
\begin{array}{c}
u \\ 
y%
\end{array}%
\right] :\left[ 
\begin{array}{c}
u \\ 
y%
\end{array}%
\right] =I_{G_{i}}v,v\in \mathcal{H}_{2}\right\} .
\end{align*}%
Here, $K_{G_{i}}$ and $I_{G_{i}},i=1,\cdots ,M,$ are the normalised SKR and
SIR of $G_{i},$ which are known and defined by%
\begin{equation*}
K_{G_{i}}=\left[ 
\begin{array}{cc}
-\hat{N}_{i} & \hat{M}_{i}%
\end{array}%
\right] ,I_{G_{i}}=\left[ 
\begin{array}{c}
M_{i} \\ 
N_{i}%
\end{array}%
\right] .
\end{equation*}%
Considering the existence of possible model uncertainties, let $\Delta
_{I_{G_{i}}},i=1,\cdots ,M,$ denote uncertainties satisfying 
\begin{equation*}
\left\Vert \Delta _{I_{G_{i}}}\right\Vert _{\infty }\leq \delta _{I_{i}}<1.
\end{equation*}%
Recall that 
\begin{gather*}
\left\{ \left[ 
\begin{array}{c}
u \\ 
y%
\end{array}%
\right] :\left[ 
\begin{array}{c}
u \\ 
y%
\end{array}%
\right] =\left( I_{G_{i}}+\Delta _{I_{G_{i}}}\right) v,v\in \mathcal{H}%
_{2},\left\Vert \Delta _{I_{G_{i}}}\right\Vert _{\infty }\leq \delta
_{I_{i}}\right\} \\
=\left\{ \mathcal{I}_{G_{i}}+\Delta \mathcal{I}_{G_{i}}:\delta \left( 
\mathcal{I}_{G_{i}},\mathcal{I}_{G_{i}}+\Delta \mathcal{I}_{G_{i}}\right)
\leq \delta _{I_{i}}\right\} .
\end{gather*}%
This motivates us to define fault classifiability (isolability) as follows.

\begin{Def}
Given fault classes $G_{i},i=1,\cdots ,M.$ The $j$-th fault class with $j\in
\left\{ 1,\cdots ,M\right\} $ is called classifiable, if 
\begin{equation}
\forall i\in \left\{ 1,\cdots ,M,i\neq j\right\} ,\delta \left( \mathcal{I}%
_{G_{i}},\mathcal{I}_{G_{j}}\right) >\max \left\{ \delta _{I_{i}},\delta
_{I_{j}}\right\} .  \label{eq5-22}
\end{equation}%
The faults are called classifiable or equivalently isolable, if all fault
classes are classifiable, i.e.%
\begin{equation}
\forall i,j\in \left\{ 1,\cdots ,M\right\} ,i\neq j,\delta \left( \mathcal{I}%
_{G_{i}},\mathcal{I}_{G_{j}}\right) >\max \left\{ \delta _{I_{i}},\delta
_{I_{j}}\right\} .  \label{eq5-23}
\end{equation}
\end{Def}

Condition (\ref{eq5-22}) implies that the binary classification method
introduced in the last subsection can be applied to classify (isolate) the $j
$-th fault class. And, if this is true for all fault classes, as described
by condition (\ref{eq5-23}), then all faults are classifiable (isolable). In
the sequel, on the assumption that the fault classes under consideration are
classifiable, a fault classification algorithm is proposed. Suppose that

\begin{itemize}
\item process measurement data $\left( u,y\right) $ have been collected and,
based on them and by means of the projection-based detection algorithms
proposed in Section 3, the fault has been detected,

\item the collected data $\left( u,y\right) $ as well as the
projection-based binary fault classification method proposed in the last
subsection are applied to the classification of the detected fault. The
thresholds adopted in the classification algorithm are denoted by $%
J_{th,i},i=1,\cdots ,M,$ where%
\begin{equation}
J_{th,i}=\frac{\delta _{I_{i}}}{\sqrt{1-\delta _{I_{i}}^{2}}}\left\Vert 
\mathcal{P}_{\mathcal{I}_{G_{i}}}\left[ 
\begin{array}{c}
u \\ 
y%
\end{array}%
\right] \right\Vert _{2}  \label{eq4-24}
\end{equation}%
following Theorem \ref{Theo3-1}.
\end{itemize}

\bigskip

\textbf{Fault classification algorithm}

\begin{itemize}
\item Generation of $M$ projection-based residuals,%
\begin{align*}
\mathcal{P}_{\mathcal{I}_{G_{i}}}& =\mathcal{L}_{I_{G_{i}}}\mathcal{L}%
_{I_{G_{i}}}^{\ast },p_{\mathcal{I}_{G_{i}}}=\mathcal{P}_{\mathcal{I}%
_{G_{i}}}\left[ 
\begin{array}{c}
u \\ 
y%
\end{array}%
\right] ,i=1,\cdots ,M, \\
r_{\mathcal{I}_{G_{i}}}& =\left[ 
\begin{array}{c}
u \\ 
y%
\end{array}%
\right] -p_{\mathcal{I}_{G_{i}}}=\left( \mathcal{I}-\mathcal{L}_{I_{G_{i}}}%
\mathcal{L}_{I_{G_{i}}}^{\ast }\right) \left[ 
\begin{array}{c}
u \\ 
y%
\end{array}%
\right] ;
\end{align*}

\item Decision logic%
\begin{align}
\left\Vert r_{\mathcal{I}_{G_{j}}}\right\Vert _{2}& \leq J_{th,j}\text{ and }%
\forall i\in \left\{ 1,\cdots ,M,i\neq j\right\} ,J_{th,i}<\left\Vert r_{%
\mathcal{I}_{G_{i}}}\right\Vert _{2}\leq \delta \left( \mathcal{I}_{G_{j}},%
\mathcal{I}_{G_{i}}\right) \left\Vert \left[ 
\begin{array}{c}
u \\ 
y%
\end{array}%
\right] \right\Vert _{2}  \label{eq5-25} \\
& \Longrightarrow \text{the fault belongs to the }j\text{-th fault class,} 
\notag \\
\left\Vert r_{\mathcal{I}_{G_{j}}}\right\Vert _{2}& \leq J_{th,j}\text{ and }%
\exists i\in \left\{ 1,\cdots ,M,i\neq j\right\} ,\left\Vert r_{\mathcal{I}%
_{G_{i}}}\right\Vert _{2}\leq J_{th,i}  \label{eq5-26} \\
& \Longrightarrow \text{the fault belongs to both the }j\text{-th and }i%
\text{-th fault classes.}  \notag
\end{align}
\end{itemize}

Next, we briefly explain the proposed algorithm. It is clear that in the
first step, the measurement data $\left( u,y\right) $ are first projected
onto $\mathcal{I}_{G_{i}}$ and, based on it, the residuals $r_{\mathcal{I}%
_{G_{i}}},i=1,\cdots ,M,$ are generated. In the second step, by means of the
decision logic (\ref{eq5-25}) a decision is made to which fault class the
measurement data $\left( u,y\right) $ belong. As discussed in the last
subsection, it is possible that the fault may simultaneously belong to more
than one fault class. Accordingly, the rule (\ref{eq5-26}) is introduced.

\section{Two modified projection-based fault detection schemes}

During real-time implementation of the projection-based residual generator (%
\ref{eq3-3}), the following two problems may arise: (i) the involved online
computation of $\left\Vert r_{\mathcal{I}_{G}}\right\Vert _{2}$ as described
by (\ref{eq3-23}) or (\ref{eq3-24}), and (ii) the (infinitely) long time
interval required for the computation of $l_{2}$-norm of the residual
signal. In this section, two alternative realisation schemes are proposed.

\subsection{A fault detection scheme with projection onto $\mathcal{L}_{2}$
space}

Comparing with a standard observer-based residual generation and evaluation
makes it clear that the term $\left\Vert \mathcal{P}_{\mathcal{H}_{2}^{\bot
}}\mathcal{L}_{I_{G}^{\sim }}\left[ 
\begin{array}{c}
u \\ 
y%
\end{array}%
\right] \right\Vert _{2}$ in the projection-based residual requires extra
(online) computation in addition to the implementation of an observer.
Moreover, the discussion in Section 3 reveals that this term is dedicated to
detecting "past" faults in the (nominal) system image subspace. In other
words, neglecting the term $\left\Vert \mathcal{P}_{\mathcal{H}_{2}^{\bot }}%
\mathcal{L}_{I_{G}^{\sim }}\left[ 
\begin{array}{c}
u \\ 
y%
\end{array}%
\right] \right\Vert _{2}$ could (considerably) reduce the online
computation, although at the cost of missing detection of those faults in
the set defined by $\mathcal{L}_{I_{G}}\mathcal{P}_{\mathcal{H}_{2}^{\bot }}%
\mathcal{L}_{I_{G}^{\sim }}\left[ 
\begin{array}{c}
u \\ 
y%
\end{array}%
\right] .$ In this subsection, we propose a projection-based scheme for the
realisation of this trade-off strategy.

\bigskip

Remember that the first term of $\left\Vert r_{\mathcal{I}%
_{G_{0}}}\right\Vert _{2}$ in (\ref{eq3-24}) is the realisation of $\mathcal{%
L}_{K_{G_{0}}^{\sim }}\mathcal{L}_{K_{G_{0}}},$ which is an operator of the
orthogonal projection given by%
\begin{equation*}
\mathcal{L}_{K_{G_{0}}^{\sim }}\mathcal{L}_{K_{G_{0}}}:\mathcal{L}%
_{2}\rightarrow \mathcal{L}_{2},\left( \mathcal{L}_{K_{G_{0}}^{\sim }}%
\mathcal{L}_{K_{G_{0}}}\right) ^{2}=\mathcal{L}_{K_{G_{0}}^{\sim }}\mathcal{L%
}_{K_{G_{0}}}.
\end{equation*}%
Here, $G_{0}$ represents the nominal system transfer matrix. It is clear
that 
\begin{equation}
\mathcal{P}_{\mathcal{K}_{G_{0}}}\left( \mathcal{L}_{2}\right) =\mathcal{I}-%
\mathcal{L}_{K_{G_{0}}^{\sim }}\mathcal{L}_{K_{G_{0}}}:\mathcal{L}%
_{2}\rightarrow \mathcal{L}_{2}  \label{eq6-1}
\end{equation}%
defines an orthogonal projection onto the kernel subspace $\mathcal{K}%
_{G_{0}}$ in $\mathcal{L}_{2},$%
\begin{equation*}
\mathcal{K}_{G_{0}}=\left\{ \left[ 
\begin{array}{c}
u \\ 
y%
\end{array}%
\right] \in \mathcal{L}_{2}:\left[ 
\begin{array}{cc}
-\hat{N}_{0} & \hat{M}_{0}%
\end{array}%
\right] \left[ 
\begin{array}{c}
u \\ 
y%
\end{array}%
\right] =0\hspace{-2pt}\right\} ,
\end{equation*}%
and, accordingly, 
\begin{equation}
r_{\mathcal{K}_{G_{0}}}=\left( \mathcal{I}-\mathcal{P}_{\mathcal{K}%
_{G_{0}}}\right) \left[ 
\begin{array}{c}
u \\ 
y%
\end{array}%
\right] =\mathcal{L}_{K_{G_{0}}^{\sim }}\mathcal{L}_{K_{G_{0}}}\left[ 
\begin{array}{c}
u \\ 
y%
\end{array}%
\right]  \label{eq6-2}
\end{equation}%
gives the implementation form of the residual vector. Consequently, 
\begin{equation*}
\left\Vert r_{\mathcal{K}_{G_{0}}}\right\Vert _{2}=\left\Vert \mathcal{L}%
_{K_{G_{0}}^{\sim }}\mathcal{L}_{K_{G_{0}}}\left[ 
\begin{array}{c}
u \\ 
y%
\end{array}%
\right] \right\Vert _{2}=\left\Vert \mathcal{L}_{K_{G_{0}}}\left[ 
\begin{array}{c}
u \\ 
y%
\end{array}%
\right] \right\Vert _{2},
\end{equation*}%
as expected. That means, for the detection purpose with the residual
evaluation function $\left\Vert r_{\mathcal{K}_{G_{0}}}\right\Vert _{2},$
the needed online computation is the observer-based residual generator (\ref%
{eq2-7a})-(\ref{eq2-7b}) or equivalently the SKR (\ref{eq2-7c}) leading to%
\begin{equation}
\left\Vert r_{0}\right\Vert _{2}=\left\Vert r_{\mathcal{K}%
_{G_{0}}}\right\Vert _{2}.  \label{eq6-13}
\end{equation}%
Next, we study threshold setting. Analogue to the uncertainty model (\ref%
{eq3-5})-(\ref{eq3-5a}), left-coprime factor uncertainty is introduced, 
\begin{align}
G& =\hat{M}^{-1}\hat{N}=\left( \hat{M}_{0}+\Delta _{\hat{M}}\right)
^{-1}\left( \hat{N}_{0}+\Delta _{\hat{N}}\right) ,\Delta _{\hat{N}},\Delta _{%
\hat{M}}\in \mathcal{RH}_{\infty },  \label{eq6-12a} \\
K_{G}& =\left[ 
\begin{array}{cc}
-\hat{N} & \text{ }\hat{M}%
\end{array}%
\right] =\left[ 
\begin{array}{cc}
-\hat{N}_{0}-\Delta _{\hat{N}} & \hat{M}_{0}+\Delta _{\hat{M}}%
\end{array}%
\right] =K_{G_{0}}+\Delta _{K},  \label{eq6-12b} \\
K_{G_{0}}& =\left[ 
\begin{array}{cc}
-\hat{N}_{0} & \hat{M}_{0}%
\end{array}%
\right] ,\Delta _{K}=\left[ 
\begin{array}{cc}
-\Delta _{\hat{N}} & \Delta _{\hat{M}}%
\end{array}%
\right] ,\sup \left\Vert \Delta _{K}\right\Vert _{\infty }=\delta _{\Delta
_{K}}<1  \label{eq6-12c}
\end{align}%
with normalised SKR $K_{G}.$ The threshold is defined by 
\begin{equation}
J_{th}=\sup_{\left\Vert \Delta _{K}\right\Vert _{\infty }\leq \delta
_{\Delta _{K}}}\left\Vert r_{\mathcal{K}_{G_{0}}}\right\Vert _{2}.
\label{eq6-3}
\end{equation}%
Notice that 
\begin{equation}
r_{0}(z)=\left[ 
\begin{array}{cc}
-\hat{N}_{0}(z) & \hat{M}_{0}(z)%
\end{array}%
\right] \left[ 
\begin{array}{c}
u(z) \\ 
y(z)%
\end{array}%
\right] =\left[ 
\begin{array}{cc}
\Delta _{\hat{N}}\text{ } & -\Delta _{\hat{M}}%
\end{array}%
\right] \left[ 
\begin{array}{c}
u(z) \\ 
y(z)%
\end{array}%
\right] ,  \label{eq6-14}
\end{equation}%
which, thanks to (\ref{eq6-13}), leads to, 
\begin{gather*}
\sup_{\left\Vert \Delta _{K}\right\Vert _{\infty }\leq \delta _{\Delta
_{K}}}\left\Vert r_{\mathcal{K}_{G_{0}}}\right\Vert _{2}=\sup_{\left\Vert
\Delta _{K}\right\Vert _{\infty }\leq \delta _{\Delta _{K}}}\left\Vert
r_{0}\right\Vert _{2}=\sup_{\left\Vert \Delta _{K}\right\Vert _{\infty }\leq
\delta _{\Delta _{K}}}\left\Vert \left[ 
\begin{array}{cc}
\Delta _{\hat{N}}\text{ } & -\Delta _{\hat{M}}%
\end{array}%
\right] \left[ 
\begin{array}{c}
u \\ 
y%
\end{array}%
\right] \right\Vert _{2} \\
=\delta _{\Delta _{K}}\left\Vert \left[ 
\begin{array}{c}
u \\ 
y%
\end{array}%
\right] \right\Vert _{2}.
\end{gather*}%
As a result, we have the following theorem.

\begin{Theo}
\label{Theo6-1}Given the model (\ref{eq2-1}) with model uncertainty
satisfying (\ref{eq6-12a})-(\ref{eq6-12c}), and suppose that
projection-based residual generator (\ref{eq6-14}) is used for the detection
purpose, then the corresponding threshold is given by 
\begin{align}
J_{th}& =\frac{\delta _{\Delta _{K}}}{\sqrt{1-\delta _{\Delta _{K}}^{2}}}%
\left\Vert \mathcal{P}_{\mathcal{K}_{G_{0}}}\left[ 
\begin{array}{c}
u \\ 
y%
\end{array}%
\right] \right\Vert _{2}  \label{eq6-4a} \\
& =\frac{\delta _{\Delta _{K}}}{\sqrt{1-\delta _{\Delta _{K}}^{2}}}\left(
\left\Vert \left[ 
\begin{array}{c}
u \\ 
y%
\end{array}%
\right] \right\Vert _{2}^{2}-\left\Vert r_{0}\right\Vert _{2}^{2}\right)
^{1/2}.  \label{eq6-4b}
\end{align}
\end{Theo}

\begin{proof}
The proof is similar to the one of Theorem \ref{Theo3-1}. Hence, we only
give the major step. It follows from 
\begin{eqnarray*}
\forall \left[ 
\begin{array}{c}
u \\ 
y%
\end{array}%
\right]  &\in &\mathcal{K}_{G}=\left\{ \left[ 
\begin{array}{c}
u \\ 
y%
\end{array}%
\right] \in \mathcal{L}_{2}:\left[ 
\begin{array}{cc}
-\hat{N} & \hat{M}%
\end{array}%
\right] \left[ 
\begin{array}{c}
u \\ 
y%
\end{array}%
\right] =0\hspace{-2pt}\right\}  \\
\left\Vert r_{\mathcal{K}_{G_{0}}}\right\Vert _{2}^{2} &=&\left\Vert
r_{0}\right\Vert _{2}^{2}\leq \delta _{\Delta _{K}}^{2}\left\Vert \left[ 
\begin{array}{c}
u \\ 
y%
\end{array}%
\right] \right\Vert _{2}^{2}=\delta _{\Delta _{K}}^{2}\left( \left\Vert 
\mathcal{P}_{\mathcal{K}_{G_{0}}}\left[ 
\begin{array}{c}
u \\ 
y%
\end{array}%
\right] \right\Vert _{2}^{2}+\left\Vert r_{\mathcal{K}_{G_{0}}}\right\Vert
_{2}^{2}\right) 
\end{eqnarray*}%
that 
\begin{align*}
\left\Vert r_{\mathcal{K}_{G_{0}}}\right\Vert _{2}& =\left\Vert
r_{0}\right\Vert _{2}\leq \frac{\delta _{\Delta _{K}}}{\sqrt{1-\delta
_{\Delta _{K}}^{2}}}\left\Vert \mathcal{P}_{\mathcal{K}_{G_{0}}}\left[ 
\begin{array}{c}
u \\ 
y%
\end{array}%
\right] \right\Vert _{2} \\
& =\frac{\delta _{\Delta _{K}}}{\sqrt{1-\delta _{\Delta _{K}}^{2}}}\left(
\left\Vert \left[ 
\begin{array}{c}
u \\ 
y%
\end{array}%
\right] \right\Vert _{2}^{2}-\left\Vert r_{0}\right\Vert _{2}^{2}\right)
^{1/2}.
\end{align*}
The proof is completed.
\end{proof}

\bigskip

At the end of this subsection, we would like to compare the standard
observer-based and the above introduced projection-based detection methods.
In view of residual generation, the relation $\left\Vert r_{0}\right\Vert
_{2}=\left\Vert r_{\mathcal{K}_{G}}\right\Vert _{2}$ implies that both
methods are equivalent. From the information aspect, both residuals, $r_{0}$
and $r_{\mathcal{K}_{G}},$ serve as a tool to gain information about changes
in the system dynamics, and, in this regard, they contain the same
information amount. This is the logic consequence of the intimate relation
between the system kernel subspace and observer-based residual generation.
Considering that the online computation cost for generating $r_{0}$ is
considerably lower than that for $r_{\mathcal{K}_{G}},$ it is reasonable to
apply observer-based residual generator (\ref{eq6-14}) for the residual
generation purpose. It is noteworthy that, as a by-product, it provides us
with an optimal observer-based solution for detecting faults in uncertain
dynamic systems, a challenging issue as reported in \cite%
{LD-Automatica-2020,Li-IEEETCST2020}. Concerning with threshold setting,
notice that the basic idea of the existing observer-based methods, roughly
speaking, consists in substituting $y$ in the residual evaluation function
by its upper-bound in the way 
\begin{align*}
\left\Vert r_{0}\right\Vert _{2}& =\left\Vert \left[ 
\begin{array}{cc}
\Delta _{\hat{N}}\text{ } & -\Delta _{\hat{M}}%
\end{array}%
\right] \left[ 
\begin{array}{c}
u \\ 
y%
\end{array}%
\right] \right\Vert _{2}\leq \delta _{\Delta _{K}}\left\Vert \left[ 
\begin{array}{c}
u \\ 
y%
\end{array}%
\right] \right\Vert _{2}\leq \delta _{\Delta _{K}}\left( 1+\delta
_{y}^{2}\right) ^{1/2}\left\Vert u\right\Vert _{2}, \\
& \Longrightarrow J_{th,r_{0}}:=\delta _{\Delta _{K}}\left( 1+\delta
_{y}^{2}\right) ^{1/2}\left\Vert u\right\Vert _{2},
\end{align*}%
where $J_{th,r_{0}}$ is the threshold, and 
\begin{equation*}
\left\Vert y\right\Vert _{2}\leq \delta _{y}\left\Vert u\right\Vert _{2}
\end{equation*}%
for some $\delta _{y}>0$ as an upper-bound of the system dynamics. In light
of our discussion in Subsection 3.3 and calling that 
\begin{equation*}
\left\Vert \mathcal{P}_{\mathcal{K}_{G_{0}}}\left[ 
\begin{array}{c}
u \\ 
y%
\end{array}%
\right] \right\Vert _{2}\leq \left\Vert \left[ 
\begin{array}{c}
u \\ 
y%
\end{array}%
\right] \right\Vert _{2}\leq \left( 1+\delta _{y}^{2}\right)
^{1/2}\left\Vert u\right\Vert _{2},
\end{equation*}%
it can be concluded that the threshold setting (\ref{eq6-4a}) of the
projection-based method delivers better detection performance than the
observer-based methods.

\bigskip

In a nutshell, comparing with existing observer-based schemes, the
projection-based detection method proposed above\ offers better detection
performance with identical online computations.

\subsection{A projection-based fault detection over a finite time interval}

In real applications, the $l_{2}$-norm (ref. to (\ref{eq2-12a})) of the
residual signal has to be approximately computed over a finite time
interval. A reasonable solution of this concern is to substitute the inner
product definition given in (\ref{eq2-12a}) by 
\begin{equation*}
\left\langle x,y\right\rangle =\sqrt{\sum\limits_{k=0}^{N}x^{T}(k)y(k)}%
,N<\infty ,x,y\in \mathcal{H}_{2}.
\end{equation*}%
Accordingly, the design of a projection-based residual generator should be
performed in the framework of time-varying systems. This is a challenging
topic and outside the scope of this work. Below, we propose a practical
solution, which can be realised both in the model-based and data-driven
fashions \cite{Ding2014,Ding2020}.

\bigskip

For our purpose, we first introduce the following input-output (I/O) system
model, which is achieved based on the (nominal) state space model (\ref%
{eq2-2a})-(\ref{eq2-2b}) and widely adopted in subspace technique aided
process identification and data-driven fault detection \cite%
{Huang_2008_book,Ding2014,Ding2020}, 
\begin{align}
y_{s}(k)& =\Gamma _{s}L_{p}z_{p}+H_{u,s}u_{s}(k),z_{p}=\left[ 
\begin{array}{c}
u_{p} \\ 
y_{p}%
\end{array}%
\right] ,  \label{eq6-6} \\
y_{s}(k)& =\left[ 
\begin{array}{c}
y(k-s) \\ 
\vdots  \\ 
y(k)%
\end{array}%
\right] \in \mathbb{R}^{(s+1)m},u_{s}(k)=\left[ 
\begin{array}{c}
u(k-s) \\ 
\vdots  \\ 
u(k)%
\end{array}%
\right] \in \mathbb{R}^{(s+1)p},  \notag \\
y_{p}& =\left[ 
\begin{array}{c}
y(k-s-s_{p}) \\ 
\vdots  \\ 
y(k-s-1)%
\end{array}%
\right] \in \mathbb{R}^{s_{p}m},u_{p}=\left[ 
\begin{array}{c}
u(k-s-s_{p}) \\ 
\vdots  \\ 
u(k-s-1)%
\end{array}%
\right] \in \mathbb{R}^{s_{p}p},  \notag \\
\Gamma _{s}& =\left[ 
\begin{array}{c}
C \\ 
CA \\ 
\vdots  \\ 
CA^{s}%
\end{array}%
\right] \in \mathbb{R}^{(s+1)m\times n},H_{u,s}=\left[ 
\begin{array}{cccc}
D & 0 &  &  \\ 
CB & \ddots  & \ddots  &  \\ 
\vdots  & \ddots  & \ddots  & 0 \\ 
CA^{s-1}B & \cdots  & CB & D%
\end{array}%
\right] \in \mathbb{R}^{(s+1)m\times (s+1)p}  \notag \\
L_{p}& =\left[ 
\begin{array}{cccccc}
A_{K}^{s_{p}-1}B_{K} & \cdots  & B_{K} & A_{K}^{s_{p}-1}K & \cdots  & K%
\end{array}%
\right] ,A_{K}=A-KC,B_{K}=B-KD,  \notag
\end{align}%
where $K$ is the Kalman-filter or observer gain matrix, and $s,s_{p}$ are
two integers typically chosen larger than or equal to $n$. The reader is
referred to, for instance, \cite{Huang_2008_book,Ding2014,Ding2020} for
details about the above model. It is noteworthy that $L_{p}z_{p}$ is a
(good) approximation of the state vector $x(k-s).$\bigskip 

Define 
\begin{equation*}
\mathcal{K}_{G}^{I/O}=\left\{ \left[ 
\begin{array}{c}
z_{p} \\ 
u_{s}(k) \\ 
y_{s}(k)%
\end{array}%
\right] \in \mathbb{R}^{(s+s_{p}+1)\left( m+p\right) }:\left[ 
\begin{array}{ccc}
-\Gamma _{s}L_{p} & -H_{u,s} & I%
\end{array}%
\right] \left[ 
\begin{array}{c}
z_{p} \\ 
u_{s}(k) \\ 
y_{s}(k)%
\end{array}%
\right] =0\hspace{-2pt}\right\} 
\end{equation*}%
as the kernel subspace of the I/O system model (\ref{eq6-6}). Endowed with
the inner product, 
\begin{align*}
\left\langle \alpha ,\beta \right\rangle &
=\sum\limits_{i=1}^{s+s_{p}+1}\alpha ^{T}(i)\beta (i),\alpha =\left[ 
\begin{array}{c}
\alpha (1) \\ 
\vdots  \\ 
\alpha (s+s_{p}+1)%
\end{array}%
\right] ,\beta =\left[ 
\begin{array}{c}
\beta (1) \\ 
\vdots  \\ 
\beta (s+s_{p}+1)%
\end{array}%
\right] \in \mathbb{R}^{(s+s_{p}+1)\left( m+p\right) }, \\
\left\langle \alpha ,\alpha \right\rangle & =:\left\Vert \alpha \right\Vert ,
\end{align*}%
$\mathcal{K}_{G}^{I/O}$ builds a closed subspace in Hilbert space. It is
straightforward that 
\begin{align}
\mathcal{P}_{\mathcal{K}_{G}^{I/O}}& :=I-K_{I/O}^{T}K_{I/O},  \label{eq6-7}
\\
K_{I/O}& =\Sigma ^{-1/2}\left[ 
\begin{array}{ccc}
-\Gamma _{s}L_{p} & -H_{u,s} & I%
\end{array}%
\right] ,\Sigma =\left[ 
\begin{array}{ccc}
-\Gamma _{s}L_{p} & -H_{u,s} & I%
\end{array}%
\right] \left[ 
\begin{array}{c}
-\left( \Gamma _{s}L_{p}\right) ^{T} \\ 
-H_{u,s}^{T} \\ 
I%
\end{array}%
\right]   \notag
\end{align}%
is an orthogonal projection onto $\mathcal{K}_{G}^{I/O}.$ Analogue to Lemma
1, $\mathcal{P}_{\mathcal{K}_{G}^{I/O}}$ has the following properties. Let%
\begin{equation*}
I_{I/O}=\left[ 
\begin{array}{cc}
I & 0 \\ 
0 & I \\ 
\Gamma _{s}L_{p} & H_{u,s}%
\end{array}%
\right] \hat{\Sigma}^{-1/2},\hat{\Sigma}=\left[ 
\begin{array}{ccc}
I & 0 & \left( \Gamma _{s}L_{p}\right) ^{T} \\ 
0 & I & H_{u,s}^{T}%
\end{array}%
\right] \left[ 
\begin{array}{cc}
I & 0 \\ 
0 & I \\ 
\Gamma _{s}L_{p} & H_{u,s}%
\end{array}%
\right] .
\end{equation*}%
It holds 
\begin{gather}
K_{I/O}I_{I/O}=0,\left[ 
\begin{array}{c}
K_{I/O} \\ 
I_{I/O}^{T}%
\end{array}%
\right] \left[ 
\begin{array}{cc}
K_{I/O}^{T} & I_{I/O}%
\end{array}%
\right] =\left[ 
\begin{array}{cc}
I & 0 \\ 
0 & I%
\end{array}%
\right] \Longrightarrow   \notag \\
\left[ 
\begin{array}{cc}
K_{I/O}^{T} & I_{I/O}%
\end{array}%
\right] \left[ 
\begin{array}{c}
K_{I/O} \\ 
I_{I/O}^{T}%
\end{array}%
\right] =I\Longleftrightarrow \mathcal{P}_{\mathcal{K}%
_{G}^{I/O}}=I-K_{I/O}^{T}K_{I/O}=I_{I/O}I_{I/O}^{T}.  \label{eq6-8}
\end{gather}%
In fact, $K_{I/O}$ and $I_{I/O}$ are the normalised kernel and image
representations of the I/O model (\ref{eq6-6}). By means of $\mathcal{P}_{%
\mathcal{K}_{G}^{I/O}}$, a projection-based residual can be generated as
follows%
\begin{equation}
r_{I/O}(k)=\left( \mathcal{I}-\mathcal{P}_{\mathcal{K}_{G}^{I/O}}\right) %
\left[ 
\begin{array}{c}
z_{p} \\ 
u_{s}(k) \\ 
y_{s}(k)%
\end{array}%
\right] =K_{I/O}^{T}K_{I/O}\left[ 
\begin{array}{c}
z_{p} \\ 
u_{s}(k) \\ 
y_{s}(k)%
\end{array}%
\right] .  \label{eq6-9}
\end{equation}%
The corresponding evaluation function is%
\begin{align}
\left\Vert r_{I/O}(k)\right\Vert & =\left\Vert K_{I/O}^{T}K_{I/O}\left[ 
\begin{array}{c}
z_{p} \\ 
u_{s}(k) \\ 
y_{s}(k)%
\end{array}%
\right] \right\Vert =\left\Vert K_{I/O}\left[ 
\begin{array}{c}
z_{p} \\ 
u_{s}(k) \\ 
y_{s}(k)%
\end{array}%
\right] \right\Vert =\left\Vert r_{s}(k)\right\Vert ,  \label{eq6-9a} \\
r_{s}(k)& =K_{I/O}\left[ 
\begin{array}{c}
z_{p} \\ 
u_{s}(k) \\ 
y_{s}(k)%
\end{array}%
\right] =\Sigma ^{-1/2}\left( y_{s}(k)-\Gamma
_{s}L_{p}z_{p}-H_{u,s}u_{s}(k)\right) \in \mathbb{R}^{(s+1)m}.
\label{eq6-9b}
\end{align}%
Equations (\ref{eq6-9a})-(\ref{eq6-9b}) show that, for the detection
purpose, generating $r_{s}(k)$ is sufficient. It is known that the residual
vector $r_{s}(k)$ is widely used in the so-called parity space or
data-driven methods \cite{Ding2020}.

\bigskip

Next, threshold setting is addressed. On the assumption that model
uncertainties cause variations in the nominal system model and lead to%
\begin{gather}
y_{s}(k)=\left( \Gamma _{s}L_{p}+\Delta _{x}\right) z_{p}+\left(
H_{u,s}+\Delta _{u}\right) u_{s}(k),  \label{eq6-11a} \\
\Delta K_{I/O}=\Sigma ^{-1/2}\left[ 
\begin{array}{ccc}
-\Delta _{x} & -\Delta _{u} & 0%
\end{array}%
\right] , \\
\sup_{\Delta _{x},\Delta _{u}}\left\Vert \Delta K_{I/O}\right\Vert
_{2}=\sup_{\Delta _{x},\Delta _{u}}\bar{\sigma}\left( \Delta K_{I/O}\right)
=\delta _{I/O}<1  \label{eq6-11b}
\end{gather}%
with $\Delta _{x},\Delta _{u}$ representing the uncertainties, the threshold
is set to be%
\begin{equation*}
J_{th}=\sup_{\left\Vert \Delta K_{I/O}\right\Vert _{2}\leq \delta
_{I/O}}\left\Vert r_{I/O}(k)\right\Vert .
\end{equation*}%
Departing from the relations%
\begin{gather*}
\left\Vert r_{I/O}(k)\right\Vert =\left\Vert r_{s}(k)\right\Vert
,y_{s}(k)=\left( \Gamma _{s}L_{p}+\Delta _{x}\right) z_{p}+\left(
H_{u,s}+\Delta _{u}\right) u_{s}(k) \\
r_{s}(k)=\Sigma ^{-1/2}\left( \Delta _{x}z_{p}+\Delta _{u}u_{s}(k)\right)
=-\Delta K_{I/O}\left[ 
\begin{array}{c}
z_{p} \\ 
u_{s}(k) \\ 
y_{s}(k)%
\end{array}%
\right] ,
\end{gather*}%
and definitions%
\begin{align*}
\mathcal{K}_{G,\Delta K}^{I/O}& =\left\{ \left[ 
\begin{array}{c}
z_{p} \\ 
u_{s}(k) \\ 
y_{s}(k)%
\end{array}%
\right] :\left[ 
\begin{array}{ccc}
-\Gamma _{s}L_{p}-\Delta _{x} & -H_{u,s}-\Delta _{u} & I%
\end{array}%
\right] \left[ 
\begin{array}{c}
z_{p} \\ 
u_{s}(k) \\ 
y_{s}(k)%
\end{array}%
\right] =0\hspace{-2pt}\right\} , \\
\mathcal{K}_{G,\delta }^{I/O}& =\left\{ \mathcal{K}_{G,\Delta
K}^{I/O}:\left\Vert \Delta K_{I/O}\right\Vert _{2}\leq \delta _{I/O}\hspace{%
-2pt}\right\} ,
\end{align*}%
we have%
\begin{equation}
\forall \left[ 
\begin{array}{c}
u \\ 
y%
\end{array}%
\right] \in \mathcal{K}_{G,\Delta K}^{I/O}\subset \mathcal{K}_{G,\delta
}^{I/O},\left\Vert r_{s}(k)\right\Vert \leq \delta _{I/O}\left\Vert \left[ 
\begin{array}{c}
z_{p} \\ 
u_{s}(k) \\ 
y_{s}(k)%
\end{array}%
\right] \right\Vert .  \label{eq6-17}
\end{equation}%
In the light of Theorem \ref{Theo3-1}, the following theorem becomes obvious.

\begin{Theo}
\label{Theo6-2}Given the I/O model (\ref{eq6-6}) with model uncertainty
satisfying (\ref{eq6-11a})-(\ref{eq6-11b}), and suppose that
projection-based residual generator (\ref{eq6-9}) and evaluation function (%
\ref{eq6-9a}) are used for the detection purpose, then the corresponding
threshold is given by 
\begin{align}
J_{th}& =\frac{\delta _{I/O}}{\sqrt{1-\delta _{I/O}^{2}}}\left\Vert \mathcal{%
P}_{\mathcal{K}_{G}^{I/O}}\left[ 
\begin{array}{c}
z_{p} \\ 
u_{s}(k) \\ 
y_{s}(k)%
\end{array}%
\right] \right\Vert  \label{eq6-18a} \\
& =\frac{\delta _{I/O}}{\sqrt{1-\delta _{I/O}^{2}}}\left( \left\Vert \left[ 
\begin{array}{c}
z_{p} \\ 
u_{s}(k) \\ 
y_{s}(k)%
\end{array}%
\right] \right\Vert ^{2}-\left\Vert r_{s}(k)\right\Vert ^{2}\right) ^{1/2}.
\label{eq6-18b}
\end{align}
\end{Theo}

\begin{proof}
Since 
\begin{equation*}
\left\Vert \left[ 
\begin{array}{c}
z_{p} \\ 
u_{s}(k) \\ 
y_{s}(k)%
\end{array}%
\right] \right\Vert ^{2}=\left\Vert \mathcal{P}_{\mathcal{K}_{G}^{I/O}}\left[
\begin{array}{c}
z_{p} \\ 
u_{s}(k) \\ 
y_{s}(k)%
\end{array}%
\right] \right\Vert ^{2}+\left\Vert r_{I/O}(k)\right\Vert ^{2},
\end{equation*}%
it follows from (\ref{eq6-17}) that%
\begin{equation*}
\forall \left[ 
\begin{array}{c}
u \\ 
y%
\end{array}%
\right] \in \mathcal{K}_{G,\Delta K}^{I/O}\subset \mathcal{K}_{G,\delta
}^{I/O},\left\Vert r_{s}(k)\right\Vert ^{2}\leq \delta _{I/O}^{2}\left(
\left\Vert \mathcal{P}_{\mathcal{K}_{G}^{I/O}}\left[ 
\begin{array}{c}
z_{p} \\ 
u_{s}(k) \\ 
y_{s}(k)%
\end{array}%
\right] \right\Vert ^{2}+\left\Vert r_{I/O}(k)\right\Vert ^{2}\right) ,
\end{equation*}%
which leads to (\ref{eq6-18a}) and (\ref{eq6-18b}). 
\end{proof}

\section{An experimental study}

\subsection{Description of the experimental system}

Due to their typical characteristics of a chemical process, three-tank
systems are widely accepted as a benchmark process in laboratories for
process control and fault diagnosis. The three-tank system used in our
experimental study is laboratory setup TTS20, whose model and the model
parameters are summarised as follows \cite{Ding2020}:%
\begin{align*}
\mathcal{A}\dot{h}_{1}& =Q_{1}-Q_{13},\mathcal{A}\dot{h}%
_{2}=Q_{2}+Q_{32}-Q_{20},\mathcal{A}\dot{h}_{3}=Q_{13}-Q_{32}, \\
Q_{13}& =a_{1}s_{13}\text{sgn}(h_{1}-h_{3})\sqrt{2g|h_{1}-h_{3}|}, \\
Q_{32}& =a_{3}s_{23}\text{sgn}(h_{3}-h_{2})\sqrt{2g|h_{3}-h_{2}|}%
,Q_{20}=a_{2}s_{0}\sqrt{2gh_{2}},
\end{align*}%
where $Q_{1},Q_{2}$ are incoming mass flow (cm$^{3}$/s), $Q_{ij}$ is the
mass flow (cm$^{3}$/s) from the $i$-th tank to the $j$-th tank, $%
h_{i}(t),i=1,2,3,$ are the water level (cm) in the $i$-th tank and
measurement variables, and $s_{13}=s_{23}=s_{0}=s_{n.}$

\begin{table}[h]
\caption{Parameters of TTS20}
\label{table 18-1}\centering\vspace{0.2cm} 
\begin{tabular}{|c|c|c|c|}
\hline
Parameters & Symbol & Value & Unit \\ \hline
cross section area of tanks & $\mathcal{A}$ & $154$ & cm$^{2}$ \\ \hline
cross section area of pipes & $s_{n}$ & $0.5$ & cm$^{2}$ \\ \hline
max. height of tanks & $H_{max}$ & $62$ & cm \\ \hline
max. flow rate of pump 1 & $Q_{1_{max}}$ & $100$ & cm$^{3}$/s \\ \hline
max. flow rate of pump 2 & $Q_{2_{max}}$ & $100$ & cm$^{3}$/s \\ \hline
coeff. of flow for pipe 1 & $a_{1}$ & $0.45$ &  \\ \hline
coeff. of flow for pipe 2 & $a_{2}$ & $0.60$ &  \\ \hline
coeff. of flow for pipe 3 & $a_{3}$ & $0.45$ &  \\ \hline
\end{tabular}%
\end{table}
\bigskip

For our purpose of testing the projection-based fault detection schemes, the
above nonlinear model is linearised at the operating point $h_{1}=30$cm, $%
h_{2}=20$cm, $h_{3}=24$cm and discretised with a sampling time equal to $5$%
s. The resulted nominal model is given by%
\begin{align*}
x(k+1)& =Ax(k)+Bu(k),x(k)=\left[ 
\begin{array}{c}
h_{1}(k) \\ 
h_{2}(k) \\ 
h_{3}(k)%
\end{array}%
\right] ,y(k)=\left[ 
\begin{array}{c}
x_{1}(k) \\ 
x_{2}(k)%
\end{array}%
\right] ,u(k)=\left[ 
\begin{array}{c}
Q_{1}(k) \\ 
Q_{2}(k)%
\end{array}%
\right] , \\
A& =\left[ {%
\begin{array}{ccc}
0.9272 & -0.0012 & 0.0605 \\ 
-0.0064 & 0.8829 & 0.0561 \\ 
0.0645 & 0.0632 & 0.8695%
\end{array}%
}\right] ,B=\left[ {%
\begin{array}{cc}
0.03126 & -0.0000 \\ 
-0.0001 & 0.0305 \\ 
0.0011 & 0.0011%
\end{array}%
}\right] .
\end{align*}%
In order to regulate the water level in tank 1 and tank 2, an observer-based
state feedback controller is adopted,%
\begin{align}
\hat{x}(k+1)& =\left( A-LC\right) \hat{x}(k)+Bu(k)+Ly(k),  \label{eq7-4a} \\
u(k)& =F\hat{x}(k)+v(k),  \label{eq7-4b}
\end{align}%
where $F,L$ are designed using LQ optimal control and Kalman filter
algorithms, respectively, which results in%
\begin{equation}
F=\left[ {%
\begin{array}{ccc}
-8.8470 & -0.0663 & -1.3821 \\ 
-0.0406 & -7.8259 & -1.1961%
\end{array}%
}\right] ,L=\left[ {%
\begin{array}{cc}
0.7125 & 0.0213 \\ 
0.0070 & 0.7202 \\ 
0.0661 & 0.0917%
\end{array}%
}\right] .  \label{eq7-4}
\end{equation}%
$v$ is the reference signal that is set corresponding to the operating point 
$h_{1}=30$cm, $h_{2}=20$cm and given by 
\begin{equation*}
v(k)=\left[ {%
\begin{array}{c}
323.2805 \\ 
219.1368%
\end{array}%
}\right] .
\end{equation*}

\subsection{Design of fault detection systems}

The main objective of experimental test is to compare detection performance
of observer-based and projection-based fault detection systems. To this end,
three fault detection systems are designed as follows:

\begin{itemize}
\item A projection-based detection system, as proposed in Subsection 6.1.
According to the relation%
\begin{equation*}
\left\Vert r_{\mathcal{K}_{G_{0}}}\right\Vert _{2}=\left\Vert
r_{0}\right\Vert _{2},
\end{equation*}%
the residual generation is realised using an observer-based residual
generator with the observer gain $L_{0}$ and post filter $W_{0}$ computed by
means of (\ref{eq2-10a}) and (\ref{eq2-10b}), which results in%
\begin{equation*}
L_{0}=\left[ {%
\begin{array}{cc}
0.0094 & 0.0015 \\ 
0.0015 & 0.0046 \\ 
0.0046 & 0.0024%
\end{array}%
}\right] ,W_{0}=\left[ {%
\begin{array}{cc}
0.9951 & -0.0008 \\ 
-0.0008 & 0.9974%
\end{array}%
}\right] .
\end{equation*}%
The computation of $\left\Vert r_{0}\right\Vert _{2}$ is approximated by%
\begin{equation*}
\left\Vert r_{0}\right\Vert _{2}=\left(
\sum\limits_{i=k}^{k+120}r_{0}^{T}(k+i)r_{0}(k+i)\right) ^{1/2}
\end{equation*}%
with a moving evaluation window. As discussed in Subsection \ref%
{subsection3-3}, the residual evaluation function is built by a
normalisation of $r_{0},$%
\begin{equation*}
J_{N}=\frac{\left\Vert r_{0}\right\Vert _{2}}{\left\Vert \left[ 
\begin{array}{c}
u \\ 
y%
\end{array}%
\right] \right\Vert _{2}}=\left( \frac{\sum%
\limits_{i=k}^{k+120}r_{0}^{T}(k+i)r_{0}(k+i)}{\sum\limits_{i=k}^{k+120}%
\left( u^{T}(k+i)u(k+i)+y^{T}(k+i)y(k+i)\right) }\right) ^{1/2}.
\end{equation*}%
The corresponding threshold is set following Theorem \ref{Theo6-1} and by
means of a normalisation, 
\begin{equation*}
J_{th,N}=\frac{\delta _{\Delta _{K}}}{\sqrt{1-\delta _{\Delta _{K}}^{2}}}%
\left( 1-J_{N}^{2}\right) ^{1/2}.
\end{equation*}%
To determine an upper-bound of left-coprime factor uncertainty $\Delta _{K}$
(refer to (\ref{eq6-12a})-(\ref{eq6-12c})) $\delta _{\Delta _{K}}$, possible
variations from $\pm 3.3\%$ in the system parameters $\mathcal{A}%
,s_{n},a_{1},a_{2},a_{3}$ are considered, which leads to 
\begin{equation*}
\Delta _{K}=\left[ 
\begin{array}{cc}
-\Delta _{\hat{N}} & \Delta _{\hat{M}}%
\end{array}%
\right] ,\left\Vert \Delta _{K}\right\Vert _{\infty }\leq \delta _{\Delta
_{K}}=0.052.
\end{equation*}

\item An observer-based fault detection system. Here, the Kalman filter
applied in the observer-based controller is used for residual generation,
with gain matrix $L$ given by (\ref{eq7-4}), and the generated residual is
denoted by%
\begin{equation*}
r_{0,k}(k)=y(k)-C\hat{x}(k).
\end{equation*}%
The threshold is set using the algorithm given in Theorem 2 \cite%
{LD-Automatica-2020}, 
\begin{equation*}
J_{th,k}=\frac{\beta \delta _{\Delta _{K}}}{1-\delta _{\Delta _{K}}b}%
\left\Vert v\right\Vert _{2},\left\Vert v\right\Vert _{2}=\left(
\sum\limits_{i=k}^{k+120}v^{T}(k+i)v(k+i)\right) ^{1/2},
\end{equation*}%
where 
\begin{equation*}
\beta =\left\Vert \left( I-C(zI-A+LC)^{-1}(L-L_{0})\right)
W_{0}^{-1}\right\Vert _{\infty }=1.6257,b=\left\Vert \left[ 
\begin{array}{c}
U_{0} \\ 
V_{0}%
\end{array}%
\right] \right\Vert _{\infty }=7.9179,
\end{equation*}%
$\left( V_{0},U_{0}\right) $ is the normalised RCP of the controller (\ref%
{eq7-4a})-(\ref{eq7-4b}) (refer to (\ref{eq2-5a})).

\item An observer-based fault detection system, which consists of an
observer-based residual generator whose SKR, different from a Kalman filter,
is normalised with $L_{0},W_{0}$ as the observer gain matrix and
post-filter. As a result, it delivers $r_{0}$ as the residual signal.
Similar to the above observer-based fault detection system, the threshold
setting is achieved using the algorithm given in Theorem 2 \cite%
{LD-Automatica-2020}, which leads to%
\begin{equation*}
J_{th}=\frac{\delta _{\Delta _{K}}}{1-\delta _{\Delta _{K}}b}\left\Vert
v\right\Vert _{2},\left\Vert v\right\Vert _{2}=\left(
\sum\limits_{i=k}^{k+120}v^{T}(k+i)v(k+i)\right) ^{1/2}.
\end{equation*}
\end{itemize}

\subsection{Experimental results and analysis}

To test the proposed projection-based fault detection system and compare it
with the existing observer-based ones, the following two faults are realised
directly on the laboratory setup TTS20:

\begin{itemize}
\item a 5\% leakage fault in tank 1,

\item a 7\% plugging fault in pipe 1 (connecting tank 1 and tank 3).
\end{itemize}

Figures 3 - 8 show the test results. In detail,

\begin{itemize}
\item the top figure in Figure 3 gives the responses of the residual
evaluation function $J_{N}$ and the threshold $J_{th,N}$ during fault-free
operation over the time interval $[1500$s, $2150$s$]$ and faulty operation
(5\% leakage) over $[2150$s, $4000$s$].$ A successful fault detection is
achieved. The threshold $J_{th,N}$ is shown in the bottom figure to
demonstrate that it decreases as the residual increases;

\item Figure 4 shows the residual evaluation function $\left\Vert
r_{0}\right\Vert _{2}$ and the threshold $J_{th}$ over the same intervals
with the 5\% leakage fault. A successful fault detection is demonstrated as
well;

\item the top figure in Figure 5 showcases the residual evaluation function $%
\left\Vert r_{0,k}\right\Vert _{2}$ and the threshold $J_{th,k}$ over the
same intervals with the 5\% leakage fault, and demonstrates, together with
the bottom figure, that the fault cannot be detected;

\item Figures 6 - 8 give, analogue to Figures 3 - 5, the responses of the
residual evaluation functions $J_{N},\left\Vert r_{0}\right\Vert
_{2},\left\Vert r_{0,k}\right\Vert _{2}$ and the corresponding thresholds $%
J_{th,N},J_{th},J_{th,k}$ during fault-free operation over the time interval 
$[2500$s, $4000$s$]$ and faulty operation (7\% plugging) over $[4000$s, $%
6500 $s$],$ respectively. It is apparent that the fault can be successfully
detected by the projection-based detection system, while both observer-based
detection systems fail.
\end{itemize}

\begin{figure}[h]
\centering\includegraphics[width=10cm,height=7cm]{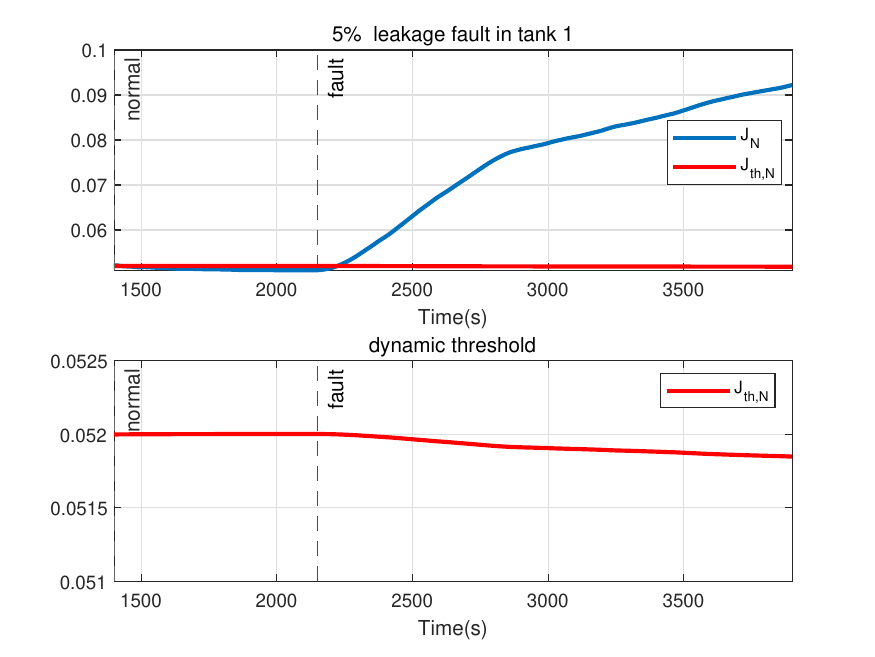}
\caption{Detection of a leakage fault in tank 1 using the projection-based
detection system}
\end{figure}
\begin{figure}[h]
\centering\includegraphics[width=10cm,height=7cm]{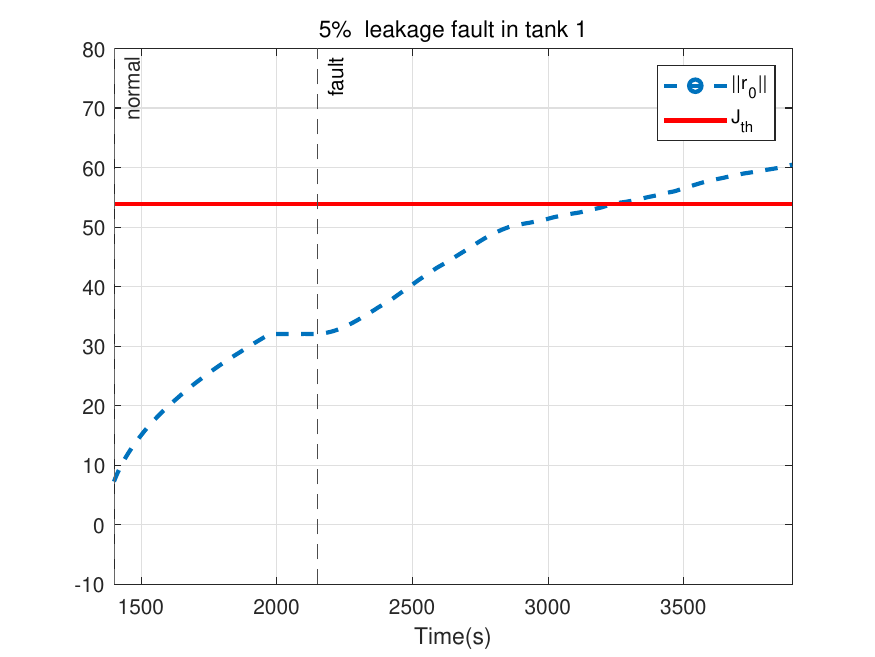}
\caption{Detection of a leakage fault in tank 1 using the observer-based
detection system (a normalised SKR)}
\end{figure}

\begin{figure}[h]
\centering\includegraphics[width=10cm,height=7cm]{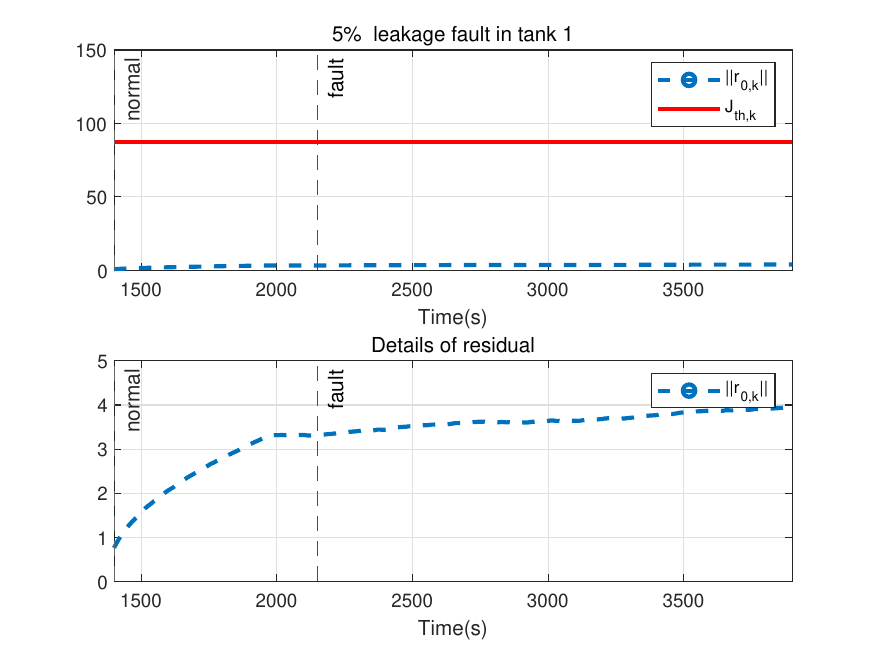}
\caption{Detection of a leakage fault in tank 1 using the observer-based
detection system (Kalman filter-based)}
\end{figure}

\begin{figure}[h]
\centering\includegraphics[width=10cm,height=7cm]{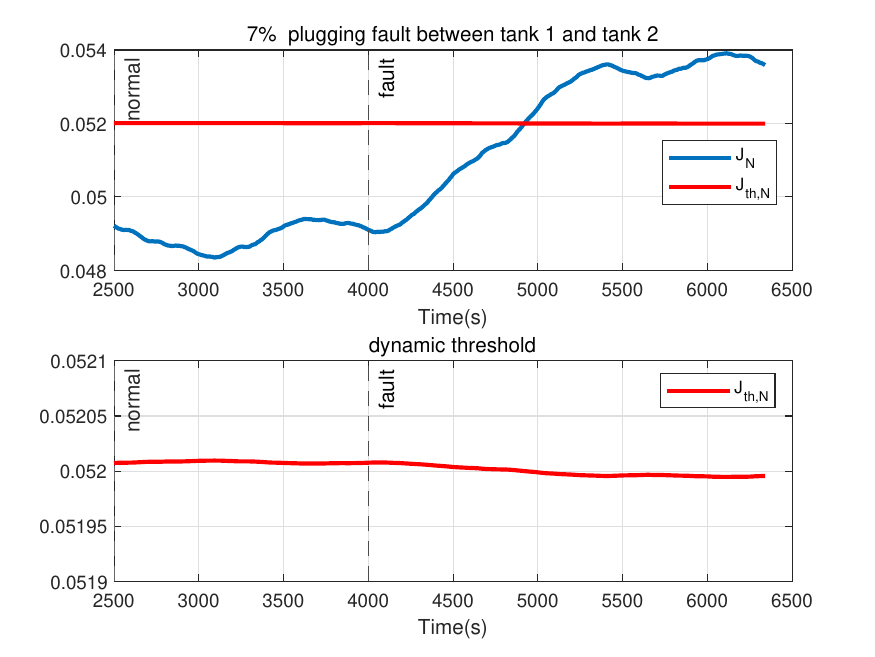}
\caption{Detection of a plugging fault in pipe 1 using the projection-based
detection system}
\end{figure}

\begin{figure}[h]
\centering\includegraphics[width=10cm,height=7cm]{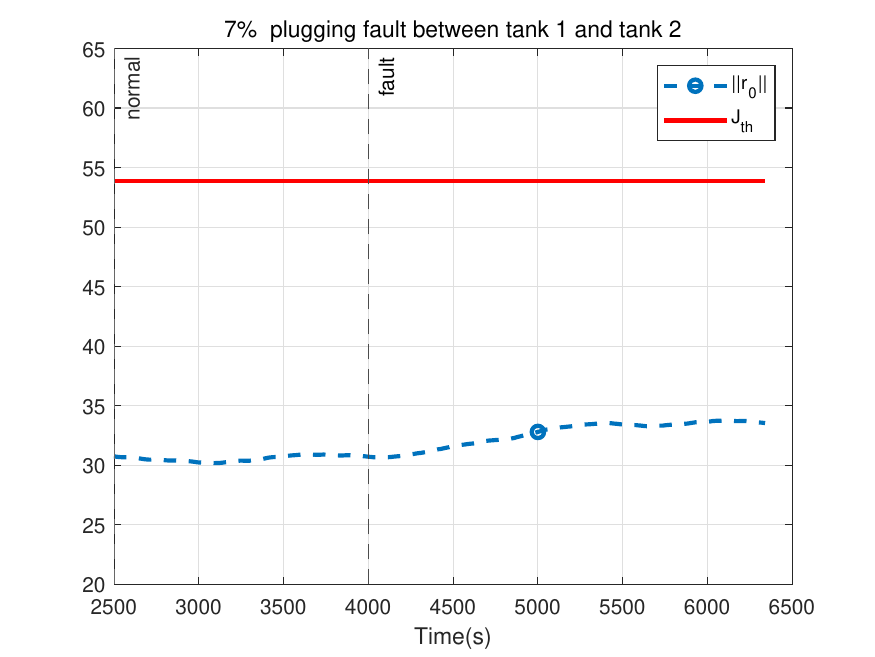}
\caption{Detection of a plugging fault in pipe 1 using the observer-based
detection system (a normalised SKR)}
\end{figure}
\begin{figure}[h]
\centering\includegraphics[width=10cm,height=7cm]{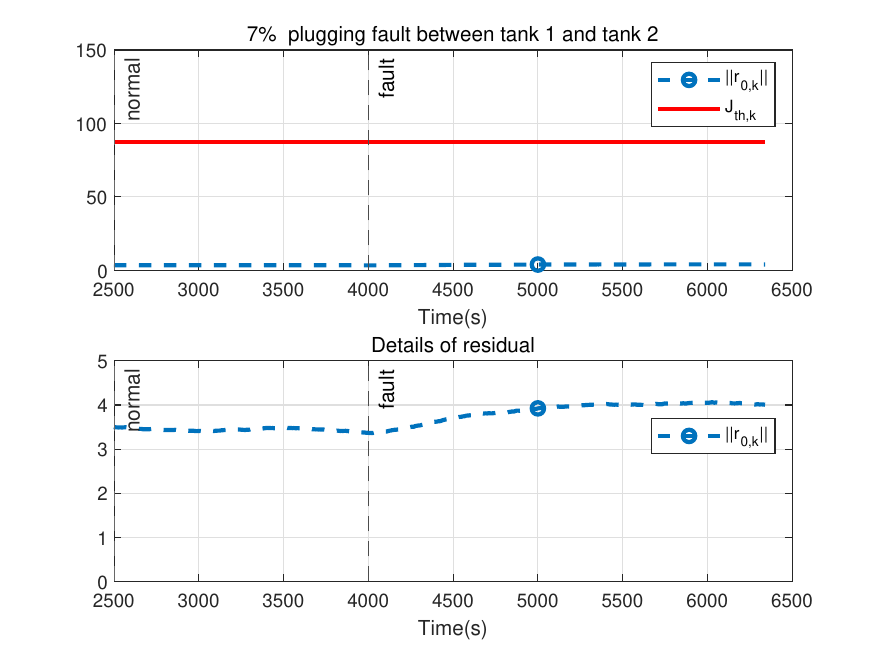}
\caption{Detection of a plugging fault in pipe 1 using the observer-based
detection system (Kalman filter-based)}
\end{figure}

Next, we briefly analyse the above experimental results. The first
conclusion is that the Kalman filter-based detection system is less capable
of dealing with robust fault detection, when model uncertainties exists. The
major reason lies in the residual generator setting. Recall that the
projection-based method proposed in Subsection 6.1 leads to the normalised
SKR as the (optimal) setting of the residual generator. From the system
analysis point of view, \cite{Ding2020} (Subsection 9.2.3) illustrates and
verifies this conclusion. The normalised SKR is co-inner \cite{Ding2020},
that is, all singular values of SKR from the input pair $\left( u,y\right) $
to the residual $r_{0}$ are identically equal to $1$. Checking the relation 
\begin{equation*}
r_{0,k}(z)=R(z)r_{0}(z),R(z)=\left( I-C(zI-A+LC)^{-1}(L-L_{0})\right)
W_{0}^{-1}
\end{equation*}%
and the maximal and minimal singular values of $R(z)$,%
\begin{equation*}
\sigma _{\max }\left( R\right) =\left\Vert R(z)\right\Vert _{\infty
}=1.6257,\sigma _{\min }\left( R\right) =0.0556,
\end{equation*}%
reveals that the Kalman filter-based residual generation is far away from
optimal. We would like to emphasise that \cite{DingIAS00} has proved that an
optimal residual generation is achieved if the transfer function matrix from
the input to the residual is co-inner.

\bigskip

Comparing the results in Figures 3 and 4 as well as 6 and 7 demonstrates
obviously that the projection-based detection system is more sensitive to
the faults than the observer-based one. Since both detection systems have
the same residual generator, i.e. the normalised SKR, the different
threshold settings are the result of the different detection performance. To
put it in a nutshell, it can be concluded that the projection-based
detection methods result in optimal fault detection performance.

\section{Concluding remarks}

In the previous sections, we have presented fault diagnosis schemes in the
projection-based framework, including the basic fault detection (one-class
classification) scheme, two detection methods for feedback control systems,
binary and multi-class fault classification methods towards fault detection
and isolation, as well as two modified fault detection methods. Although
they have been developed for different application purposes under possibly
different system configuration assumptions, they have one point in common,
namely the corresponding fault diagnosis system design follows the uniform
procedure with the steps: (i) definition and determination of an orthogonal
projection operator, (ii) construction of the projection-based residual
generator and realisation of online implementation algorithm, mainly
consisting of an observer plus, possibly, an additional filter, and (iii)
gap metric-based threshold settings. All these results are summarised in
Theorems \ref{Theo3-1}-\ref{Theo4-1} and \ref{Theo6-1}-\ref{Theo6-2}. It is
remarkable that all resulted fault diagnosis systems are optimal with
respect to the classification distance metric, and the above design
procedure can be applied both in the model-based and data-driven fashions.
So far, the basic targets of our work, as described in Introduction, have
been successfully reached.

\bigskip

Moreover, we have illustrated the major differences between the
projection-based and observer-based fault detection schemes, and
demonstrated the advantage of the projection-based scheme in enhancing fault
detectability. In this context, a modified approach has also been proposed
for the design of projection-based fault detection systems that are
comparable with observer-based ones with regard to the online computation
and offer better detection performance.

\bigskip

In our study on projection-based binary (fault) classification, it has been
revealed that there generally exists an overlapping subspace of fault-free
and faulty operations, which is described by an one-to-one mapping between
the image subspaces of the fault-free and faulty system dynamics. This
result is helpful for us to understand mechanisms of incipient and
intermittent faults. Accordingly, an optimal classification scheme has been
proposed, which is also the basic algorithm for multi-class fault
classification (fault isolation).

\bigskip

Our work in this paper has focused on detecting and isolating multiplicative
faults in dynamic systems with (multiplicative/parametric) uncertainties,
motivated by the fact that there exist no general and systematic solutions
for the relevant issues. It is natural to raise a question if
projection-based methods can be applied to systems with additive
disturbances (unknown inputs) and faults. Below, we briefly illustrate such
a case. For our purpose, consider the nominal model (\ref{eq2-2a})-(\ref%
{eq2-2b}) with $l_{2}$-norm bounded unknown input vector $d\in \mathbb{R}%
^{k_{d}},$ and unknown fault vector $f\in \mathbb{R}^{k_{f}},$ 
\begin{gather}
x(k+1)=Ax(k)+Bu(k)+E_{f}f(k)+E_{d}d(k),  \label{eq3-14a} \\
y(k)=Cx(k)+Du(k)+F_{f}f(k)+F_{d}d(k),\left\Vert d\right\Vert _{2}\leq \delta
_{d},  \label{eq3-14b}
\end{gather}%
where $E_{d},F_{d},E_{f},F_{f}$ are known matrices and $\delta _{d}$ is the
known upper bound. The corresponding optimal fault detection problem has
been extensively studied \cite{HouUKACC96,DingIAS00,WYL2007}. In order to
compare with the existing results, we only consider projection $\mathcal{P}_{%
\mathcal{K}_{G}}$ defined by (\ref{eq6-1}) and assume that 
\begin{equation}
\forall \theta \in \lbrack 0,2\pi ],rank\left[ 
\begin{array}{cc}
A-e^{j\theta }I & E_{d} \\ 
C & F_{d}%
\end{array}%
\right] =n+m.  \label{eq3-15}
\end{equation}%
It turns out%
\begin{equation}
\left\Vert r_{\mathcal{K}_{G}}\right\Vert _{2}=\left\Vert \mathcal{L}%
_{K_{G_{0}}}\left[ 
\begin{array}{c}
u \\ 
y%
\end{array}%
\right] \right\Vert _{2}=\left\Vert r_{0}\right\Vert _{2}.  \label{eq3-16}
\end{equation}%
It is well-known \cite{Ding2008} that 
\begin{align*}
r_{0}(z)& =\hat{M}_{0}(z)y(z)-\hat{N}_{0}(z)u(z)=\hat{N}_{d}(z)d(z), \\
\hat{N}_{d}(z)& =\left( A-LC,E_{d}-LF_{d},WC,WF_{d}\right) \in \mathcal{RH}%
_{\infty }.
\end{align*}%
Consequently, the threshold is set to be \cite{DingAUTO94}%
\begin{equation}
J_{th}=\sup_{\left\Vert d\right\Vert _{2}\leq \delta _{d}}\left\Vert
r_{0}\right\Vert _{2}=\left\Vert \hat{N}_{d}\right\Vert _{\infty }\delta
_{d}.  \label{eq3-18}
\end{equation}%
On the other hand, it is known that the above threshold setting is
considerably conservative \cite{Ding2008}. Alternatively, we introduce a
projection-based method to improve the threshold setting. Note that 
\begin{equation*}
\mathcal{R}_{d}=\left\{ d\in \mathcal{L}_{2}:\hat{N}_{d}d\neq 0\right\} 
\end{equation*}%
is a subspace in $\mathcal{L}_{2}$, and the orthogonal projection of $d$
onto $\mathcal{R}_{d}$ is given by 
\begin{equation*}
\mathcal{P}_{\mathcal{R}_{d}}:\mathcal{L}_{2}\rightarrow \mathcal{L}_{2},%
\mathcal{P}_{\mathcal{R}_{d}}=\mathcal{L}_{\hat{N}_{d,0}^{\sim }}\mathcal{L}%
_{\hat{N}_{d,0}},\hat{N}_{d,0}(z)\hat{N}_{d,0}^{\sim }(z)=I
\end{equation*}%
with $\hat{N}_{d,0}$ satisfying \cite{Ding2008}%
\begin{gather}
\hat{N}_{d,0}(z)=\left( A-L_{d}C,E_{d}-L_{d}F_{d},W_{d}C,W_{d}F_{d}\right) ,
\label{eq3-16a} \\
L_{d}=\left( AXC^{T}+E_{d}F_{d}^{T}\right) \left(
CXC^{T}+F_{d}F_{d}^{T}\right) ^{-1},W_{d}=\left(
CXC^{T}+F_{d}F_{d}^{T}\right) ^{-1/2},  \label{eq3-16b} \\
X=AXA^{T}+E_{d}E_{d}^{T}-L_{d}\left( CXC^{T}+F_{d}F_{d}^{T}\right) L_{d}^{T}.
\notag
\end{gather}%
As shown in \cite{Ding2008}, there exists a post-filter $R(z)$ so that%
\begin{align*}
R\hat{N}_{d}& =\hat{N}_{d,0}\Longrightarrow Rr_{0}=\hat{N}_{d,0}d=:\bar{r}%
_{0}, \\
R(z)& =\left( A-L_{d}C,\left( L_{d}-L\right)
W^{-1},-W_{d}C,W_{d}W^{-1}\right) \in \mathcal{RH}_{\infty }.
\end{align*}%
Note that 
\begin{equation}
\left\Vert \mathcal{P}_{\mathcal{R}_{d}}d\right\Vert _{2}=\left\Vert \hat{N}%
_{d,0}d\right\Vert _{2}=\left\Vert \bar{r}_{o}\right\Vert _{2},\forall d\in 
\mathcal{R}_{d},\left\Vert \mathcal{P}_{\mathcal{R}_{d}}d\right\Vert
_{2}=\left\Vert d\right\Vert _{2}.  \label{eq3-17a}
\end{equation}%
The second equation in (\ref{eq3-17a}) leads to 
\begin{equation}
J_{th}=\sup_{\left\Vert d\right\Vert _{2}\leq \delta _{d}}\left\Vert \bar{r}%
_{o}\right\Vert _{2}=\delta _{d}.  \label{eq3-17}
\end{equation}%
To demonstrate that $\bar{r}_{o}$ delivers a better fault detection
performance in comparison with $r_{0},$ it is sufficient to check $%
\left\Vert \bar{r}_{o}\right\Vert _{2}$ vs. $\left\Vert r_{0}\right\Vert
_{2}/\left\Vert \hat{N}_{d}\right\Vert _{\infty },$ i.e. on the condition of
the same threshold equal to $\delta _{d}.$ It can be seen that%
\begin{equation*}
\gamma \left\Vert r_{0}\right\Vert _{2}=\gamma \left\Vert \hat{N}%
_{d}d\right\Vert _{2}\leq \left\Vert d\right\Vert _{2}=\left\Vert \bar{r}%
_{o}\right\Vert _{2},\gamma =\frac{1}{\left\Vert \hat{N}_{d}\right\Vert
_{\infty }}.
\end{equation*}%
This implies that $\bar{r}_{o}$ delivers higher fault detectability than $%
\gamma r_{0}.$

\bigskip

We would like to call the reader's attention that the above fault detection
system with the residual generator (\ref{eq3-16a})-(\ref{eq3-16b}) and
threshold (\ref{eq3-17}) is the so-called unified solution \cite%
{DingIAS00,Ding2008}, which can be iteratively approached using linear
matrix inequality solutions \cite{HouUKACC96,WYL2007}. The above result
gives a geometric solution and interpretation of optimal detection for
systems with additive unknown inputs and faults in terms of projection-based
detection methods.

\bigskip

Our final remark is dedicated to the following MMP, 
\begin{equation}
\inf_{Q\in \mathcal{H}_{\infty }}\left\Vert \left[ 
\begin{array}{cc}
-\hat{N}_{1} & \hat{M}_{1}%
\end{array}%
\right] -Q\left[ 
\begin{array}{cc}
-\hat{N}_{2} & \hat{M}_{2}%
\end{array}%
\right] \right\Vert _{\infty },  \label{eq7-1}
\end{equation}%
where $\left[ 
\begin{array}{cc}
-\hat{N}_{i} & \hat{M}_{i}%
\end{array}%
\right] $ is the normalised SKR of system $G_{i}=\hat{M}_{i}^{-1}\hat{N}%
_{i},i=1,2,$ whose solution is defined as T-gap \cite{Georgiou&Smith90} and
adopted in \cite{LD-Automatica-2020} under the concept of $\mathcal{K}$-gap
for fault detection study. For our purpose, define 
\begin{equation}
\mathcal{K}_{G_{i}}\left( \mathcal{H}_{2}^{\bot }\right) =\left\{ \left[ 
\begin{array}{c}
u \\ 
y%
\end{array}%
\right] \in \mathcal{H}_{2}^{\bot }:\left[ 
\begin{array}{cc}
-\hat{N}_{i} & \hat{M}_{i}%
\end{array}%
\right] \left[ 
\begin{array}{c}
u \\ 
y%
\end{array}%
\right] =0\hspace{-2pt}\right\} \in \mathcal{H}_{2}^{\bot },i=1,2,
\label{eq7-2}
\end{equation}%
as a dual subspace to the $\mathcal{H}_{2}$ image subspace $\mathcal{I}%
_{G_{i}}$. It turns out%
\begin{gather*}
\mathcal{P}_{\mathcal{K}_{G_{i}}\left( \mathcal{H}_{2}^{\bot }\right) }:%
\mathcal{H}_{2}^{\bot }\rightarrow \mathcal{H}_{2}^{\bot },\mathcal{P}_{%
\mathcal{K}_{G_{i}}\left( \mathcal{H}_{2}^{\bot }\right) }=\mathcal{I}-%
\mathcal{L}_{K_{G_{i}}^{\sim }}\mathcal{P}_{\mathcal{H}_{2}^{\bot }}\mathcal{%
L}_{K_{G_{i}}},i=1,2, \\
\vec{\delta}\left( \mathcal{K}_{G_{1}}\left( \mathcal{H}_{2}^{\bot }\right) ,%
\mathcal{K}_{G_{2}}\left( \mathcal{H}_{2}^{\bot }\right) \right) =\left\Vert
\left( \mathcal{I}-\mathcal{P}_{\mathcal{K}_{G_{1}}\left( \mathcal{H}%
_{2}^{\bot }\right) }\right) \mathcal{P}_{\mathcal{K}_{G_{2}}\left( \mathcal{%
H}_{2}^{\bot }\right) }\right\Vert .
\end{gather*}%
Note that%
\begin{gather*}
\mathcal{I}-\mathcal{P}_{\mathcal{K}_{G_{1}}\left( \mathcal{H}_{2}^{\bot
}\right) }=\mathcal{L}_{K_{G_{1}}^{\sim }}\mathcal{P}_{\mathcal{H}_{2}^{\bot
}}\mathcal{L}_{K_{G_{1}}}:\mathcal{H}_{2}^{\bot }\rightarrow \mathcal{H}%
_{2}^{\bot }, \\
\mathcal{P}_{\mathcal{K}_{G_{2}}\left( \mathcal{H}_{2}^{\bot }\right) }=%
\mathcal{I}-\mathcal{L}_{K_{G_{2}}^{\sim }}\mathcal{P}_{\mathcal{H}%
_{2}^{\bot }}\mathcal{L}_{K_{G_{2}}}=\mathcal{L}_{I_{G_{2}}}\mathcal{L}%
_{I_{G_{2}}^{\sim }}+\mathcal{L}_{K_{G_{2}}^{\sim }}\mathcal{P}_{\mathcal{H}%
_{2}}\mathcal{L}_{K_{G_{2}}}:\mathcal{H}_{2}^{\bot }\rightarrow \mathcal{H}%
_{2}^{\bot }, \\
\Longrightarrow \left( \mathcal{I}-\mathcal{P}_{\mathcal{K}_{G_{1}}\left( 
\mathcal{H}_{2}^{\bot }\right) }\right) \mathcal{P}_{\mathcal{K}%
_{G_{2}}\left( \mathcal{H}_{2}^{\bot }\right) }=\mathcal{L}_{K_{G_{1}}^{\sim
}}\left[ 
\begin{array}{cc}
\mathcal{P}_{\mathcal{H}_{2}^{\bot }}\mathcal{L}_{K_{G_{1}}}\mathcal{L}%
_{I_{G_{2}}}\mathcal{P}_{\mathcal{H}_{2}^{\bot }} & \mathcal{P}_{\mathcal{H}%
_{2}^{\bot }}\mathcal{L}_{K_{G_{1}}}\mathcal{L}_{K_{G_{2}}^{\sim }}\mathcal{P%
}_{\mathcal{H}_{2}}%
\end{array}%
\right] \left[ 
\begin{array}{c}
\mathcal{L}_{I_{G_{2}}^{\sim }} \\ 
\mathcal{L}_{K_{G_{2}}}%
\end{array}%
\right] ,
\end{gather*}%
which yields%
\begin{equation}
\vec{\delta}\left( \mathcal{K}_{G_{1}}\left( \mathcal{H}_{2}^{\bot }\right) ,%
\mathcal{K}_{G_{2}}\left( \mathcal{H}_{2}^{\bot }\right) \right) =\left\Vert %
\left[ 
\begin{array}{cc}
\mathcal{P}_{\mathcal{H}_{2}^{\bot }}\mathcal{L}_{K_{G_{1}}}\mathcal{L}%
_{I_{G_{2}}}\mathcal{P}_{\mathcal{H}_{2}^{\bot }} & \mathcal{P}_{\mathcal{H}%
_{2}^{\bot }}\mathcal{L}_{K_{G_{1}}}\mathcal{L}_{K_{G_{2}}^{\sim }}\mathcal{P%
}_{\mathcal{H}_{2}}%
\end{array}%
\right] \right\Vert .  \label{eq7-3}
\end{equation}%
In \cite{Georgiou&Smith90}, it has been proved that the operator norm on the
right-hand side of (\ref{eq7-3}) is equal to the solution of MMP (\ref{eq7-1}%
), which implies%
\begin{equation*}
\vec{\delta}\left( \mathcal{K}_{G_{1}}\left( \mathcal{H}_{2}^{\bot }\right) ,%
\mathcal{K}_{G_{2}}\left( \mathcal{H}_{2}^{\bot }\right) \right) =\inf_{Q\in 
\mathcal{H}_{\infty }}\left\Vert \left[ 
\begin{array}{cc}
-\hat{N}_{1} & \hat{M}_{1}%
\end{array}%
\right] -Q\left[ 
\begin{array}{cc}
-\hat{N}_{2} & \hat{M}_{2}%
\end{array}%
\right] \right\Vert _{\infty }.
\end{equation*}%
This result corrects the claim in \cite{LD-Automatica-2020} that MMP (\ref%
{eq7-1}) is the gap from $\mathcal{K}_{G_{1}}$ to $\mathcal{K}_{G_{2}}$.

\bigskip

\bigskip

\textbf{Appendix A: Proof of Lemma 1} \bigskip

By means of the plant model (\ref{eq4-5}), control law (\ref{eq4-2a})-(\ref%
{eq4-2b}) and the Bezout identity (\ref{eq2-5a}), we have%
\begin{gather*}
\left[ 
\begin{array}{cc}
I & -K \\ 
-G & I%
\end{array}%
\right] ^{-1}\left[ 
\begin{array}{c}
I \\ 
0%
\end{array}%
\right] =\left[ 
\begin{array}{cc}
M & 0 \\ 
0 & V_{0}%
\end{array}%
\right] \left[ 
\begin{array}{cc}
M & -U_{0} \\ 
-N & V_{0}%
\end{array}%
\right] ^{-1}\left[ 
\begin{array}{c}
I \\ 
0%
\end{array}%
\right] \\
=\left[ 
\begin{array}{cc}
M & 0 \\ 
0 & -V_{0}%
\end{array}%
\right] \left( \left[ 
\begin{array}{cc}
M_{0} & U_{0} \\ 
N_{0} & V_{0}%
\end{array}%
\right] +\left[ 
\begin{array}{c}
\Delta _{M} \\ 
\Delta _{N}%
\end{array}%
\right] \left[ 
\begin{array}{cc}
I & \text{ }0%
\end{array}%
\right] \right) ^{-1}\left[ 
\begin{array}{c}
I \\ 
0%
\end{array}%
\right] \\
=\left[ 
\begin{array}{cc}
M & 0 \\ 
0 & -V_{0}%
\end{array}%
\right] \left[ 
\begin{array}{cc}
I+\Delta _{1} & 0 \\ 
\Delta _{2} & I%
\end{array}%
\right] ^{-1}\left[ 
\begin{array}{c}
\hat{V}_{0} \\ 
-\hat{N}_{0}%
\end{array}%
\right] \\
=\left[ 
\begin{array}{cc}
M & 0 \\ 
0 & -V_{0}%
\end{array}%
\right] \left[ 
\begin{array}{cc}
\left( I+\Delta _{1}\right) ^{-1} & 0 \\ 
-\Delta _{2}\left( I+\Delta _{1}\right) ^{-1} & I%
\end{array}%
\right] \left[ 
\begin{array}{c}
\hat{V}_{0} \\ 
-\hat{N}_{0}%
\end{array}%
\right] , \\
\Delta _{1}=\left[ 
\begin{array}{cc}
\hat{V}_{0} & -\hat{U}_{0}%
\end{array}%
\right] \left[ 
\begin{array}{c}
\Delta _{M} \\ 
\Delta _{N}%
\end{array}%
\right] ,\Delta _{2}=\left[ 
\begin{array}{cc}
-\hat{N}_{0} & \text{ }\hat{M}_{0}%
\end{array}%
\right] \left[ 
\begin{array}{c}
\Delta _{M} \\ 
\Delta _{N}%
\end{array}%
\right] .
\end{gather*}%
Since%
\begin{equation*}
V_{0}\hat{N}_{0}=N_{0}\hat{V}_{0},V_{0}\hat{M}_{0}=I+N_{0}\hat{U}_{0}
\end{equation*}%
leads to 
\begin{gather*}
V_{0}\Delta _{2}\left( I+\Delta _{1}\right) ^{-1}\hat{V}_{0}+V_{0}\hat{N}_{0}
\\
=N_{0}\left( I-\Delta _{1}\left( I+\Delta _{1}\right) ^{-1}\right) \hat{V}%
_{0}+\Delta _{N}\left( I+\Delta _{1}\right) ^{-1}\hat{V}_{0}=N\left(
I+\Delta _{1}\right) ^{-1}\hat{V}_{0},
\end{gather*}%
it turns out%
\begin{gather*}
\left[ 
\begin{array}{cc}
M & 0 \\ 
0 & -V_{0}%
\end{array}%
\right] \left[ 
\begin{array}{cc}
\left( I+\Delta _{1}\right) ^{-1} & 0 \\ 
-\Delta _{2}\left( I+\Delta _{1}\right) ^{-1} & I%
\end{array}%
\right] \left[ 
\begin{array}{c}
\hat{V}_{0} \\ 
-\hat{N}_{0}%
\end{array}%
\right] =\left[ 
\begin{array}{c}
M \\ 
N%
\end{array}%
\right] \left( I+\Delta _{1}\right) ^{-1}\hat{V}_{0} \\
\Longrightarrow \left[ 
\begin{array}{c}
u \\ 
y%
\end{array}%
\right] =\left[ 
\begin{array}{cc}
I & -K \\ 
-G & I%
\end{array}%
\right] ^{-1}\left[ 
\begin{array}{c}
I \\ 
0%
\end{array}%
\right] v=\left[ 
\begin{array}{c}
M \\ 
N%
\end{array}%
\right] \left( I+\Delta _{1}\right) ^{-1}\hat{v}.
\end{gather*}%
Thus, the lemma is proved. %

\end{document}